\documentclass[11pt]{article}

\usepackage{fullpage} 
\usepackage{comment}
\usepackage{amsthm}
\usepackage{tikz}
\usetikzlibrary{positioning,calc,arrows.meta}
\usepackage{pgfplots}
\pgfplotsset{compat=1.18}
\usetikzlibrary{patterns}
\usepgfplotslibrary{fillbetween}
\usetikzlibrary{intersections}
\usepackage{pgfplots}

\usepackage{csquotes}

\usepackage[dvipsnames]{xcolor}

\usepackage{bigints}
\usepackage{amsmath,amssymb}
\usepackage{xcolor}
\usepackage{tcolorbox}

\usepackage{thm-restate}

\usepackage{natbib}
\bibpunct{(}{)}{;}{a}{,}{,}
 
\usepackage{hyperref}

\usepackage{comment}
\newcommand{\xhdr}[1]{\vspace{2mm}\noindent{\bf {#1}\ }}

\usepackage{url}            \usepackage{booktabs}       \usepackage{amsfonts}       \usepackage{nicefrac}       \usepackage{microtype}      \usepackage{dsfont}

\usepackage{mathtools, bm}

\usepackage{algorithm}
\usepackage{algpseudocode}

\usepackage{subcaption}
\usepackage{multirow}
\usepackage{float}
\usepackage{xspace}

\usepackage{enumitem}

\usepackage{xfrac}

\usepackage{cleveref}
\crefname{enumi}{part}{parts}

\newcommand{\colorfootnote}[1]{\textcolor{blue}{\footnote{#1}}}
\theoremstyle{plain}
\newtheorem{theorem}{Theorem}[section]
\newtheorem{lemma}[theorem]{Lemma}

\newtheorem{corollary}[theorem]{Corollary}

\theoremstyle{plain}
\newtheorem{definition}{Definition}[section] \newtheorem{example}[definition]{Example}
\newtheorem{assumption}[definition]{Assumption} 
\newtheorem{remark}[definition]{Remark}

\usepackage{color-edits}
\addauthor[Yiding]{yf}{purple}    \addauthor[Mengfan]{mf}{blue}
\addauthor[Yaohao]{yh}{brown}
\addauthor[Tian]{tb}{Maroon}

\allowdisplaybreaks

\usepackage{etoolbox}
\preto\part{\setcounter{section}{0}}

\newcommand{\supp}[1]{\mathtt{supp}(#1)}
\newcommand{\pert}{\delta}
\newcommand{\temp}{z}
\newcommand{\tempfunc}{Z}

\newcommand{\hr}{h}
\newcommand{\HR}{H}

\newcommand{\ratiofb}{\mathrm{PoA}}
\newcommand{\ratiosb}{\mathrm{PoDM}}

\newcommand{\virval}{\phi}

\newcommand{\mech}{\mathcal{M}}
\newcommand{\allocrule}{\boldsymbol{\alloc}}
\newcommand{\paymentrule}{\boldsymbol{\payment}}
\newcommand{\alloc}{x}
\newcommand{\payment}{p}
\newcommand{\utip}{U} \newcommand{\utipfb}{\utip_{\textsc{fb}}}
\newcommand{\utipsb}{\utip_{\textsc{sb}}}
\newcommand{\utippost}{\utip_{\mathrm{posted}}}
   \newcommand{\agents}{[n]}
\newcommand{\agentnum}{n}
\newcommand{\tempnum}{k}

\newcommand{\costprf}{\boldsymbol{\cost}}

\newcommand{\outcome}{\lambda}
\newcommand{\contract}{\alpha}

\newcommand{\rewardfun}{r} \newcommand{\quantile}{q}
\newcommand{\monoquan}{q^{*}}
\newcommand{\revenuecurve}{R}

\newcommand{\ratiooa}{\eta}
\newcommand{\ratioma}{\tau}

\newcommand{\succpr}{\rho}

\newcommand{\reserve}{\theta}

\newcommand{\price}{p}
\newcommand{\poa}{PoA\xspace}
\newcommand{\podm}{PoDM\xspace}

\newcommand{\cost}{c}
\newcommand{\costdis}{G}
\newcommand{\costdisprf}{\boldsymbol{\costdis}}
\newcommand{\costdens}{\costdis'}
\newcommand{\costfun}{\upsilon}
\newcommand{\costdisfun}{\varphi}

\newcommand{\val}{v}
\newcommand{\valdis}{F}
\newcommand{\valdisprf}{\boldsymbol{\valdis}}
\newcommand{\valdens}{\valdis'}
\newcommand{\valreserve}{\theta}

\newcommand{\valprf}{\boldsymbol{\val}}
\newcommand{\valmax}{\val_{\max}}

\newcommand{\welfare}{w}
\newcommand{\welfareprf}{\boldsymbol{\welfare}}

\newcommand{\welfaremax}{\welfare_{\max}}
\newcommand{\welfaredis}{F}
\newcommand{\welfaredisprf}{\boldsymbol{\welfaredis}}
\newcommand{\welfaredens}{\welfaredis'}

 \newcommand{\cwelfare}{contracted contribution\xspace}
\newcommand{\cwelfares}{contracted contributions\xspace}

\newcommand{\cwf}{\welfare}
\newcommand{\cwfprf}{\boldsymbol{\welfare}}
\newcommand{\contribution}{contribution\xspace}
\newcommand{\contributions}{contributions\xspace}
\newcommand{\Contribution}{Contribution\xspace}

\newcommand{\upsupp}{\bar{w}}
\newcommand{\lowsupp}{\underline{w}}
\newcommand{\ratiolu}{\kappa}
\newcommand{\harmonic}{\mathcal{H}}

\newcommand{\zjs}{intermediaries\xspace}
\newcommand{\zj}{intermediary\xspace}
\newcommand{\Zj}{Intermediary\xspace}
\newcommand{\Zjs}{Intermediaries\xspace}

\newcommand{\lownum}{s}
\newcommand{\upnum}{\ell}

\newcommand{\ctrctone}{\boldsymbol{1}}
\newcommand{\ctrct}{t}
\newcommand{\ctrctprf}{\boldsymbol{\ctrct}}
\newcommand{\ctrctlnr}{\contract}

\newcommand{\ctrctlnropt}{\ctrctlnr^{\star}}
\newcommand{\ctrctprfvar}{\ctrctprf^{\prime}}
\newcommand{\ctrctprfopt}{\ctrctprf^{\star}}

\newcommand{\otcnum}{m}
\newcommand{\otc}{\lambda}
\newcommand{\otcs}{[\otcnum]}
\newcommand{\otcdis}{\rho}
\newcommand{\otcdisprf}{\boldsymbol{\otcdis}}

\newcommand{\inner}[2]{\left\langle #1, #2 \right\rangle}
\newcommand{\exrwd}{\rewardfun}
\newcommand{\rwd}{\gamma}
\newcommand{\rwdprf}{\boldsymbol{\rwd}}

\newcommand{\ommhrupbound}{3.033}
\newcommand{\apmhrupbound}{3.448}

\newcommand{\set}[1]{\left\{#1\right\}}
\newcommand{\parent}[1]{\left(#1\right)}
\newcommand{\parents}[1]{\left[#1\right]}

\newcommand{\suchthat}{\,:\,}
\newcommand{\deq}{\triangleq}

\DeclareMathOperator*{\argmax}{arg\,max}

\newcommand{\reals}{\mathbb{R}}
\newcommand{\nnreals}{\mathbb{R}_+}
\newcommand{\naturals}{\mathbb{N}}

\DeclareMathOperator*{\prb}{\mathrm{Pr}}
\newcommand{\prob}[2][]{\prb\ifthenelse{\not\equal{}{#1}}{\nolimits_{#1}}{}\!\left[{\def\givenn{\middle|}#2}\right]}
\newcommand{\expect}[2][]{\mathbb{E}\ifthenelse{\not\equal{}{#1}}{_{#1}}{}\!\left[{\def\givenn{\middle|}#2}\right]}

\newcommand{\tparen}{\big}
\newcommand{\tprob}[2][]{\text{Pr}\ifthenelse{\not\equal{}{#1}}{_{#1}}{}\tparen[{\def\given{\tparen|}#2}\tparen]}
\newcommand{\texpect}[2][]{\mathbb{E}\ifthenelse{\not\equal{}{#1}}{_{#1}}{}\tparen[{\def\given{\tparen|}#2}\tparen]}

\newcommand{\sprob}[2][]{\text{Pr}\ifthenelse{\not\equal{}{#1}}{_{#1}}{}[#2]}
\newcommand{\sexpect}[2][]{\mathbb{E}\ifthenelse{\not\equal{}{#1}}{_{#1}}{}[#2]}

\newcommand{\dd}{{\mathrm d}}
\newcommand{\ee}{{\mathrm e}}
\newcommand{\DDX}[2][]{{\frac{\dd #1}{\dd #2}}}

\newcommand{\bigO}[2]{\mathcal{O}_{#1}\parent{#2}}
\newcommand{\smallO}[2]{o_{#1}\parent{#2}}
\newcommand{\bigtheta}[2]{\Theta_{#1}\parent{#2}}
\newcommand{\bigomega}[2]{\Omega_{#1}\parent{#2}}

\newcommand{\bO}[1]{\mathcal{O}\parent{#1}}

\newcommand{\ceil}[1]{\lceil #1 \rceil}
\newcommand{\floor}[1]{\lfloor #1 \rfloor}
\newcommand{\bceil}[1]{\left\lceil #1 \right\rceil}

\newcommand{\plus}[1]{{\left( #1 \right)^+}}

\title{Contracting with a Mechanism Designer}

\author{
Tian Bai\thanks{University of Bergen. Email: {\tt tian.bai@uib.no}}
\and
Yiding Feng\thanks{Hong Kong University of Science and Technology. Email: {\tt ydfeng@ust.hk}}
\and
Yaohao Liu\thanks{University of Electronic Science and Technology of China. Email: {\tt yaohao.liu594@gmail.com}}
\and
Mengfan Ma\thanks{Central China Normal University. Email: {\tt mengfanma1@gmail.com}}
\and
Mingyu Xiao\thanks{University of Electronic Science and Technology of China. Email: {\tt myxiao@uestc.edu.cn}}
}
\date{}

\begin{document}

\maketitle
\begin{abstract}
    This paper explores the economic interactions within modern crowdsourcing markets. In these markets, employers issue requests for tasks, platforms facilitate the recruitment of crowd workers, and workers complete tasks for monetary rewards. Recognizing that these roles serve distinct functions within the ecosystem, we introduce a three-party model that distinguishes among the principal (the requester), the intermediary (the platform), and the pool of agents (the workers). The principal, unable to directly engage with agents, relies on the intermediary to recruit and incentivize them. This interaction unfolds in two stages: first, the principal designs a profit-sharing contract with the intermediary; second, the intermediary implements a mechanism to select an agent to complete the delegated task.

We analyze the proposed model as an extensive-form Stackelberg game. Our contributions are threefold:  
First, we fully characterize the subgame perfect equilibrium of our model.
In particular, the principal's contract design problem can be represented as \emph{virtual value pricing}, a novel auction-theoretic formulation. We identify the \emph{optimality of linear contracts}, even when the task has multiple outcomes and agents' cost distributions are asymmetric.
Second, to quantify the principal's utility loss from delegation and information asymmetry, we introduce the \emph{price of double marginalization} (PoDM) and the classical \emph{price of anarchy} (PoA). We derive tight or nearly tight bounds on both ratios under regular and monotone hazard rate distributions.  
Finally, we extend our analysis to two natural variants of the base model: (i) the intermediary is restricted to anonymous pricing mechanisms, and (ii) the principal lacks precise information about the market size.

 \end{abstract}

\thispagestyle{empty}
\newpage
\setcounter{page}{1}

\section{Introduction}
\label{sec:intro}
Over the past two decades, crowdsourcing has become an integral part of the modern economy \citep{How-06}, enabling companies and institutions to connect with a globally distributed workforce through an open call.
According to a report by \citet{link-dataintelo}, the global crowdsourcing market was valued at USD 2.8 billion in 2023 and is projected to reach USD 7.5 billion by 2032, growing at a compound annual growth rate of 11\% during the forecast period. 
This growth is primarily driven by the increasing adoption of social and digital platforms, which enable more efficient and large-scale crowdsourcing efforts.

Prominent platforms such as Freelancer, 99designs, and Upwork demonstrate how crowdsourcing has expanded across various industries, including product development, marketing, data collection, innovation, and others. 
One of the most attractive features of these platforms is their ability to provide flexible and quick access to a diverse talent pool, transcending geographical barriers to enable seamless collaboration. 
This capacity is particularly appealing to employers who (i) face high costs when hiring local workers, (ii) have difficulty finding skilled professionals, or (iii) require the flexibility to scale projects up or down based on demand, including one-time hires.
Acting as intermediaries, these platforms allow employers to function as task requesters while relieving them of project management duties, facilitating efficient collaboration with crowd workers worldwide.

To better understand the dynamics of crowdsourcing markets, it is essential to identify three primary parties: employers (also known as requesters)\colorfootnote{For instance, in \href{https://requester.mturk.com/}{Amazon Mechanical Turk}, employers are referred to as ``requesters''.}, platforms, and crowd workers.
The general workflow operates as follows: employers issue requests for tasks or projects, platforms facilitate the recruitment of crowd workers, and crowd workers complete the assigned tasks.
To motivate workers, platforms typically use monetary incentives to compensate for their efforts.\colorfootnote{There are crowdsourcing systems that offer workers non-monetary incentives, such as altruistic fulfillment \citep{MR-22,SLS-22}, access to information \citep{JCP-09,PBS-18,GJ-21}, and entertainment \citep{VD-08,DVARPKM-24}.
However, in this paper, we focus on monetary incentives and emphasize that the overwhelming majority of crowdsourcing tasks are completed for financial compensation.}
In most cases, however, the types of crowd workers remain private, motivating platforms to adopt mechanisms that extract this information while maximizing the platforms' profit.
In these mechanisms, crowd workers are incentivized to report their types (costs for task completion), after which the platforms allocate tasks and determine payments based on the revealed types.
This platform-worker interaction, modeled as a mechanism, is well understood and has been extensively studied in the literature \citep{Mye-81, Sin-10, SM-13, SW-15, CHS-19, HWHC-23, HAKA_2023, NPS-24}.
However, most existing crowdsourcing models focus solely on this mechanism-based interaction between platforms and workers, overlooking the employers' role in the three-party ecosystem.
In reality, employers play a crucial role in the overall system, as they are the ones who provide payment to the platform, which then compensates the workers and retains the difference between the two payments as its profit.
Unlike employers, platforms have no intrinsic interest in the task itself but are primarily concerned with profit maximization. 
Thus, the way employers compensate the platform directly influences how platforms design their mechanisms.
We provide two concrete motivating examples.

\begin{example}
\href{https://www.freelancer.com,}{Freelancer}
    is a leading online crowdsourcing platform that connects clients with freelancers across various industries.
    As one of the largest freelance marketplaces globally, it has a network of over 78 million professionals from 247 countries, specializing in fields such as web development, data collection, writing, and marketing.
    Clients post job listings with requirements and budgets, while the platform ensures compliance with payment terms and project expectations.
    The platform selects freelancers through bidding, algorithmic recommendations, or direct hiring.
    Upon completion of the task, clients submit payment to the platform, which takes a service fee before transferring the remainder to freelancers.
\end{example}

\begin{example}
\label[example]{eg:99designs}
    The online freelancer platform, \href{https://www.99designs.com}{99designs}, connects clients with a global community of professional graphic designers, specializing in custom designs for logos, branding, websites, and packaging.
    As one of the largest design marketplaces, it has facilitated over 1 million creative projects, with a network of designers from more than 100 countries.
    Clients initiate a design contest or project on the platform by providing a design brief and offering a reward. 
    The platform recommends suitable designers based on the client's needs and the designers' expertise.
    Designers submit proposals in order to be selected for the projects.
    Upon selection, the platform collects payment from the client, subtracts its commission, and transfers the remaining amount to the winning designer.
\end{example}

Motivated by the applications discussed above, we introduce a \emph{principal-\zj-agents} model that captures the economic interactions among the employer (the principal), the platform (the \zj), and the crowd workers (the agents) in crowdsourcing markets. 
In this model, the employer interacts with the platform through a \emph{contract}.
Under this contract, the employer delegates a task to the platform and commits to paying a service fee contingent on the task's realized outcome, which represents the quality level of task completion.
To fulfill this contract, the platform must recruit a worker from a pool of available workers to execute the task.
We model this recruitment process as a Bayesian mechanism, where the workers' private types correspond to their costs of completing the task and are assumed to be independently distributed. 
Each worker, once recruited, completes the task and induces a probability distribution over outcomes, each providing a fixed reward to the principal.
We assume these outcome distributions are publicly known to both the principal and the \zj (see \Cref{subsec:model discussion} for further details on this assumption). 
A graphical illustration of our model can be found in \Cref{fig:model illustration}.

To justify modeling the employer-platform interaction as a contract, we note that the mechanisms employed by online platforms are typically not disclosed. Even when some details are available, they often consist only of high-level descriptions---such as the auction format---while critical parameters, like reserve prices or ranking algorithms, remain hidden.\colorfootnote{For example, the reserve price on sites like \href{https://www.ebay.com/help/buying/bidding/reserve-prices-work?id=4018}{eBay} is hidden, and until the reserve is met, the system displays ``Reserve Not Met.''}
Consequently, the employer observes only the realized outcome of these mechanisms, namely the quality level at which the task is completed.
This setting aligns with the classic principal–agent model in contract theory \citep{Hol-79, GH-83}, where: (i) the employer acts as the principal and the platform as the agent; (ii) the platform's choice of mechanism constitutes the agent's hidden action; and (iii) the payment associated with each mechanism corresponds to the action's cost.
Accordingly, the platform's use of unrevealed mechanisms reflects the hidden-action assumption of the classical model.
We thus refer to our modeling paradigm as \emph{contracting with a mechanism designer}.

Beyond crowdsourcing applications, our model naturally aligns with the broader economic framework of \emph{double marginalization} \citep{Spe-50}, which describes inefficiencies resulting from sequential pricing decisions made by two monopolistic entities within a supply chain, such as a manufacturer and a retailer.
In our context, the employer acts analogously to the manufacturer by setting a contract (comparable to wholesale pricing), while the platform acts as the retailer by subsequently selecting mechanisms (comparable to retail pricing) to engage workers.
As in classic double marginalization scenarios, this decentralized decision-making by both parties can generate inefficiencies not present under an integrated decision-maker.

Our proposed three-party model raises several key questions:
\begin{quote}
    \emph{How can the principal design a contract with the \zj when the exact mechanism employed by the \zj is unknown?
    How should the \zj select her mechanism to recruit agents profitably given the contract offered by the principal? 
    What impact does the \zj's presence have on the principal's utility?}
\end{quote}
The primary goal of this paper is to address these questions by 
(i) characterizing the optimal contract and mechanism for the principal and \zj, respectively, 
and (ii) providing quantitative insights into the impact of the \zj and the market size (i.e., number of agents).

\subsection{Our Contributions and Techniques} \label[section]{sec: contributions}

In this paper, we propose a game-theoretical model that captures the key features of crowdsourcing markets, including the principal-\zj-agents interaction, the \zj's hidden mechanisms, and workers' private costs. In our extensive-form Stackelberg game, the principal offers a contract, the platform selects a mechanism, and workers submit bids to the mechanism.
We first characterize the subgame perfect equilibrium of the game.
In particular, we provide a closed-form characterization for the \zj's optimal mechanism. 
This, in turn, enables a characterization of the principal's optimal contract design.
Second, building on the equilibrium characterization, we derive tight or nearly tight bounds on the principal's utility compared to benchmark scenarios where the principal can directly access agents without the \zj. 
Finally, we extend the analysis to two natural variants: (i) the \zj is restricted to anonymous pricing mechanisms, and (ii) the principal does not have precise information about the market size.

\xhdr{Optimal contract characterization.}
The principal's goal is to determine the optimal contract that maximizes her expected utility.
To derive the subgame perfect equilibrium of our Stackelberg game, we use backward induction.
We first characterize the optimal mechanism that maximizes the \zj's utility for any given contract.
Leveraging this optimal mechanism characterization, we subsequently characterize the contract design problem under three scenarios in \Cref{thm:opt contract}, which range from general to more specialized cases: (i) agents have asymmetric outcome distributions and cost distributions; (ii) agents have identical outcome distributions and asymmetric cost distributions; and (iii) agents have identical outcome distributions and cost distributions.

Optimal contracts are generally complex, and linear contracts are often suboptimal---achieving optimality only in highly restricted settings \citep[e.g.,][]{HM-87,Car-15,DRT-19}. 
In contrast, our results establish the optimality of linear contracts even in settings with multiple outcomes, possibly asymmetric cost distributions, but identical outcome distributions across agents (i.e., scenarios (ii) and (iii)).
This substantially simplifies the contract design space.
Intuitively, when all agents share the same outcome distribution, the principal's payoff depends only on whether an agent is selected, rather than on the identity of the selected agent.
The optimal mechanism, which determines the selection probability, depends only on the expected payment from the contract to the \zj.
Therefore, if there exists an optimal (possibly nonlinear) contract that induces a certain expected payment, we can construct a linear contract with the same expected payment.
By the above reasoning, such a linear contract must also be optimal.

Interestingly, when the cost distributions are also identical (i.e., scenario (iii)), our characterization reveals that the optimal contract design problem can be reformulated as a novel \emph{virtual value pricing problem} in a single-item auction: the seller offers buyers a take-it-or-leave-it price for an item but her received payment is equal to the virtual value of the price rather than the price itself if the item is sold (see \Cref{sec:myerson revisited} for the definition of virtual value). 
Intuitively, the fact that the principal receives the virtual value of the price, rather than the price itself, stems directly from the double marginalization effect.
(The virtual value of a price is always less than or equal to the price under standard regularity conditions, implying a revenue loss due to the presence of the \zj.)
In \Cref{thm:opt contract}, we show that the seller's optimal revenue under the virtual value pricing equals the principal's optimal utility under a contract.
Additionally, there exists a one-to-one correspondence between the seller's price and the principal's contract.  
This virtual value pricing formulation establishes a close connection between contract design in our model and pricing in the classic single-item auction and may be of independent interest. See \Cref{example:running} for an illustration.

As discussed above, the characterization of the optimal contract relies on the analysis of the optimal downstream mechanism implemented by the {\zj}. 
We show that the \zj's mechanism design problem is equivalent to a single-item auction, where the \zj and agents play the roles of the seller and buyers, respectively, and the task is the auctioned item. 
Via a bijection between agents' costs and buyers' values (\Cref{lemma:convert mechanism}), the \zj's utility coincides with the seller's revenue, allowing us to apply standard auction theory. 
In particular, the \zj's optimal mechanism is the \emph{Bayesian revenue-optimal mechanism} \citep{Mye-81} parameterized by the contract offered by the principal (\Cref{thm:opt mech}).
Essentially, we establish a one-to-one correspondence between instances of our mechanism problem and single-item auctions.
We illustrate the principal's optimal contract and \zj's optimal downstream mechanism in the following example:

\begin{example} \label{example:running}
Consider a task with a binary outcome (i.e., success and failure).
The reward is $4$ under success and $0$ under failure.
There are two agents with identical success probability $0.5$ and i.i.d.\ uniform cost distributions supported on $[0, 1]$.

First, suppose the principal and the \zj share the reward equally (50-50 split). Specifically, if the task is successfully completed, the \zj receives $4\times 0.5 = 2$.
The optimal mechanism selected by the \zj (given this 50-50 split contract) works as follows: if both agents' costs exceed $0.5$, neither agent is assigned the task; otherwise, the task is assigned to the agent with the lower cost and she is paid---regardless of the task outcome---the minimum value between the reserve $0.5$ and the other agent's cost.

However, the 50-50 split contract is not optimal for the principal. Suppose the principal chooses the contract that shares $\contract$ fraction of the reward to the \zj. The principal's expected utility (under the \zj's subsequent optimal mechanism) is 
\begin{align*}
    2(1 - \contract)(1-(1 - \contract)^{2}),
\end{align*}
which is optimized at $\contract^{\star} = \frac{3 - \sqrt{3}}{3} \approx 0.423$ yielding an optimal payoff of $\frac{4\sqrt{3}}{9} \approx 0.77$. 
By contrast, the 50-50 split contract ($\contract=0.5$) only yields a payoff of $0.75$.
    
The principal's contract design can be viewed as the following virtual value pricing problem: the contract (reward share) $\contract$ corresponds to posting a price $\price = 2 - \contract$ to two buyers with values independently and uniformly distributed over $[1, 2]$.
A sale occurs only if at least one buyer's value exceeds the price $\price$.
Under the uniform distribution over $[1, 2]$, the virtual value of price $\price$ is $2\price - 2$ and the probability of the item being sold is $1 - (\price - 1)^{2}$.
The corresponding optimal price is $\price^{\star} = 2 - \contract^{\star} = \frac{3 + \sqrt{3}}{3} \approx 1.577$.
\end{example}

\xhdr{Impact of the \zj.} 
Understanding the \zj's impact on the principal's utility is crucial for both developing theoretical insights and informing practical decisions in the crowdsourcing markets.
If the principal's utility when working with the \zj nearly matches that achieved through direct access to workers, then employing the \zj is justified as a cost-effective alternative; conversely, a substantial gap would indicate that the platform's fee is prohibitively high, potentially deterring the principal from employing the \zj.
To quantify the efficiency impact of the \zj, we introduce two performance measures. The \emph{price of double marginalization} (\podm) is defined as the ratio of the principal's utility when she directly engages with agents through a dominant strategy incentive compatible (DSIC) mechanism to her utility in our model with an \zj.
This captures the utility loss due solely to delegation. The \emph{price of anarchy} (\poa) compares the principal's utility in the ideal full-information setting---where agents' private costs are observable and allocations can be optimized directly---with her utility in our model, thereby quantifying the total efficiency loss, including that due to information asymmetry.
For clarity, we refer to the utility under the DSIC mechanism as the \emph{second-best benchmark}, and the utility under the full-information setting as the \emph{first-best benchmark}.\colorfootnote{Our definition of the second-best benchmark slightly deviates from the standard in the auction theory literature \citep{MS-83,DMSW-22}.}

In our analysis, we study these two ratios under the assumption that agents have identical success probabilities and cost distributions.
Under this assumption, as we have established, the principal's maximum utility in our model equals the optimal revenue achieved under the virtual value pricing.
By the correspondence established in our characterization of the \zj's optimal mechanism, the first-best benchmark corresponds to the maximum welfare in the single-item auction instances, whereas the second-best benchmark corresponds to the maximum revenue in the single-item auction.
Thus, bounding the two ratios in our model is equivalent to bounding their analogues in the auction framework---namely, the ratio of optimal revenue under virtual value pricing to the maximum welfare, and to the maximum revenue under a standard auction, respectively.
This connection allows us to leverage established auction theory techniques---such as the revenue curve and hazard rate function---to bound the two ratios.
Our main results on the two ratios are summarized in \Cref{tab:results}.

\begin{table}[t]
    \centering
    \clearpage{}\renewcommand{\arraystretch}{2}

\begin{tabular}{c|cc|cc|}
\cline{2-5}
                                           & \multicolumn{2}{c|}{Optimal mechanism}                                                                                          & \multicolumn{2}{c|}{Anonymous pricing}                                                                                           \\ \cline{2-5} 
                                           & \multicolumn{1}{c|}{\renewcommand{\arraystretch}{1}\begin{tabular}[c]{@{}c@{}}Regular\\\small [\Cref{thm:regular bounds}]\end{tabular}}                            
                                           & \renewcommand{\arraystretch}{1}\begin{tabular}[c]{@{}c@{}}MHR\\\small [\Cref{thm:mhr ratio bounds for n in naturals,thm:r_fb asymp lb}]\end{tabular}                                                                   
                                           & \multicolumn{1}{c|}{\renewcommand{\arraystretch}{1}\begin{tabular}[c]{@{}c@{}}Regular\\\small [\Cref{thm:regular bounds under ap}]\end{tabular}}                            
                                           & \renewcommand{\arraystretch}{1}\begin{tabular}[c]{@{}c@{}}MHR\\\small [\Cref{thm:mhr bounds under ap}]\end{tabular}                                                                    \\ \hline
\multicolumn{1}{|c|}{\podm} & \multicolumn{1}{c|}{$\bigtheta{}{\log \ratiolu}$}       & $[\ee, \ommhrupbound] \overset{\agentnum\rightarrow\infty}{\longrightarrow}1$ & \multicolumn{1}{c|}{$\bigtheta{}{\log \ratiolu}$}       & $[\ee, \apmhrupbound]$ \\ \hline
\multicolumn{1}{|c|}{\poa}  & \multicolumn{1}{c|}{$\bigtheta{}{(\log \ratiolu)^{2}}$} & $\ee^{2}\overset{\agentnum\rightarrow\infty}{\longrightarrow}1$       & \multicolumn{1}{c|}{$\bigtheta{}{(\log \ratiolu)^{2}}$} & $\ee^{2}$        \\ \hline
\end{tabular}

\clearpage{}
    \caption{Bounds of the price of double marginalization (\podm) and the price of anarchy (\poa). Here $\agentnum$ is the number of agents, $\ratiolu$ is the ratio of 
    the largest to the smallest value
in the support of the \contribution distribution (See \Cref{def:welfare and cwelfare} for the formal definition).}
\label{tab:results}
\end{table}

\addtocounter{definition}{-1}
\begin{example}[cont.]
    The first-best benchmark is $\frac{5}{3}$, achieved by selecting the agent with the lower cost and paying her exactly her cost. 
    The second-best benchmark is $\frac{4}{3}$, achieved by assigning the task as follows: if both agents' costs exceed $1$, neither agent is assigned the task; otherwise, the task is assigned to the agent with the lower cost, and she is paid the maximum value between $2$ minus the other agent's cost and $1$.
    Consequently, the price of anarchy is $\frac{5\sqrt{3}}{4} \approx 2.165$ and the price of double marginalization is $\sqrt{3} \approx 1.732$. 
    Namely, the principal's utility loss from intermediation is at most $42.26\%$ when compared to the scenario where she has direct interaction with agents, and at most $53.81\%$ when compared to the scenario where she can even fully observe agents' cost profile.
\end{example}

First, note that without any assumptions on the cost distributions, both ratios are unbounded.
In auction theory, a common assumption is that the value distribution is regular (see \Cref{sec:myerson revisited} for its definition).
Since there is a bijection between the agents' costs in our problem and the buyers' values in the corresponding auction, we impose the analogous regularity assumption on the cost distribution.
Under this assumption, we show that if the corresponding maximum possible value of the buyers is at most a constant multiple of their minimum possible value, then both ratios are tightly bounded by the logarithm of this constant (see the second column of \Cref{tab:results}).
To establish this result, we utilize the fact that the slope of the revenue curve equals the corresponding virtual value and the revenue curve is concave under the regularity assumption.
See \Cref{thm:regular bounds} for details.

Then, we consider the case where the corresponding value distribution satisfies the monotone hazard rate (MHR) condition, which is a strictly stronger assumption than regularity.
Our first main result shows that under MHR distributions, both ratios are constant (see the third column of \Cref{tab:results}), which is a significant improvement over the results obtained under the regularity assumption.
Our second main result provides a bound on the price of anarchy that depends on the market size and shows that for a fixed distribution, the two ratios asymptotically approach $1$ as the market size tends to infinity.
This implies that the impact of the \zj's participation and the presence of agents' private information on the principal's utility approaches zero as the market expands.
Formal statements are provided in \Cref{thm:mhr ratio bounds for n in naturals,thm:r_fb asymp lb}.

We observe that the \poa is approximately the square of the \podm.
As discussed above, given the established connection between our model and the single-item auction setting, the first-best benchmark in \poa and the second-best benchmark in \podm can be interpreted as the expected welfare and expected revenue, respectively.
Meanwhile, the principal's optimal expected utility corresponds to virtual value pricing, i.e., the product of the virtual value of the price and the purchase probability.
Since the expected revenue equals the expected virtual welfare, \podm effectively measures the ratio between welfare and revenue in a modified instance where the buyer's value distribution is replaced by the virtual value distribution from the original instance.
Notably, this transformation does not preserve key structural properties of the original distribution: for example, a regular (resp., MHR) distribution may induce a virtual value distribution that is not regular (resp., MHR). Nevertheless, our analysis suggests that certain properties (e.g., probability tail bounds) implied by regularity or MHR conditions still hold approximately under this transformation.

\xhdr{\Zj with anonymous pricing.} 
We also consider an extended scenario in which the \zj's mechanism space is restricted to anonymous pricing mechanisms rather than the full set of mechanisms, and the agents have identical outcome distributions and cost distributions.
\colorfootnote{A natural question is whether there is any monotonicity in the principal's utility regarding two \zjs with an inclusion relationship between their mechanism spaces? We give a negative answer to this question in \Cref{apx:diffrent mech space}.}

Our result shows that the principal's problem of designing an optimal contract remains a virtual value pricing involving a single virtual agent, whose value equals the first-order statistic of the real agents' values (see \Cref{thm:prin's utility under ap}).
We then quantify the \zj's impact measured by \podm and \poa as we did in the basic scenario.
Interestingly, we find that in the worst-case analysis, restricting the \zj's mechanism space to pricing mechanisms has nearly the same impact on the principal's utility as allowing unrestricted mechanisms (see the last two columns of \Cref{tab:results}).
Formal statements are provided in \Cref{thm:regular bounds under ap,thm:mhr bounds under ap}.

\begin{table}[t]
    \centering
    \clearpage{}\renewcommand{\arraystretch}{2}
\begin{tabular}{c|cc|}
\cline{2-3} 
                                           & \multicolumn{1}{c|}{\renewcommand{\arraystretch}{1}\begin{tabular}[c]{@{}c@{}}Regular~\small [\Cref{thm:robust regular bounds}]\end{tabular}} 
                                           & \renewcommand{\arraystretch}{1}\begin{tabular}[c]{@{}c@{}}MHR~\small [\Cref{thm:robust MHR bounds}]\end{tabular} \\ \hline
\multicolumn{1}{|c|}{\podm} & \multicolumn{1}{c|}{$\frac{\upnum}{\lownum} \cdot \bO{\log \ratiolu}$}                                                                                                                                            & \multirow{2}{*}{$\bO{\frac{\log \upnum}{\log \lownum}}$}                                                                                                                             \\ \cline{1-2}
\multicolumn{1}{|c|}{\poa}  & \multicolumn{1}{c|}{$\frac{\upnum}{\lownum} \cdot \bO{(\log \ratiolu)^{2}}$}                                                                                                                                      &                                                                                                                                                                                          \\ \hline
\end{tabular}\clearpage{}
    \caption{Bounds of \podm and \poa for the unknown market size setting. Here $\ratiolu$ is the ratio of 
    the largest to the smallest value in the support of the \contribution distribution, and $\upnum$ and $\lownum$ are the largest and smallest values of the range of the number of agents, respectively.}
    \vspace{-0.5cm} 
    \label{tab:results of unknown n}
\end{table}

\xhdr{Contract design for the unknown market size.} 
Considering that in practice it is difficult for principals to know the exact number of agents on the platform, we extend our model to a scenario in which the principal only knows that the market size lies in the interval $[\lownum,\upnum]$.

Our results show that when the market size is unknown, both ratios suffer a multiplicative loss compared to the bounds established for known market sizes.
Under regular distributions, the loss is a factor of $\frac{\upnum}{\lownum}$, while under MHR distributions, it reduces to $\bO{\frac{\log \upnum}{\log \lownum}}$ (see \Cref{tab:results of unknown n}).
The former bound arises from optimizing the contract for any market size in $[\lownum:\upnum]$, while the latter is uniquely attained by optimizing for the minimal size $\lownum$---optimizing for any other size fails to provide this guarantee.
Moreover, we provide an example illustrating that if the principal has no knowledge of the market size---i.e., if the number of agents ranges from 1 to infinity---then, regardless of the contract design, the gap between her utility and these benchmarks cannot be upper-bounded.
Formal statements are provided in \Cref{thm:robust regular bounds,thm:robust MHR bounds,thm:robust negative}.

\subsection{Related Work}  

Our work lies broadly within the fields of contract design and mechanism design, both of which have been extensively studied and are central to the Econ-CS community.
For an overview, we refer interested readers to a recent survey on algorithmic contract design \citep{DFT-24-1} and a textbook on algorithmic mechanism design \citep{NRTV-07}.
More specifically, our work connects to three key areas: market design for crowdsourcing, intermediation in online advertising, and Bayesian auction design.
Below, we outline these connections in detail.

\xhdr{Market design in crowdsourcing.}
Designing effective mechanisms is a key challenge in crowdsourcing systems, where requesters seek to purchase services from crowd workers.
This is often studied as a \emph{procurement} problem, focusing on mechanisms that solicit bids, determine worker selection, and set payments based on reported costs.
\citet{Sin-10}, \citet{CGL-11}, and \citet{NPS-24} explore incentive-compatible mechanisms under \emph{budget-feasible mechanism design}, ensuring utility maximization within a budget, while \citet{SW-15} extend this framework to large-scale crowdsourcing.
\citet{SM-13} examines pricing mechanisms where workers accept tasks at pre-specified rates. 
Additionally, \citet{SSYJ-24} examines how entry restrictions influence effort and revenue in all-pay contests, while \citet{CHS-19} analyzes reward structures in crowdsourcing contests to maximize the principal's utility, providing insights into optimal participation policies and prize allocations.
Additionally, \citet{CHS-19} and \citet{SSYJ-24} model crowdsourcing as \emph{all-pay auctions}, as some platforms require entrants to exert effort upfront to participate (see \Cref{eg:99designs}).
However, posted pricing alone fails to incentivize high-quality work, as workers can exert minimal effort without penalty.
To address this, requesters may use contracts where payments depend on output quality. Since effort levels are unobservable, mechanisms must dynamically adjust contracts as they learn about worker performance.
This motivates studying crowdsourcing under the \emph{repeated principal-agent} framework, where workers arrive online and requesters continuously update quality-contingent payments \citep[e.g.,][]{HSV-14,ZBYWJJ-23,CGR-24}.
Previous work typically models crowdsourcing markets as a two-party interaction.
In contrast, our model distinguishes the requester, platform, and workers as three separate entities and examines all of their economic interactions.

\xhdr{Intermediation in online advertising.} 
In online advertising markets, the \emph{ad exchange} functions as an \zj, as it has exclusive access to advertisers (i.e., buyers) seeking to purchase impressions (i.e., items) from the publisher (i.e., seller).
This role conceptually aligns with the \zj in our model.
\citet{FMMP-10} analyze equilibrium in a setting where advertisers compete for a unique impression from a central seller through intermediaries, with the flexibility to choose which \zj to engage.
\citet{SGP-13, SGP-14} focus on advertisers who are not liquidity-constrained and bid through intermediaries but restrict attention to variants of the Vickrey auction. 
\citet{BC-17} examine the optimal contract an \zj offers to advertisers when their budgets and targeting criteria are private. 
\citet{BCG-21} explore multistage intermediaries in display advertising under second-price rules and characterizes a subgame perfect equilibrium among \zjs and the seller.
\citet{ZC-24} study multi-layer auction markets under first-price auction rules.
As noted in \Cref{sec: contributions}, when agents have identical success probabilities and cost distributions, the optimal contract design problem simplifies to a virtual value pricing problem, closely related to the publisher's revenue optimization in online advertising.
Our results on \podm and \poa are also relevant, as they essentially compare publisher revenues in markets with and without ad exchanges.

\xhdr{Bayesian auction design.}
To bound the price of double marginalization (PoDM) and price of anarchy (PoA) in our model, we express the principal's optimal utility through virtual value pricing, thereby connecting it to the literature on Bayesian auction design, which studies revenue-optimal and efficiency-maximizing mechanisms under Bayesian priors over bidders' valuations.
\citet{Mye-81} introduces virtual values and characterizes the revenue-optimal auction as a virtual welfare maximizer.
\citet{AGM-09} examine the trade-off between revenue and efficiency, showing that optimal mechanisms often introduce inefficiencies.
Several works explore simple yet near-optimal mechanisms: \citet{DRY-15} show how a single value sample can approximate optimal pricing, while \citet{JLQTX-19} analyze the effectiveness of uniform pricing.
Auctions with MHR distributions have been widely studied: \citet{GYZK-18}, \citet{JLQ-19}, and \citet{GPZ-21} evaluate how simple mechanisms perform relative to optimal auctions under MHR assumptions.
\citet{Yan-11} studies the impact of value correlation on auction revenue, while \citet{FJ-24} investigate when simple mechanisms can approximate optimal auctions even in irregular settings.

\xhdr{Double marginalization and efficiency loss.}
Double marginalization arises when upstream and downstream firms with market power set markups independently, leading to inflated prices and reduced welfare.
\citet{Spe-50} shows that vertical integration can eliminate this inefficiency by internalizing markups---a structure conceptually similar to our benchmarks without an \zj.
\citet{RT-86} further demonstrate that vertical contracts, such as two-part tariffs, can replicate the integrated outcome without requiring ownership consolidation.
\citet{PR-07} introduce the \emph{price of anarchy} to quantify efficiency loss in decentralized supply chains under price-only contracts, a setting analogous to our anonymous pricing restriction.
While their metric focuses on total system profit, our measures center on the principal's utility. 
\citet{CL-05} show that revenue-sharing contracts can coordinate decentralized supply chains and flexibly divide surplus, though practical concerns such as administrative burden and incentive distortion may limit their effectiveness.
\citet{BF-05} analyze decentralized supply chains with competing retailers under demand uncertainty and demonstrate that price-discount sharing and buy-back contracts can restore efficiency, yet neither study quantifies inefficiency through formal ratios.
More recent work by \citet{CLV-24} investigates vertical relationships under asymmetric information and strategic foreclosure, focusing on welfare impacts without defining explicit efficiency metrics.
\citet{LM-22} further cautions that eliminating double marginalization through vertical integration may distort bargaining and harm competition.
In contrast to these studies, we define and bound two formal efficiency measures---the \emph{price of anarchy} (\poa) and the \emph{price of double marginalization} (\podm)---to quantify the \zj's impact from the principal's perspective in environments with informational and contractual frictions. 
\section{Preliminaries}
\label{sec:prelim}

In this section, we introduce our three-party model and revisit key concepts for revenue maximization in the Bayesian single-item auctions.
In addition, we outline the assumptions and notations that are used throughout the paper.

\subsection{The Three-Party Model}\label{sec:3p model}
We consider a model involving three parties: a \emph{principal}, an \emph{\zj}, and $\agentnum$ \emph{agents}.
The principal has a task that can be completed by any one of the agents. 
However, the principal cannot directly access the agents.
Instead, she can only interact with an \zj who has access to the agents. 
The interaction unfolds over two periods.
At time $T = 1$, the principal delegates the task to the \zj via a \emph{contract}: a mapping from outcomes to transfer amounts with the promise that if a particular outcome occurs, the principal will transfer the corresponding amount to the \zj.
At time $T = 2$, the \zj recruits an agent via a \emph{mechanism} consisting of an allocation rule that maps the agents' reported costs to the selected agent and a payment rule that specifies how much agents should be paid.
The selected agent (if any) then completes the task at some cost.
Details about the contract and mechanism designed by the principal and \zj are discussed below.

We assume that the agents' completion of the task can result in $\otcnum$ possible outcomes, where each outcome $j \in \otcs$ is associated with a value $\otc_{j} \geq 0$ representing the reward received by the principal if outcome $j$ occurs.
Each agent $i \in \agents$ has a private cost $\cost_{i}$ for completing the task, independently drawn from a publicly known distribution $\costdis_{i}$.
If selected, agent $i$ induces a probability distribution $\otcdisprf_{i} = (\otcdis_{i, 1}, \otcdis_{i, 2}, \ldots, \otcdis_{i, \otcnum}) \in [0, 1]^{\otcnum}$ over the outcome space $\otcs$, where the probability of outcome $j$ is $\otcdis_{i,j}$ and $\sum_{j \in \otcs} \otcdis_{i, j} = 1$.
The expected reward profile \yhedit{generated by} agent $i$ is defined as $\rwdprf_{i} \deq (\otc_{1} \cdot  \otcdis_{i,1}, \otc_{2} \cdot \otcdis_{i,2}, \ldots, \otc_{\otcnum} \cdot \otcdis_{i,\otcnum})$, and the total expected reward generated by agent $i$ is $\exrwd_{i} \deq \sum_{j \in \otcs} \otc_{j} \cdot \otcdis_{i, j}$.
Both the outcome distribution $\otcdisprf_{i}$ and the expected reward $\exrwd_{i}$ are assumed to be public information for all parties.\colorfootnote{All our results hold even when the reward profile $\boldsymbol{\otc} = (\otc_{1}, \otc_{2}, \ldots, \otc_{\otcnum})$ is unobservable to the intermediary. See \Cref{apx:unobservable reward} for more explanation.}

\xhdr{Contract design between the principal and \zj.}
The principal cannot directly observe the actions of the \zj, that is, the mechanism selected by the \zj or the agent subsequently chosen to complete the task.
Instead, the principal observes only the stochastic outcome of the task resulting from this mechanism.\colorfootnote{This setting corresponds to the notion of \emph{moral hazard} (hidden actions) in the classic principal-agent model.} Consequently, the principal designs a contract that specifies the payment to the \zj based on the observed outcome of the task. In our model, a contract is represented by a vector $\ctrctprf = (\ctrct_{1}, \ctrct_{2}, \ldots, \ctrct_{\otcnum}) \in \nnreals^{\otcnum}$, where $\ctrct_{j}$ ($j \in [\otcnum]$) denotes the fraction of the reward transferred from the principal to the \zj if outcome $j$ is realized.
By requiring that transfers are always non-negative, we enforce the standard \emph{limited liability} property in the contract theory. Note that $\ctrct_{j}$ may exceed $1$.
As a sanity check, if agent $i$ completes the task, the expected payment from the principal to the \zj is $\inner{\ctrctprf}{\rwdprf_{i}}$, and the principal's expected utility is $\inner{\ctrctone - \ctrctprf}{\rwdprf_{i}}$, where $\ctrctone = (1, 1, \ldots, 1)$ denotes the $\otcnum$-dimensional all-ones vector.

By the above definition, the optimal contract that maximizes the principal's utility may be impractical in real-world applications\citep[see, e.g.,][]{Car-15,DRT-19}.
Therefore, we also consider a simple yet widely studied alternative: the \emph{linear} contract.
In a linear contract, the principal commits to paying the \zj a fixed fraction of the realized outcome (i.e., payments are linear in the rewards of the outcomes).
Formally, a linear contract is represented by a scalar $\ctrctlnr \in [0, 1]$, so that the payment to the \zj is $\ctrctlnr \cdot \outcome_j$ if outcome $j$ occurs. A linear contract $\ctrctlnr$ is equivalent to a contract vector $\ctrctprf = \ctrctlnr \cdot \ctrctone$. 

\xhdr{Mechanism design between the \zj and agents.}
After the principal issues a contract $\ctrctprf \in \nnreals^{\otcnum}$ to the \zj in time $T = 1$, the \zj chooses a mechanism $\mech_{\ctrctprf} = (\allocrule_{\ctrctprf}, \paymentrule_{\ctrctprf})$ consisting of the allocation rule $\allocrule_{\ctrctprf}$ and the payment rule $\paymentrule_{\ctrctprf}$ in time $T = 2$.
Given any reported cost profile $\costprf = (\cost_{1}, \cost_{2}, \ldots, \cost_{\agentnum})$, the allocation rule $\allocrule_{\ctrctprf} \colon \reals_{\geq 0}^{\agentnum} \to [0, 1]^{\agentnum}$ maps the cost profile to the probability of agents being assigned the task.
The allocation rule $\allocrule_{\ctrctprf}$ satisfies the feasibility constraint, i.e., $\sum_{i\in\agents} \alloc_{\ctrctprf, i}(\costprf) \leq 1$ for every reported cost profile $\costprf$.
The payment rule $\paymentrule_{\ctrctprf} \colon \reals_{\geq 0}^{\agentnum} \to \reals^{\agentnum}$ maps the cost profile to the payments from \zj to each agent, i.e., $\paymentrule_{\ctrctprf} = (\payment_{\ctrctprf, 1}, \payment_{\ctrctprf, 2}, \ldots, \payment_{\ctrctprf, \agentnum})$, where $\payment_{\ctrctprf, i}$ is the payment function to agent $i$.
We assume the mechanism is dominant strategy incentive compatible (DSIC), i.e., every agent maximizes her utility by reporting her cost truthfully regardless of other agents' reports,\colorfootnote{Due to the standard revelation principle \citep{Mye-81}, this is without loss of generality.} and ex post individual rational (ex-post IR), i.e., every agent receives non-negative utilities in the equilibrium regardless of other agents' costs.

We assume that all parties are risk-neutral and have quasi-linear utilities.
The principal, the \zj, and agents aim to maximize their own utilities.
The expected utility of the principal is $\sum_{i\in\agents}\inner{\ctrctone - \ctrctprf}{\rwdprf_{i}} \cdot \alloc_{\ctrctprf, i}(\costprf)$.
The expected utilities of the \zj and agent $i$ are $\sum_{i\in\agents} \parent{\inner{\ctrctprf}{\rwdprf_{i}} \cdot \alloc_{\ctrctprf, i}(\costprf) - \payment_{\ctrctprf, i}(\costprf)}$ and $\payment_{\ctrctprf, i}(\costprf) - \cost_{i} \cdot \alloc_{\ctrctprf, i}(\costprf)$.\colorfootnote{Our model assumes that there is no moral hazard among agents, i.e., agents have no hidden action.}
A graphical illustration of this model is provided in \Cref{fig:model illustration}. 

We model the interaction among the three parties as an extensive-form Stackelberg game and analyze its subgame perfect equilibrium (SPE).
An SPE is a behavioral strategy profile where, in every proper subgame, the restricted strategy profile forms a Nash equilibrium.
In our model, an SPE ensures that the principal offers a contract maximizing her utility, then the \zj implements a DSIC and ex-post IR mechanism that maximizes her utility according to the offered contract, and the agents report their costs truthfully.

\begin{figure}[t]
    \centering
    \clearpage{}\resizebox{1\textwidth}{!}{
\begin{tikzpicture}[node distance=1.5cm, auto]
\node (principal) [draw, rectangle, minimum height=2em] {principal};
\node (utility_principal) [below right= 1cm and -3.2cm of principal] {
        \begin{tabular}{c}
            principal's utility:\\
            $\sum_{i\in\agents} \inner{\ctrctone - \ctrctprf}{\rwdprf_{i}} \cdot \alloc_{\ctrctprf, i}(\costprf)$
        \end{tabular}
        };
      
\node (intermediary) [draw, rectangle, minimum height=2em, right=4cm of principal] {\zj};
        \node (utility_intermediary) [below = 1cm of intermediary] {
        \begin{tabular}{c}
             \zj's utility:  \\
             $\sum_{i\in\agents} \parent{\inner{\ctrctprf}{\rwdprf_{i}} \cdot \alloc_{\ctrctprf, i}(\costprf) - \payment_{\ctrctprf, i}(\costprf)}$ 
        \end{tabular}
        };  
      
\node (agent3) [draw, circle, right=3.5cm of intermediary.east,opacity=0, text opacity=0]{};
        \node (agent3tag) [right=.2cm of agent3, opacity=0, text opacity=0] {agent $3$};
        \node[opacity=1] at ($(agent3.mid)$) (agentdots) {$\vdots$};
        \node (agenttagdots) [right=.6cm of agentdots] {$\vdots$};
        \node (agent1) [draw, circle, above =.7cm of agent3]{};
        \node (agent1tag) [right=.1cm of agent1] {agent $1$};
        \node (agent2) [draw, circle, above =.2cm of agent3]{};
        \node (agent2tag) [right=.1cm of agent2] {agent $2$};
        \node (agentn) [draw, circle, below =.4 of agent3]{};
        \node (agentntag) [right=.1cm of agentn] {agent $\agentnum$};
\node (utility_agent) [below right = 1cm and 3cm of intermediary] {
        \begin{tabular}{c}
            agent $i$'s utility: \\
            $\payment_{\ctrctprf, i}(\costprf) - \cost_{i} \cdot \alloc_{\ctrctprf, i}(\costprf)$
        \end{tabular}
        };
        
\draw[->] (principal.east) -- (intermediary.west);
        \draw[->] (intermediary.east) -- (agent1.west);
        \draw[->] (intermediary.east) -- (agent2.west);
        \draw[->] (intermediary.east) -- (agentn.west);
      
\draw [dashed] ($(principal.west)+(-0.2cm,-1.2cm)$) -- ++(0,3.2cm);
        \draw [dashed] (intermediary) -- ++(0,2cm);
        \draw [dashed] (intermediary) -- ++(0,-1.2cm);
        \draw [dashed] ($(agent3tag.east)+(0,-1.2cm)$) -- ++(0,3.2cm);
\node (contract) [above= 1.5cm of $(principal.mid)!0.5!(intermediary.mid)$] {contract $\ctrctprf \in \nnreals^{\otcnum}$};
\node (mechanism) [above= 1.5cm of $(intermediary.mid)!0.5!(agent3tag.east)$] {mechanism $\mech_{\ctrctprf} = \langle \allocrule_{\ctrctprf}, \paymentrule_{\ctrctprf} \rangle$};
\end{tikzpicture}
}
\vspace{-1.3cm}

\clearpage{}
    \caption{
    Graphical illustration of the three-party model.
    In time $T = 1$, the principal designs contract~$\ctrctprf$ for the \zj.
    In time $T = 2$, the \zj designs mechanism $\mech_{\ctrctprf}$ with allocation rule $\allocrule_{\ctrctprf}$ and payment rule $\paymentrule_{\ctrctprf}$ for agents, where $\cost_{i}$ is the private cost of the agent $i\in\agents$.
    }
    \label{fig:model illustration}
    \vspace{-0.3cm}
\end{figure}

\begin{remark}[Connection to classic principal-agent model]
    Our three-party model can be viewed as the classic contract design in the (two-party) principal-agent models.
    Specifically, the \zj (in the three-party model) can be considered as an agent (in the classic two-party contract design model), where all possible mechanisms \tbedit{form} her hidden action space.
\end{remark}

We conclude this subsection by introducing two key concepts that are crucial in our analysis.

\begin{definition}[\Contribution and \cwelfare of agents]
\label{def:welfare and cwelfare}
    For an agent $i \in \agents$ with\footnote{For the sake of brevity, we will henceforth abuse notation and refer to an agent \emph{with} expected reward rather than expected reward \emph{generated by} an agent.}
    an expected reward $\rewardfun_{i}$ and cost $\cost_{i}$, we define the following terms:
    \begin{enumerate}
        \item The \emph{\contribution of agent $i$} is defined as $\welfare_{i} \deq \exrwd_{i} - \cost_{i}$. We denote its distribution as $\welfaredis_{i}$, which is due to the randomness of the cost $\cost_{i}$.
        \item Given a contract $\ctrctprf \in \nnreals^{\otcnum}$, the \emph{\cwelfare of agent $i$} is defined as $\cwf_{\ctrctprf, i} \deq \inner{\ctrctprf}{\rwdprf_{i}} - \cost_{i}$.
        We denote its distribution as $\valdis_{\ctrctprf, i}$, which is also due to the randomness of $\cost_{i}$.
    \end{enumerate}
\end{definition}

\noindent
The intuition behind these concepts is as follows.
Suppose an agent $i \in \agents$ is selected to complete the task.
Given the realized cost $\cost_{i}$, the expected \contribution, i.e., the total utility among all three parties, over the randomness of the task outcome is $\rewardfun_{i} - \cost_{i}$.
Since there is a one-to-one mapping between an agent's cost $\cost_{i}$ and \contribution $\welfare_{i}$, we use them interchangeably throughout the paper, along with their corresponding distributions.
Moreover, given a contract $\ctrctprf \in \nnreals^{\otcnum}$ issued by the principal, the \cwelfare $\cwf_{\ctrctprf, i}$ represents the expected total utility among the \zj and the agents, excluding the principal.
Notably, an agent's \contribution $\welfare_{i}$ remains independent of the contract $\ctrctprf$, whereas the contracted \contribution $\cwf_{\ctrctprf, i}$ explicitly depends on it.

\subsection{Model Discussion} 
\label{subsec:model discussion}

In this subsection, we discuss key assumptions underlying modeling choices.

To justify the assumption that the principal initiates the contract, we note that in real-world crowdsourcing markets, requesters typically hold greater bargaining power because multiple intermediaries compete for their business.
For example, both Upwork and Freelancer charge clients a similar fee of around 4\% \citep{link-upworkfreelance}.
This competitive dynamic is further illustrated by shifts in the microtask market: on Amazon Mechanical Turk, requesters historically benefited from low fees, but after Amazon raised its commission to 20\% in 2015, many migrated to alternatives such as Prolific, which offered better user terms and higher data quality \citep{link-turkmigrate}.
These observations align with standard contract theory, where principals with strong outside options propose contracts that agents accept.

To justify the hidden mechanism assumption, we observe that in many real-world crowdsourcing platforms, it is difficult for requesters to observe the intermediary's inner workings.
Platforms like \href{https://support.upwork.com/hc/en-us/articles/1500007918681}{Upwork} deliberately design their talent rankings, proposal sorting, and payment structures to be proprietary and opaque. Requesters typically see only surface-level outcomes, such as sorted freelancer lists and basic performance metrics.
This opacity motivates our assumption that the principal does not observe the intermediary's mechanism. Modeling full opacity isolates the core contracting problem and reflects a standard and realistic assumption for such markets.

Contract design with private information has been studied in previous work \citep{ADT-21, GSW-21, CMG-21, ADLT-23}. \citet{CMG-21} study a multi-dimensional private-type setting where both the cost and the outcome distribution are private information.
\citet{GSW-21} consider a setting where only the outcome distribution is private.  
In contrast, our model assumes that the only uncertainty lies in agents' skill levels, which determine their costs to complete the task.
This single-parameter type is also studied in \citet{ADT-21, ADLT-23}, where an agent's cost corresponds to the cost per unit of effort in their model, as the outcome space in our model is binary.  

To justify the assumption in our model that the principal and platform know workers' outcome distributions, we note that crowdsourcing platforms systematically track task completion rates and historical outcomes. This enables them to estimate outcome distributions based on accumulated data.
While the principal lacks direct access to this data, she can still infer outcome distributions from publicly available information: Platforms display ratings, completion rates, and certifications, and workers often showcase their past successes in their public profiles to attract employers. These signals provide the principal with a reasonable estimate of success likelihood.
It is worth emphasizing that all our results continue to hold when the outcome distributions are private (see \Cref{apx:unobservable reward} for a detailed discussion of this assumption).

Now we clarify the information structure in our model. The principal and the \zj know the market size $\agentnum$, as well as the success probability $\succpr_i$ and cost distribution $\costdis_i$ of each agent $i \in \agents$.
However, agent $i$'s realized cost $\cost_i$ is privately known only to the agent herself.  
Furthermore, the mechanism implemented by the \zj is hidden from the principal, who observes only the task's stochastic reward $\outcome$, or equivalently, the task's outcome. The mechanism, however, is known to the agents before their costs are realized.  
Finally, in \Cref{sec:robust}, we consider a robust setting where the exact value of $\agentnum$ is known only to the \zj, while the principal knows only a range.

\subsection{Revenue-Maximization Mechanism Design} \label{sec:myerson revisited}
In this subsection, we revisit the single-item revenue maximization in the Bayesian environment.
We consider the scenario where a seller intends to sell a single item to a single buyer.
The buyer's private value of the item, denoted by $\val$, is drawn from a prior distribution $\valdis$.
\colorfootnote{With a slight abuse of notation, we also use $\valdis$ and $\valdens$ to denote the cumulative distribution function (CDF) and probability density function (PDF) of the distribution $\valdis$, respectively.
Besides, $\supp{\valdis}$ denotes the support of $\valdis$.}

A posted-price mechanism offers a take-it-or-leave-it price $\price\in\supp{\valdis}$ to the buyer, resulting in an expected revenue of $\price \cdot (1 - \valdis(\price))$ for the seller.
The \emph{monopoly reserve price} $\reserve$ and the \emph{monopoly revenue} $\revenuecurve^{*}$ of a distribution $\valdis$ are defined as follows.
\colorfootnote{When there are multiple alternative monopoly prices, ties are broken by selecting the largest one.}
\begin{align*}
    \reserve \deq \argmax\nolimits_{\price \in \supp{\valdis}} \price \cdot (1 - \valdis(\price)),
    \quad && \quad
    \revenuecurve^{*} \deq \max\nolimits_{\price \in \supp{\valdis}} \price \cdot (1 - \valdis(\price)).
\end{align*}
According to these definitions, the maximum revenue achievable by a posted-pricing mechanism is the monopoly revenue $\revenuecurve^{*}$, obtained by offering the price $\reserve$.

The \emph{hazard rate function} $\hr \colon \supp{\valdis} \to \reals$ and \emph{virtual value function} $\virval \colon \supp{\valdis} \to \reals$ of the value distribution $\valdis$ are defined as follows.
\begin{align*}
    \hr(\price) \deq \frac{\valdens(\price)}{1 - \valdis(\price)},
    \quad && \quad
    \virval(\price) \deq \price - \frac{1}{\hr(\price)}.
\end{align*}
Additionally, the \emph{cumulative hazard rate function} $\HR \colon \supp{\valdis} \to \reals_{\geq 0}$ is defined as $\HR(\price) \deq \yhedit{\int_0^\price \hr(x) \dd x} = -\ln(1 - \valdis(\price))$.
In the following, we introduce two important types of distributions.

\xhdr{Regular distributions.}
A distribution $\valdis$ is \emph{regular} if the corresponding virtual value function $\virval(\cdot)$ is monotone non-decreasing in $\supp{\valdis}$.
Assuming that the distribution $\valdis$ is regular, $\virval^{-1}(0)$\colorfootnote{Note that the virtual value function $\virval(\cdot)$ may not be strictly increasing, i.e., there can be two different values \tbedit{both corresponding} to the same virtual value.
For ease of notation, throughout the paper, we always break ties by choosing the largest one, namely, $\virval^{-1}(\price)=\max\set{\temp\ge 0\suchthat \virval(\temp)\le \price}$.} is exactly the monopoly reserve price $\reserve$.
Furthermore, due to the monotonicity of $\virval(\cdot)$ and the definition of the monopoly reserve price $\reserve$, we have $\virval(\price) \ge 0$ for all $\price \ge \reserve$.

\xhdr{MHR distributions.}
A distribution $\valdis$ satisfies the \emph{monotone-hazard-rate} (MHR) property if the corresponding hazard rate function $\hr(\cdot)$ is monotone non-decreasing in $\supp{\valdis}$.
It is worth noting that the MHR property is a stronger assumption than the regularity property.

\xhdr{Quantile space, revenue curve, and cumulative hazard rate function.}
Sometimes, it is more convenient to work in the quantile space rather than in the actual value domain.
Specifically, the quantile of a distribution $\valdis$ with respect to a price $\price \in \supp{\valdis}$ is defined as $\quantile \deq 1 - \valdis(\price)$.
Using this, the \emph{revenue curve} $\revenuecurve \colon [0, 1] \to \reals_{\geq 0}$ of the distribution $\valdis$ is defined as $\revenuecurve(\quantile) \deq \quantile \cdot \valdis^{-1}(1 - \quantile)$ for any $\quantile \in [0, 1]$.
In other words, if $\price \in \supp{\valdis}$ is a price and $\quantile = 1 - \valdis(\price)$ is the corresponding quantile, $\revenuecurve(\quantile)$ represents the expected revenue of selling the item at price $\price$.
It can be verified that the derivative $\revenuecurve'(\quantile)$ equals the virtual value $\virval(\price)$, where $\price = \valdis^{-1}(1 - \quantile)$.
The regularity of $\valdis$ implies that $\revenuecurve'(\cdot)$ is non-increasing, making $\revenuecurve(\cdot)$ concave.
Additionally, the MHR property of $\valdis$ implies that the cumulative hazard rate function $\HR(\cdot)$ is convex.

\subsection{Assumptions and Notations}

In this subsection, we introduce the assumptions and notations that are used throughout the paper.

\begin{assumption}\label[assumption]{asp:welfare regular}
Throughout this paper, we assume that for each agent $i$, the distribution of her \contribution $\welfare_i$, denoted by $\welfaredis_i$, satisfies the following four properties:
    \begin{itemize}
        \item \emph{Regularity}: The virtual value of $\welfaredis_{i}$ denoted as $\virval_{i}$, is weakly increasing.
        It is equivalent to the cost distribution being regular \textup{\citep[e.g., see][]{MS-83,HPS-25}}, i.e., the virtual value of the cost, $\cost + \frac{\costdis(\cost)}{\costdis'(\cost)}$, is weakly decreasing.
        \item \emph{Continuity}: The support $\supp{\welfaredis_{i}}$ is a single interval, and the distribution $\welfaredis_{i}$ can only have probability mass at its \emph{right-endpoint}.
        \item \emph{Differentiability}: The CDF $\welfaredis_{i}(\cdot)$ is left- and right-differentiable everywhere within the interior of the support $\supp{\welfaredis_{i}}$.
        Without loss of generality, the corresponding PDF $\welfaredens_{i}(\cdot)$ exists and is right-continuous everywhere.
        \item \emph{Non-degeneracy}: The largest value in support $\supp{\welfaredis_{i}}$ is strictly positive, i.e., $\welfaredis_{i}(0) < 1$.
    \end{itemize}
\end{assumption}
The first three properties are standard in the mechanism design literature.
In particular, by applying the standard ironing technique \citep{Mye-81}, the main results in \Cref{sec:optimal contract} can be extended to scenarios where the cost distribution is irregular.
The final property holds without loss of generality, as an agent will never be selected to complete the task if her realized \contribution is always non-positive.

We use the notation $\agents = \{1, 2, \ldots, \agentnum\}$ and $[\lownum : \upnum] = \{\lownum, \lownum + 1, \ldots, \upnum\}$ for $\lownum \leq \upnum$.
Given a distribution $\valdis$, we reuse $\valdis$ to denote its CDF and let $\valdens$ represent its PDF.
Bold symbols are used to denote vectors or joint distributions.
For instance, $\valprf = (\val_{1}, \val_{2}, \ldots, \val_{\agentnum})$ and $\valdisprf = \valdis_{1} \otimes \valdis_{2} \otimes \cdots \otimes \valdis_{\agentnum}$.
Specifically, for $\agentnum$ identical distributions $\set{\valdis}_{i = 1}^{\agentnum}$, we denote the joint distribution $\valdisprf$ as $\valdis^{\otimes \agentnum}$.
For $\agentnum$ random variables $\valprf \sim \valdisprf$, we use $\valmax$ to denote the first-order statistic among them, i.e., $\valmax \deq \max_{i \in \agents} \val_{i}$.
For any $\temp\in \reals$, define $\plus{\temp}\deq \max\set{\temp, 0}$.
 
\section{Optimal Contract for the Principal}
\label{sec:optimal contract}
In this section, we characterize the optimal contract for the principal under various scenarios.

\begin{restatable}[Optimal contract characterization]{theorem}{thmOptContract}
\label{thm:opt contract}
    The optimal contract of the principal satisfies the following characterization:
\begin{enumerate}
[label=\textup{(}\roman*\textup{)}]
        \item \label[part]{thm: general} 
        For general instances, the optimal contract and the corresponding payoff of the principal is the optimal solution and objective value of the optimization defined as follows:
\begin{align*}
            \max_{\ctrctprf \in \nnreals^{\otcnum}} \sum\nolimits_{i \in \agents} \inner{\ctrctone - \ctrctprf}{\rwdprf_{i}} \cdot \prob[\cwfprf_{\ctrctprf} \sim \valdisprf_{\ctrctprf}]{\virval_{\ctrctprf, i}(\cwf_{\ctrctprf, i}) \geq \max\nolimits_{j \in \agents}\plus{\virval_{\ctrctprf, j}(\cwf_{\ctrctprf, j})}}.
        \end{align*}

        \item \label{thm: idential rwd} 
        For instances with identical expected reward profiles \textup{(}i.e., $\rwdprf_{i} \equiv \rwdprf$ with each $i \in \agents$\textup{)}, there exists an optimal contract that is linear.
        In particular, this optimal linear contract and the corresponding payoff of the principal are the optimal solution and objective value of the optimization defined as follows:
\begin{align*}
            \max_{\ctrctlnr \in [0,1]} \sum\nolimits_{i \in \agents} {\virval_{i}(\temp_{\ctrctlnr, i})} \cdot \prob[\welfareprf \sim \welfaredisprf]{\welfare_{i} \geq \temp_{\ctrctlnr, i} \wedge \virval_{i}(\welfare_{i}) = \max\nolimits_{j \in \agents}\virval_{j}(\welfare_{j})},
        \end{align*}
        where $\temp_{\ctrctlnr, i} \deq (1 - \ctrctlnr) \exrwd + \reserve_{\ctrctlnr, i}$, $\exrwd=\sum_{j\in \otcs} \rwd_j$ is the expected reward from agents, $\reserve_{\ctrctlnr, i}$ is the monopoly reserve price of agent $i$'s \cwelfare distribution $\valdis_{\ctrctlnr, i}$ and $\virval_{i}(\cdot)$ is the virtual value function of her \contribution distribution $\welfaredis_{i}$.

        \item \label{thm: identical rwd and dist} 
        For instances with identical expected reward profiles and \contribution distributions\colorfootnote{This is equivalent to assuming identical outcome distributions and cost distributions.} \textup{(}i.e., $\rwdprf_i \equiv \rwdprf$ and $\welfaredis_i \equiv \welfaredis$ for each $i \in \agents$\textup{)}, there exists an optimal contract that is linear.
        In particular, the optimal payoff of the principal is the optimal objective value of the optimization program defined as follows:
        \begin{align*}
            \max_{\temp \in \supp{\welfaredis}} \virval(\temp) \parent{1 - \welfaredis^{\agentnum}(\temp)}.
        \end{align*}
        Moreover, let $\temp^{\star}$ be the solution of the optimization program.
        The principal's optimal (linear) contract $\ctrctlnropt$ satisfies $\temp^{\star} = (1 - \ctrctlnropt)\exrwd + \reserve_{\ctrctlnropt}$.
    \end{enumerate}
\end{restatable}

We first discuss the implications of the above characterization and then sketch its proof.
The formal proof can be found in \Cref{apx:opt contract proof}.

\xhdr{Implications of \Cref{thm:opt contract}.}
The first two formulations in \Cref{thm:opt contract} address scenarios with heterogeneous agents (i.e., agents differ in either rewards or cost distributions), and are expressed as summations over all agents, with each term representing the utility an agent brings to the principal multiplied by that agent's probability of winning.
In contrast, the third formulation considers essentially homogeneous agents (i.e., all agents share identical reward and cost distributions), and is expressed as the utility a single agent provides to the principal multiplied by the probability that at least one agent is selected.
Next, we discuss each of these formulations separately.

The first formulation reformulates the principal's problem by expressing each agent's winning probability as a function of the contract $\ctrctprf$.
Recall that each agent's \cwelfare is determined by $\cwf_{\ctrctprf, i} = \inner{\ctrctprf}{\rwdprf_{i}} - \cost_{i}$, which depends on the agent's cost, expected reward profile, and the chosen contract.
Consequently, different contracts for distinct outcomes can reshape the ranking of agents' \cwelfare{s}, thereby affecting the selection of winners.
Furthermore, since winning probabilities are tied to the distributions of \cwelfares, they become contract-dependent as well.

The second formulation assumes all agents have the same expected reward profiles $\rwdprf_{i} \equiv \rwdprf$ but can differ in cost distributions.
{Under this condition, we show the optimality of the linear contract, thereby simplifying the principal's contract design space from the $m$-dimensional space (i.e., $\nnreals^\otcnum$) to the single-dimensional space (i.e., $[0, 1]$).}
Notably, the principal's contract design problem can be simplified to setting a ``discriminatory threshold'' $\temp_{\contract, i} \deq (1 - \contract)\exrwd + \reserve_{\contract, i}$ for each agent $i$.
Agent $i$ wins if her \contribution is the highest and meets or exceeds her posted price.
The principal's utility from agent $i$ is determined by the virtual value of her threshold $\virval_{i}(\temp_{\contract, i})$ (instead of the virtual value of her contribution).
Since $\temp_{\contract, i} = (1 - \contract)\exrwd + \reserve_{\contract, i}$ is defined for each agent $i \in \agents$, the discriminatory threshold profile $\set{\temp_{\contract, i}}_{i=1}^{\agentnum}$ is restricted by $\contract$. 
Consequently, conditional on a realized cost profile, changing the contract $\contract$ does not affect the ranking of the agents' private values or the identity of the winner.
While the threshold for each agent depends on the contract, her winning probability remains unchanged as long as her \contribution meets or exceeds her posted threshold.

The third formulation assumes all agents have identical expected reward profiles and \contribution distributions, simplifying the principal's problem further into a \emph{virtual value pricing} problem.
Here, the principal posts an anonymous price $\temp$. 
A sale occurs if at least one agent's \contribution meets or exceeds this price. 
The principal realizes a utility equal to the \emph{virtual value} of the posted price $\temp$ instead of the price itself.

\xhdr{Proof overview of \Cref{thm:opt contract}.}
We outline the main ideas behind the proof of \Cref{thm:opt contract}. 
To derive the first formulation in the general cases, we begin by applying the characterization of the \zj's optimal mechanism from \Cref{thm:opt mech} to express each agent's winning probability. 
This enables us to reformulate the principal's problem as maximizing, over all contracts, the sum over agents of their winning probabilities weighted by the principal's utility from each agent.  
For the case in which agents have identical expected reward profiles, we analyze the relationship between each agent's \contribution and her value.
Then we rewrite both the winning probabilities and the principal's utility in terms of agents' \contributions and the corresponding virtual values. 
This conversion relies on the assumption that expected reward profiles are identical.  
To establish the third formulation, which additionally assumes identical \contribution distributions, we reformulate the sum of winning probabilities as the probability that the first-order statistic of the agents' \contribution exceeds the anonymous price induced by the contract. 
Finally, we show that the principal's maximum expected utility from designing an optimal contract aligns with the maximum revenue obtainable by selecting an optimal price from the support of the \contribution distribution, as demonstrated by analyzing the value range of the anonymous price derived from the contract (see \Cref{lem:reserve continuous}).

\begin{restatable}[Optimal mechanism characterization]{proposition}{thmOptMech}
\label{thm:opt mech}
    Given any contract $\ctrctprf \in \nnreals^\otcnum$, the optimal mechanism maximizing the \zj's expected utility is the virtual welfare maximizer, where the private value of agent $i \in \agents$ is equal to her \cwelfare $\cwf_{\ctrctprf, i}$.
    Specifically, the optimal mechanism of the \zj given the contract $\ctrctprf$, denoted by $\mech_{\ctrctprf}^{\star} = (\allocrule_{\ctrctprf}^{\star}, \paymentrule_{\ctrctprf}^{\star})$, takes the reported cost profile $\costprf$ as input and outputs the allocation and payment rules as follows:
    \begin{itemize}
        \item Transform the reported cost profile $\costprf$ into the \cwelfare profile 
        \begin{align*}
            \cwfprf_{\ctrctprf} 
            = \parent{\inner{\ctrctprf}{\rwdprf_{1}} - \cost_{1}, \inner{\ctrctprf}{\rwdprf_{2}} - \cost_{2}, \ldots, \inner{\ctrctprf}{\rwdprf_{\agentnum}} - \cost_{\agentnum}},
        \end{align*}
        and compute the virtual values $\set{\virval_{\ctrctprf, i}(\cwf_{\ctrctprf, i})}_{i\in \agents}$ of the agents' \cwelfare{s} $\{\valdis_{\ctrctprf, i}\}_{i\in\agents}$. 
        \item If the above virtual values are all negative, no agent is assigned the task; otherwise, assign the task to the agent with the maximum virtual value.
        \item Define the payment $\paymentrule_{\ctrctprf}^{\star}$ from the \zj to the agents such that, for each $i \in \agents$,
        \begin{align*}
            \payment_{\ctrctprf, i}^{\star}(\costprf) \deq 
            \left(\inner{\ctrctprf}{\rwdprf_{i}} - \virval_{\ctrctprf, i}^{-1} \parent{\max_{j \in \agents \setminus \set{i}} \plus{\virval_{\ctrctprf, j}\parent{\cwf_{\ctrctprf,j}}}}
            \right)\cdot \alloc^{\star}_{\ctrctprf, i}(\costprf),
\end{align*}        
        where $\alloc^{\star}_{\ctrctprf, i}(\costprf)$ is the $i$th component of $\allocrule^{\star}_{\ctrctprf}(\costprf)$, representing the probability that agent $i$ is assigned as the winner.
    \end{itemize}
\end{restatable}

\Cref{thm:opt mech} demonstrates that the optimal mechanism for the \zj is essentially the virtual welfare maximizer mechanism, as described in \Cref{def:vwm}, where the \zj and the agents correspond to the seller and the buyers, respectively, and the agents' \cwelfare{s} are equivalent to the buyers' values of the auctioned item.
 The proof of \Cref{thm:opt mech} is relatively standard and thus deferred to \Cref{apx:opt mech proof}. 
\section{Impact of \Zj and Price of Double Marginalization}
\label{sec:impact}

In this section, we analyze the intermediary's impact on the principal's utility, which is of both theoretical and practical importance in crowdsourcing markets. If the principal's utility in the presence of an intermediary closely matches that in a setting where she directly engages with agents, then the intermediary can be considered a cost-effective implementation channel. However, a significant utility gap may reflect substantial platform costs, potentially discouraging the principal from using the platform. To quantify this impact, we compare our model to two benchmarks in which the principal interacts directly with agents.

Double marginalization refers to the inefficiency that arises when multiple parties in a marketplace -- such as an intermediary and downstream agents -- each impose their own markup. It is well established in the literature that the presence of an intermediary typically reduces the principal's utility due to this sequential distortion. This observation motivates the central goal of this section: to quantify the principal's utility loss caused by double marginalization. In doing so, we aim to complement qualitative insights with a formal characterization of efficiency loss under various benchmarks.

To this end, we introduce the concept of the \emph{price of double marginalization} (\podm), which captures the efficiency loss resulting from the intermediary's presence. 
Specifically, \podm quantifies the ratio between the principal's utility under the second-best benchmark,
in which she directly interacts with agents via an optimal DSIC mechanism, and the her utility 
obtained in our delegated mechanism design model with an \zj.

\begin{definition}
    The \emph{price of double marginalization} (\podm) is defined as the ratio of the principal's utility under the second-best benchmark to her optimal utility in our model:
    \begin{equation*}
        \ratiosb \triangleq \frac{\utipsb}{\utip^{*}},
    \end{equation*}
    where $\utip^{*}$ denotes the principal's optimal utility in our model, and $\utipsb$ denotes the second-best benchmark, corresponding to the expected revenue of the virtual welfare maximizer mechanism:
    \begin{equation*}
        \utipsb = \expect[\welfareprf \sim \welfaredisprf]{\max\nolimits_{i \in \agents}\left\{\plus{\virval_{i}(\welfare_{i})} \right\}}.
    \end{equation*}
\end{definition}

Intuitively, the \podm measures the extent to which the principal's utility deteriorates due to delegation to an intermediary, as compared to the ideal scenario in which the principal can directly implement a mechanism to maximize her expected revenue.

To evaluate the overall efficiency loss relative to the welfare-optimal outcome, we additionally consider the classical \emph{price of anarchy} (\poa). This benchmark compares the principal's utility under the first-best benchmark—which maximizes total welfare—to that in our model. 

\begin{definition}
    The \emph{price of anarchy} (\poa) is defined as the ratio of the principal's utility under the first-best benchmark to her optimal utility in our model:
    \begin{equation*}
        \ratiofb \triangleq \frac{\utipfb}{\utip^{*}},
    \end{equation*}
    where $\utip^{*}$ denotes the principal's optimal utility in our model, and $\utipfb$ denotes the \emph{first-best benchmark}, corresponding to the maximum expected social welfare achievable when the principal can observe the realized cost profile and directly select the agent that maximizes welfare:
    \begin{equation*}
        \utipfb = \expect[\welfareprf \sim \welfaredisprf]{\max\nolimits_{i \in \agents} \welfare_{i}}.
    \end{equation*}
\end{definition}

The first-best and second-best benchmarks in \tbedit{the above} definitions correspond, respectively, to the maximum welfare and the maximum revenue in a related single-item auction:

Since the \contribution distributions $\welfaredisprf$ are regular by \Cref{asp:welfare regular}, we have the optimal mechanism in the second-best benchmark is the virtual welfare maximizer mechanism.
Therefore, the second-best benchmark is equivalent to the expected revenue of the virtual welfare maximizer mechanism, i.e., $\utipsb = \expect[\welfareprf \sim \welfaredisprf]{\max\nolimits_{i \in \agents}\left\{\plus{\virval_{i}(\welfare_{i})} \right\}}$.

Furthermore, the first-best benchmark is equivalent to the maximized utility of the principal in the scenario where the principal has direct access to agents and knows their cost realizations.
Indeed, the assumption in this scenario implies that the principal can assign the task to the one with the lowest cost, thereby obtaining a utility equivalent to her \contribution.
Therefore, the first-best benchmark is equivalent to the expectation of the maximum value between $0$ and the first order statistic $\welfaremax$ from the $\agentnum$ \contribution, i.e., $\utipfb=\expect[\welfareprf \sim \welfaredisprf]{\plus{\welfaremax}}$.

In summary, the first-best benchmark yields the highest utility, followed by the second-best benchmark, while the principal's utility in our model is the lowest.
That is,
\begin{equation*}
    \utipfb \geq \utipsb \geq \utip^{*},
\end{equation*}
which implies
\begin{equation*}
    \ratiofb \geq \ratiosb \geq 1.
\end{equation*}

\begin{remark}
    Throughout this section, we define the \poa{} with respect to the \emph{principal's utility}, rather than the notion of social welfare. This perspective is similar to that adopted in \citet{HHT-14}.
    Notably, all of our upper bounds on the PoA for utility also imply the same bounds on the PoA for efficiency (i.e., total welfare), since the principal's utility is always upper bounded by the welfare benchmark.
\end{remark}

In many crowdsourcing settings, tasks are simple and uniform, requiring similar effort from workers and leading to comparable success probabilities.
The \zj also attracts a specific type of workers, where individuals with similar expertise, experience, and availability tend to self-select into the platform, resulting in homogeneous cost distributions.
Given this natural homogeneity, we bound the ratios under the following assumption.

\begin{assumption}
\label[assumption]{asp:identical r and F}
    All agents have identical expected rewards $\rewardfun$ and \contribution distributions $\welfaredis$.
\end{assumption}

 \subsection{Comparison under Regular Distributions}

In this subsection, we bound both PoDM and PoA. Our main result in \Cref{thm:regular bounds} shows that for a regular \contribution distribution $\welfaredis$, if the agents' maximum possible contribution is at most $\ratiolu$ times their minimum possible contribution, 
then \podm is tightly bounded by $\Theta(\log \ratiolu)$, while the \poa is tightly bounded by $\Theta((\log \ratiolu)^2)$.\colorfootnote{We say a ratio is \emph{tight} if, for all \contribution distributions satisfying the given conditions and any number of agents $ \agentnum $, the ratio does not exceed a certain value, while there exists at least one distribution and one $ \agentnum $ where the ratio equals or approaches this value.}

\begin{restatable}{theorem}{regularbounds}\label{thm:regular bounds}
     For agents with identical expected rewards $\rewardfun$ and regular \contribution distributions $\welfaredis$ with support $\supp{\welfaredis} = \parents{\lowsupp, \upsupp}$ $\parent{0 < \lowsupp \leq \upsupp \leq \rewardfun}$, if $1 \leq \upsupp/\lowsupp \leq \ratiolu$, the tight \podm is $\bigtheta{\ratiolu}{\log \ratiolu}$ and the tight \poa is $\bigtheta{\ratiolu}{(\log \ratiolu)^{2}}$.
\end{restatable}

Note that the bounds of the ratios in \Cref{thm:regular bounds} require the support of the \contribution distribution to be a finite positive interval.
Additionally, the tightness ensures that if agents have a regular \contribution distribution with unbounded support, the two ratios are also unbounded.

\xhdr{Proof overview of \Cref{thm:regular bounds}.}
Here, we outline the main ideas behind the proof of \Cref{thm:regular bounds}, which consists of three parts.
In the first part, we establish the upper bound for \podm under two conditions with different arguments.
For the case where $4 \leq \upsupp/\lowsupp \leq \ratiolu$, we apply the monotonicity of the virtual value to reduce the maximized utility $\utip^{*}$ (\Cref{lemma:regular reserve increase}).
We then divide the quantile space of any \contribution distribution into $2\bceil{\log \ratiolu} + 1$ pieces.
Next, using the monotonicity of the revenue curve for regular distributions, we find an upper bound for monopoly revenue $\revenuecurve^{*}$ and the ratio between the maximized utility $\utip^{*}$ and $\revenuecurve^{*}$.
For the case where $1 \leq \upsupp/\lowsupp \leq 4$, we construct a quantile $\quantile^{\dag}$ representing the largest quantile where $\revenuecurve'(\quantile^{\dag}) \geq \lowsupp / 2$.
Then, by rescaling the integral form for $\revenuecurve^{*}$, we obtain a constant upper bound for the ratio between $\utip^{*}$ and $\revenuecurve^{*}$.
We finally use the ratio of revenues from anonymous pricing and Myerson's optimal auction (\Cref{lemma:regular OPA vs AP}) to establish an upper bound for the second-best benchmark $\utipsb$.
This process allows us to establish the upper bound for \podm.

In the second part, we define the \emph{truncated equal-revenue distribution} (\Cref{def:teqr distribution}) and show that it achieves the maximum first-best benchmark for a single agent (\Cref{lemma:regular fb worst case}).
Combining this with the upper bound for \podm, we derive the upper bound for \poa.

In the third part, we provide an example (\Cref{example:regular bounds}) to illustrate the lower bounds of the two ratios and demonstrate that the two ratios are tight.
Specifically, the lower bound part of \Cref{thm:regular bounds} utilizes \Cref{example:regular bounds}.
The upper bound part utilizes \Cref{lemma:regular mhr order statistic,lemma:regular reserve increase,lemma:regular OPA vs AP,lemma:regular fb worst case}.
The formal proof of \Cref{thm:regular bounds} can be found in \Cref{apx:regular bounds proof}. \subsection{Comparison under MHR Distributions}

In this subsection, we bound both PoDM and PoA for MHR \contribution distributions. Our first main result in \Cref{thm:mhr ratio bounds for n in naturals} shows that for MHR \contribution distributions, both ratios are constant—a significant improvement over the results obtained for regular distributions.
Notably, unlike our findings for regular distributions in \Cref{thm:regular bounds}, \Cref{thm:mhr ratio bounds for n in naturals} does not impose any restrictions on the support of the \contribution distribution.

\begin{restatable}{theorem}{mhrratiobounds} \label{thm:mhr ratio bounds for n in naturals}
     For agents with identical expected rewards $\rewardfun$ and MHR \contribution distributions $\welfaredis$, the tight \podm is in $[\ee,\ommhrupbound]$ and the tight \poa is $\ee^{2}$. Moreover, if $\agentnum=1$, the tight \podm is $\ee$.
\end{restatable}

Our second main result in \Cref{thm:r_fb asymp lb} provides a market-size–dependent bound for both ratios and shows that for fixed \contribution distribution, the two ratios asymptotically approach $1$ as the market size tends to infinity. This implies that the impact of the \zj's presence and agents' private information on the principal's utility diminishes to zero as the market expands.

\begin{restatable}{theorem}{fbasymp} \label{thm:r_fb asymp lb}
    For agents with identical expected rewards $\rewardfun$ and MHR \contribution distributions $\welfaredis$, then \podm and \poa asymptotically approach $1$ as $\agentnum$ approaches positive infinity.
    Moreover, let $\harmonic_{\agentnum} = 1 + \frac{1}{2} + \cdots + \frac{1}{\agentnum}$ be the $\agentnum$th harmonic number.
    Then we have
    \begin{equation*}
        \ratiosb 
        \leq \ratiofb
        \leq \parent{1 - \frac{\ln \harmonic_{\agentnum}}{\harmonic_{\agentnum}}}^{-1} \parent{1 - \parent{1 - \frac{1 - \welfaredis(0)}{\ee^{\harmonic_{\agentnum} - \ln \harmonic_{\agentnum} + 1}}}^{\agentnum}}^{-1}
        = 1 + \bigO{\agentnum}{\frac{\log \log \agentnum}{\log \agentnum}}.
    \end{equation*}
\end{restatable}
\xhdr{Implications of \Cref{thm:r_fb asymp lb}.} The upper bound in \Cref{thm:r_fb asymp lb} depends on both the number of agents $\agentnum$ and the probability $\welfaredis(0)$.
Viewing the agents as a market of size $\agentnum$, the market quality can be measured by $1 - \welfaredis(0)$, which is the probability that an agent's expected reward exceeds her cost.
There are two interesting observations. 
Firstly, in a market of fixed quality, as the market size grows, the impact of the \zj's participation and the presence of agents' private information on the utility of the principal diminishes to zero.
Secondly, for two markets of the same size, this impact is smaller in the higher-quality market (higher $1 - \welfaredis(0)$) compared to the lower-quality one.

\xhdr{Proof overview of \Cref{thm:mhr ratio bounds for n in naturals}.}
Here, we outline the main ideas behind the proof of \Cref{thm:mhr ratio bounds for n in naturals}, which consists of three parts. 
In the first part, we start by scaling down the maximized utility $\utip^{*}$ to ensure that the revenue from anonymous pricing is at most $\ee \cdot \utip^{*}$ (\Cref{lemma:repeated proof}). 
Using the ratio between the first-best benchmark and the revenue from anonymous pricing (\Cref{lemma:mhr ratio bound}), we find that the upper bound of \poa is $\ee^{2}$.

In the second part, we apply another version of \Cref{lemma:repeated proof} to scale down $\utip^{*}$ and break the upper bound for \podm into two multiplicative factors, one of which is provided by \Cref{lemma:mhr OPA vs AP}. 
For the other factor, we use its monotonicity for scaling. We calculate the upper bound for $\agentnum \leq 44$, as shown in \Cref{tab:mhr rsb}. 
For $\agentnum > 44$, since both factors decrease, we conclude that the upper bound of \podm is $\ommhrupbound$.

In the third part, we provide an example (\Cref{example:mhr bounds}) to illustrate the lower bounds of the two ratios.
The lower bound part of \Cref{thm:mhr ratio bounds for n in naturals} utilizes \Cref{example:mhr bounds}. The upper bound part utilizes \Cref{lemma:mhr prob bound,corollary:mhr prob bound,lemma:repeated proof,lemma:mhr ratio bound,lemma:mhr OPA vs AP}.
The formal proof of \Cref{thm:mhr ratio bounds for n in naturals} can be found in \Cref{apx:mhr ratio bounds proof}

\xhdr{Proof overview of \Cref{thm:r_fb asymp lb}.}
Here we provide a high-level overview of the proof of \Cref{thm:r_fb asymp lb}.
We first reference a conclusion that provides powerful tail bounds with respect to the order statistics of the distribution (\Cref{lemma: MHR concentration}).
Through this conclusion, we prove that the ratio of the first-best benchmark to the principal's utility has an upper bound related to the contract and the expectation of the first-order statistic (\Cref{lemma: lb of fb ratio}).
Then, we scale the principal's optimal utility by selecting a specific contract, thereby obtaining an upper bound for \poa in relation to the number of agents $\agentnum$.
The formal proof of \Cref{thm:r_fb asymp lb} can be found in \Cref{apx:r_fb asymp lb proof}

\section{\Zj with Anonymous Pricing}
\label{sec:anonymous pricing}

In this section, we consider the case where the mechanism space of the \zj is restricted.
In particular, we focus on the space of anonymous pricing mechanisms, which are widely studied in the algorithmic mechanism literature and practically prevalent in the crowdsourcing industry.
Throughout this section, we will assume the following for the sake of simplicity.

\begin{assumption} \label[assumption]{asp:identical r and F and posted price}
    Agents have identical expected reward profiles $\rwdprf$ and contribution distributions $\welfaredis$, and the mechanism space of the \zj is restricted to anonymous pricing mechanisms.
\end{assumption}

\begin{definition}[Anonymous pricing mechanism]
    Given each agent $i$'s cost $\cost_{i}$, the anonymous pricing mechanism for the \zj is defined as follows:
    \begin{enumerate}
        \item Post a price $\price \in \nnreals^\otcnum$.
        \item Let $W(\costprf) \triangleq \{i \in \agents: \cost_{i} \leq \price\}$ be the set of agents whose costs are at most the posted price. The task is allocated to one of these agents with identical probability, i.e., for each $i \in \agents$,
        \begin{align*}
            \alloc_{i}^{\dag}(\costprf) \deq \frac{1}{|W(\costprf)|} \cdot \mathds{1}\{i \in W(\costprf)\}.
        \end{align*}
        \item Define the payment $\paymentrule^{\dag}$ from the \zj to the agent such that, for each $i \in \agents$,
        \begin{align*}
            \payment_{i}^{\dag}(\costprf) \deq \price \cdot \alloc_{i}^{\dag}(\costprf).
        \end{align*}
        where $\alloc^{\dag}_{i}(\costprf)$ is the $i$th component of $\allocrule^{\dag}(\costprf)$, representing the probability of agent $i$ being the winner.
    \end{enumerate}
\end{definition}

Under \Cref{asp:identical r and F and posted price}, we can obtain that given a contract $\ctrctprf \in \nnreals^\otcnum$, the \zj's expected utility under the anonymous price $\price$ can be expressed as $\parent{\inner{\ctrctprf}{\rwdprf} - \price} \cdot \prob{\exists i \in \agents, s.t., \cost_{i} \leq \price}$.
Note that the probability that there exists an agent with cost no more than the posted price $\price$ is equal to the probability that there exists an agent with \contribution no less than $\inner{\ctrctprf}{\rwdprf} - \price$.
The \zj's expected utility of an anonymous price $\price$ is equal to $\parent{\inner{\ctrctprf}{\rwdprf} - \price} \cdot \parent{1 - \welfaredis^{\agentnum}_{\ctrctprf}\parent{\inner{\ctrctprf}{\rwdprf} - \price}}$.
By definition of reserve price, the \zj's optimal anonymous price is $\inner{\ctrctprf}{\rwdprf} - \reserve^{(\agentnum)}_{\ctrctprf}$, where $\reserve^{(\agentnum)}_{\ctrctprf}$ is the reserve price of the distribution $\welfaredis^{\agentnum}_{\ctrctprf}$.
Thus, the principal's utility under contract $\ctrctprf$ is equal to 
\begin{align*}
    \inner{\ctrctone - \ctrctprf}{\rwdprf} \cdot \parent{1 - \welfaredis^{\agentnum}_{\ctrctprf}(\reserve^{(\agentnum)}_{\ctrctprf})}.
\end{align*}
Similar to results in \Cref{sec:optimal contract}, we show that when agents are ex ante symmetric, the optimal contract is linear and its induced payoff for the principal can be considered as a variant of virtual value pricing (\Cref{thm:prin's utility under ap}).
The proof is similar to the one for \Cref{thm:opt contract}, and is deferred to \Cref{apx:prin's utility under ap proof}.

\begin{restatable}{theorem}{thmOptContractUnderAP}
\label{thm:prin's utility under ap}
    Under \Cref{asp:identical r and F and posted price}, there exists an optimal contract that is linear.
    In particular, the optimal payoff of the principal (denoted by $\utippost$) is the optimal objective value of the optimization program defined as follows:
    \begin{align*}
         \utippost = \max_{\temp \in \supp{\welfaredis}} \virval^{(\agentnum)}(\temp) \parent{1 - \welfaredis^{\agentnum}(\temp)}.
    \end{align*}
\end{restatable}

Under \Cref{asp:identical r and F and posted price}, \Cref{thm:prin's utility under ap} implies that the principal's problem of designing an optimal contract remains a virtual value pricing problem as established \Cref{thm: identical rwd and dist} of \Cref{thm:opt contract}. The only difference is that the corresponding virtual value function becomes that of the distribution $\welfaredis^{\agentnum}$. 
Intuitively, this virtual value pricing can be viewed as posting a price to a single virtual agent whose value is equal to the first-order statistic among those of the real agents.
If she buys at that price, a payment equal to her virtual value will be received.
Note that the principal's utility in this case is no more than that in the case where there is no restriction on the \zj's mechanism space. A natural question is whether there is any monotonicity in the principal's utility regarding two \zjs with an inclusion relationship between their mechanism spaces? We offer two negative examples for this question in \Cref{apx:diffrent mech space}.

We next analyze the impact of the \zj measured by \poa and \podm as we have done in \Cref{sec:impact}.\footnote{Here we define \poa and \podm under \Cref{asp:identical r and F and posted price} by $\ratiofb \deq \frac{\utipfb}{\utippost}$ and $\ratiosb \deq \frac{\utipsb}{\utippost}$, where benchmarks $\utipfb,\utipsb$ follow the same definitions in \Cref{sec:impact}.} 
Interestingly, as shown in \Cref{thm:regular bounds under ap,thm:mhr bounds under ap}, even when the intermediary is restricted to the space of anonymous pricing mechanisms, the (worst-case) \podm and \poa are the same as the setting where the intermediary can design any DSIC, ex-post IR mechanisms.
The proofs of both theorems are similar to their analogy in \Cref{sec:impact} and thus are deferred into \Cref{apx:regular bounds under ap proof,apx:mhr bounds under ap proof}, respectively.

\begin{restatable}{theorem}{thmRegularBoundUnderAP}
\label{thm:regular bounds under ap}
    Under \Cref{asp:identical r and F and posted price}, if \contribution distributions $\welfaredis$ are regular with support $\supp{\welfaredis} = \parents{\lowsupp, \upsupp}$ $\parent{0 < \lowsupp \leq \upsupp \leq \rewardfun}$, if $1 \leq \upsupp/\lowsupp \leq \ratiolu$, the tight \podm is $\bigtheta{\ratiolu}{\log \ratiolu}$ and the tight \poa is $\bigtheta{\ratiolu}{(\log \ratiolu)^{2}}$.
\end{restatable}

\begin{restatable}{theorem}{thmMHRBoundUnderAP}
\label{thm:mhr bounds under ap}
   Under \Cref{asp:identical r and F and posted price}, if \contribution distributions $\welfaredis$ are MHR, the tight \podm is in $[\ee, \apmhrupbound]$ and the tight \poa is $\ee^{2}$. Moreover, if $\agentnum=1$, the tight \podm is $\ee$.
\end{restatable}

\section{Unknown Market Size}
\label{sec:robust}

In practice, the principal may lack perfect knowledge of the number of agents or incur high costs to acquire it.
Motivated by this, we introduce and explore the \emph{linear contract design for unknown market size}, where the principal knows only a range for the number of agents, i.e., $\agentnum \in [\lownum:\upnum]$, rather than its exact value.
Our goal is to determine how the principal in our three-party model should design linear contracts to safeguard her utility in this robust setting.  
To evaluate the principal's utility under a given linear contract, we continue using the first-best and second-best benchmarks and redefine the price of anarchy and the price of double marginalization as follows:
\begin{align*}
    \ratiofb(\lownum,\upnum) = \min\nolimits_{\contract \in [0, 1]} \max\nolimits_{\agentnum\in[\lownum:\upnum]} \frac{\utipfb(\agentnum)}{\utip(\contract, \agentnum)}, &&
    \ratiosb(\lownum,\upnum) = \min\nolimits_{\contract \in [0, 1]} \max\nolimits_{\agentnum\in[\lownum:\upnum]} \frac{\utipsb(\agentnum)}{\utip(\contract, \agentnum)},
\end{align*}
where for the number of agents $\agentnum$, $\utipfb(\agentnum)$ (resp.\ $\utipsb(\agentnum)$) is the first-best benchmark (resp.\ the second-best benchmark) and $\utip(\contract, \agentnum)$ is the principal's utility under linear contract $\contract$.

In the above definition, both ratios are defined in the worst case over the range $[\lownum:\upnum]$ of possible agent numbers.  
To ensure that these ratios have upper bounds related to the market size range, $[\lownum:\upnum]$, we propose two contract design schemes under regular and MHR \contribution distributions, respectively.  
Our results show that when the market size is unknown, both ratios suffer a multiplicative loss compared to the bounds established for known market sizes in \Cref{sec:impact}.
Under regular distributions, the loss factor is $\frac{\upnum}{\lownum}$ (\Cref{thm:robust regular bounds}), whereas under MHR distributions, it reduces to $\frac{\log \upnum}{\log \lownum}$ (\Cref{thm:robust MHR bounds}). 
Notably, the former bound arises from optimizing the contract for any market size in $[\lownum:\upnum]$, while the latter is uniquely attained by optimizing for the minimal size $\lownum$---optimizing for any other size fails to provide this guarantee.
Moreover, we provide an example illustrating that if the principal has no knowledge of the market size—i.e., if the number of agents ranges from $1$ to infinity—then, regardless of the contract design, the gap between her utility and these benchmarks cannot be upper bounded (\Cref{thm:robust negative}).

\begin{restatable}{theorem}{robustbounds}\label{thm:robust regular bounds}
    For agents with identical expected rewards $\rewardfun$ and regular \contribution distributions $\welfaredis$ with support $\supp{\welfaredis}=[\lowsupp, \upsupp]$ $(0 \leq \lowsupp \leq \upsupp \leq \rewardfun)$, if $1 \leq \upsupp/\lowsupp \leq \ratiolu$ and the principal is only aware of the range of the number of agents $\agentnum\in[\lownum:\upnum]$, then 
    $ \ratiosb(\lownum,\upnum) = \frac{\upnum}{\lownum} \cdot \bigO{\ratiolu}{\log \ratiolu}$, and $\ratiofb(\lownum,\upnum) = \frac{\upnum}{\lownum} \cdot \bigO{\ratiolu}{(\log \ratiolu)^{2}}$.
\end{restatable}

The intuition behind this theorem is that when $\agentnum \in [\lownum:\upnum]$, if the principal designs the contract as if the market size were $\upnum$, then the worst-case upper bounds for the two ratios increase by a multiplicative factor of $\bceil{\upnum/\lownum}$ compared to the case where the principal knows the exact value of $\agentnum$. 
Moreover, when $\lownum = \upnum$, meaning the principal knows the exact market size, this result aligns with \Cref{thm:regular bounds} and \Cref{thm:mhr ratio bounds for n in naturals}.
The detailed proof of \Cref{thm:robust regular bounds} can be found in \Cref{apx:robust regular bounds proof}.

\begin{restatable}{theorem}{mhrrobustbounds}\label{thm:robust MHR bounds}
    For agents with identical expected rewards $\rewardfun$ and MHR \contribution distributions $\welfaredis$, if the principal is only aware of the range of the number of agents $\agentnum\in[\lownum:\upnum]$, then
        \begin{align*}
            \ratiosb(\lownum, \upnum) 
            \leq \ratiofb(\lownum, \upnum) 
            \leq \frac{\ln \parent{\frac{1}{2} (1 - \welfaredis(0)) \upnum + 1}}{\ln \parent{\frac{1}{2} (1 - \welfaredis(0)) \lownum + 1}} \cdot \ee^{2} + \bigO{\lownum,\upnum}{\frac{\log \log \lownum \cdot \log \upnum}{(\log \lownum)^{2}}}.
        \end{align*}
\end{restatable}
Under MHR \contribution distributions, \Cref{thm:robust MHR bounds} provides better results than \Cref{thm:robust regular bounds}.
Meanwhile, when $\lownum=\upnum$, the result of this theorem is also consistent with the results of \Cref{thm:mhr ratio bounds for n in naturals}.
The detailed proof of \Cref{thm:robust MHR bounds} can be found in \Cref{apx:robust MHR bounds proof}.

We remark that the ratios in the above two theorems tend to infinity when $\upnum$ approaches infinity and $\lownum$ remains constant.
Intuitively, this captures the scenario where the principal has almost no knowledge of the range of $\agentnum$. 
In \Cref{thm:robust negative}, we show that in this case, no matter how the principal designs the contract, the two ratios cannot be upper bounded.
The detailed proof of \Cref{thm:robust negative} can be found in \Cref{apx:robust negative proof}.

\begin{restatable}{theorem}{thmRobustNegative}\label{thm:robust negative}
    For any small constant $\varepsilon > 0$ and agents with identical expected rewards $\rewardfun$ and MHR \contribution distributions $\welfaredis$, if the principal is only aware of the range of the number of agents $\agentnum\in[1: \infty]$, no linear contract $\contract$ satisfies that both $\ratiofb(1, \infty)$ and $\ratiosb(1, \infty)$ have upper bounds.
\end{restatable}

\section{Conclusions and Future Work} \label{sec:conclusion}

This paper presents a three-party model that captures the strategic and contractual interactions in crowdsourcing markets, where a principal delegates task execution through a profit-sharing contract with an intermediary who, in turn, selects an agent via a mechanism. 
Modeling this interaction as an extensive-form Stackelberg game, we characterize the subgame perfect equilibrium and show that the intermediary's optimal mechanism corresponds to a Bayesian revenue-optimal mechanism parameterized by the principal's contract.
This characterization enables the principal's contract design problem to be reformulated in a more explicit and tractable manner and reveals that linear contracts are optimal even when the task has multiple outcomes and agents have asymmetric
cost distributions, provided their outcome distributions are identical, substantially simplifying the contract design space.
Moreover, when agents also share identical cost distributions, we reduce the principal's problem to a virtual value pricing formulation and derive tight or nearly tight bounds on the price of double marginalization (PoDM) and the price of anarchy (PoA), thereby quantifying the principal's utility loss due to delegation and information asymmetry under regular and monotone hazard rate (MHR) assumptions.
We further demonstrate that our results remain robust when the intermediary is restricted to simpler mechanisms such as anonymous pricing, and when the principal has only partial information about market size.

Several promising directions remain for future research.
One natural extension is to study non-myopic intermediaries who seek to maximize cumulative utility over repeated interactions.
Another is to consider a reversal in the Stackelberg order, where the intermediary first selects a mechanism class and the principal subsequently responds with a contract---potentially altering the equilibrium structure.
It would also be valuable to explore more complex market structures involving multiple layers of intermediation, such as hierarchical or networked delegation.
Finally, extending the PoDM and PoA analysis to heterogeneous agent populations---particularly when both reward and cost distributions vary---would broaden the applicability of our framework to more realistic environments.

\bibliography{mybibfile}

\appendix

\section{Unobservable Reward}\label{apx:unobservable reward}
Our results hold even when the reward profile $\boldsymbol{\otc} = (\otc_{1}, \otc_{2}, \ldots, \otc_{\otcnum})$ is unobservable to the intermediary.
For convenience, a contract $\ctrctprf = (\ctrct_{1}, \ctrct_{2}, \ldots, \ctrct_{\otcnum})$ in our paper is defined as a ratio, implying that the intermediary's payment is $\ctrct_{j} \otc_{j}$ when outcome $j \in [\otcnum]$ occurs. In standard literature, contracts typically specify the outcome-based payment directly (i.e., $\ctrct_{j} \otc_{j}$). Thus, given a contract, the intermediary already knows the payment for every possible outcome. This information is sufficient for the intermediary to design the optimal mechanism, and therefore all our results continue to hold without requiring observability of $\boldsymbol{\otc}$.

\section{Two \Zjs with Different Mechanism Spaces}\label{apx:diffrent mech space}

In our model, the principal's utility depends critically on the constraints faced by the \zj in her mechanism design spaces.
A natural question arises: consider two \zjs with different mechanism design spaces.
If one's mechanism space strictly contains the other's, does the principal always benefit from the more flexible \zj?
We show that this is not universally true by offering two examples, illustrating that the inclusion-relation of mechanism design spaces for the two \zjs does not imply the order of the principal's optimal utility.

\begin{example}
    \Zj $A$ can choose all mechanisms while \zj $B$ can choose a single mechanism that never assigns the task to agents. 
    In this example, the principal has higher utility facing $A$ than $B$.
\end{example}

\begin{example}
    \Zj $A$ can choose all mechanisms while \zj $B$ can choose a single mechanism: assigning the task to the agent with the lowest cost and paying her the second lowest cost among the other agents. 
    In this example, the principal has higher utility facing $B$ than~$A$.
\end{example}

The intuition behind this is that the utility of the \zj depends both on the allocation and payment of the mechanism, but only the allocation directly affects the utility of the principal.
Therefore, the \zj that optimize her own utility within a smaller mechanism space may design allocation rules that are either more or less favorable to the principal.

\section{Missing Proofs}
\label{apx:missing proofs}

This section provides all proofs omitted from the main text.

\subsection{Proof of Proposition~\ref{thm:opt mech}}
\label{apx:opt mech proof}
In this section, we will prove \Cref{thm:opt mech}.

\thmOptMech*

An immediate implication from the above result is as follows.

\begin{corollary}\label{cor:opt mech contract invariance}
    If two contracts $\ctrctprf, \ctrctprfvar \in \nnreals^\otcnum$ yield the same expected payment for each agent, i.e., $\inner{\ctrctprf}{\rwdprf_{i}} = \inner{\ctrctprfvar}{\rwdprf_{i}}$ for all $i \in \agents$, then their corresponding optimal mechanisms, $\mech_{\ctrctprf}^{\star}$ and $\mech_{\ctrctprfvar}^{\star}$, are identical. Moreover, the principal's utility is identical under both contracts.
\end{corollary}

At a high level, to prove \Cref{thm:opt mech}, we first show that any truthful mechanism between the \zj and $\agentnum$ agents can be converted into a truthful single-item auction between a seller and $\agentnum$ buyers, and vice versa, as described in \Cref{lemma:convert mechanism}.
The key is to convert agent $i$'s private cost $\cost_{i}$ into buyer $i$'s private value $\val_{i}$ in the auction.
The contract $\ctrctprf$ plays a pivotal role in this conversion: $\val_{i} = \inner{\ctrctprf}{\rwdprf_{i}} - \cost_{i}$, which is the \cwelfare of agent $i$.
Thus, the seller's revenue in the auction meets the \zj's utility in the mechanism.
Then, we establish that the regularity or MHR property of the \contribution distribution is preserved in the distribution of the \cwelfare for any contract $\ctrctprf$, as shown in \Cref{lemma:valuedis welfaredis share regular/mhr}.
This crucial property enables us to apply standard auction theory, demonstrating that the optimal mechanism for the \zj is indeed the virtual welfare maximizer mechanism as follows.

Consider the multi-buyer scenario where a seller wants to sell a single item to one of $\agentnum$ buyers, whose private values of the item, $\set{\val_{i}}_{i \in\agents}$, are drawn independently from regular distributions $\set{\valdis_{i}}_{i \in\agents}$.
\citet{Mye-81} proved that the \emph{virtual welfare maximizer} mechanism is DSIC, ex-post IR, and maximizes the seller's expected revenue.

\begin{definition}[Virtual welfare maximizer, \citealp{Mye-81}]\label{def:vwm}
    Let the value $\val_{i}$ of each buyer $i \in \agents$ be drawn from a distribution $\valdis_{i}$.
    The \emph{virtual welfare maximizer} is defined as follows:
    \begin{enumerate}
        \item Transform the \textup{(}truthfully reported\textup{)} values $\set{\val_{i}}_{i \in \agents}$ into the corresponding virtual values $\set{\virval_{i}(\val_{i})}_{i \in \agents}$, where $\virval_{i}(\cdot)$ is the virtual value function of distribution $\valdis_{i}$.
\item {If the buyers' virtual values are all negative, the seller withholds the item; otherwise, allocate the item to the buyer with the maximum non-negative virtual value}, i.e.,
        \begin{align*}
            \allocrule^{\dag}(\valprf) 
            \deq \argmax\limits_{\Vert \allocrule \Vert_{1} \leq 1} \sum\limits_{i \in \agents} \virval_{i}(\val_{i}) \cdot \alloc_{i}.
        \end{align*}
\item Define the payment $\paymentrule^{\dag}$ from the buyers to the seller such that,  for each $i \in \agents$,
        \begin{align*}
            \payment_{i}^{\dag}(\valprf) \deq 
            \virval_{i}^{-1} \parent{\max\nolimits_{j \in \agents \setminus \set{i}} \plus{\virval_{j}\parent{\val_{j}}}}
            \cdot \alloc^{\dag}_{i}(\valprf),
\end{align*}
        where $\alloc^{\dag}_{i}$ is the $i$th component of $\allocrule^{\dag}$, representing the probability of agent $i$ being the winner.
    \end{enumerate}
\end{definition}

\begin{lemma}[\citealp{Mye-81}]\label{lemma: myerson optimal auction}
    Assuming the value $\val_{i}$ of each buyer $i \in \agents$ is drawn independently from a regular distribution $\valdis_{i}$, the virtual welfare maximizer mechanism is a DSIC and ex-post IR mechanism that maximizes the expected revenue.
\end{lemma}

The allocation rule of the virtual welfare maximizer mechanism ensures that the auctioned item is only assigned to buyers whose bids meet or exceed their respective monopoly reserve prices.
In other words, as long as there exists at least one such buyer, the item will be successfully auctioned.

To prove \Cref{thm:opt mech}, we first present how, when the contract is given by the principal, the mechanism between the \zj and the $\agentnum$ agents can be converted into a single-item auction between a seller and $\agentnum$ buyers, and vice versa.

\begin{lemma} \label[lemma]{lemma:convert mechanism}
    Given a contract $\ctrctprf \in \nnreals^\otcnum$, consider an auction between a seller with a single item and $\agentnum$ buyers, where each buyer $i$'s private value of the item is equal to her \cwelfare, i.e., $\val_{i} = \inner{\ctrctprf}{\rwdprf_{i}} - \cost_{i}$.
    Then, the following hold:
    \begin{enumerate}
        \item For any truthful mechanism $\mech_{\ctrctprf} = \parent{\allocrule_{\ctrctprf}, \paymentrule_{\ctrctprf}}$ of the \zj, we can construct an auction $\mech_{\ctrctprf}^{\dag} = (\allocrule^{\dag}_{\ctrctprf}, \paymentrule^{\dag}_{\ctrctprf})$ such that 
        \begin{align*}
            \allocrule^{\dag}_{\ctrctprf}(\valprf) = \allocrule_{\ctrctprf}(\costprf),
            \quad && \quad
            \payment^{\dag}_{\ctrctprf, i}(\valprf) = \inner{\ctrctprf}{\rwdprf_{i}} \cdot \alloc_{\ctrctprf, i}(\costprf) - \payment_{\ctrctprf, i}(\costprf).
        \end{align*}
        Then we have: (i) the auction $\mech_{\ctrctprf}^{\dag}$ is truthful; (ii) the revenue of the seller in $\mech_{\ctrctprf}^{\dag}$ is equal to the utility of the \zj in $\mech_{\ctrctprf}$.
        \item For any truthful auction $\mech_{\ctrctprf}^{\dag} = (\allocrule^{\dag}_{\ctrctprf}, \paymentrule^{\dag}_{\ctrctprf})$ of the seller, we can construct a mechanism $\mech_{\ctrctprf} = \parent{\allocrule_{\ctrctprf}, \paymentrule_{\ctrctprf}}$ of the \zj such that
        \begin{align*}
            \allocrule_{\ctrctprf}(\costprf) = \allocrule^{\dag}_{\ctrctprf}(\valprf),
            \quad && \quad
            \payment_{\ctrctprf, i}(\costprf) = \inner{\ctrctprf}{\rwdprf_{i}} \cdot \alloc^{\dag}_{\ctrctprf, i}(\valprf) - \payment^{\dag}_{\ctrctprf, i}(\valprf).
        \end{align*} 
        Then we have: (i) the mechanism $\mech_{\ctrctprf}$ is truthful; (ii) the utility of the \zj in $\mech_{\ctrctprf}$ is equal to the revenue of the seller in $\mech_{\ctrctprf}^{\dag}$.
        \item The distribution of $\val_{i}$ is $\valdis_{\ctrctprf, i}$, then its CDF is 
        \begin{align*}
            \valdis_{\ctrctprf, i}(\temp) = \welfaredis_{i}(\inner{\ctrctone - \ctrctprf}{\rwdprf_{i}} + \temp).
        \end{align*}
    \end{enumerate}
\end{lemma}

\begin{proof}
    Let $\mech_{\ctrctprf} = (\allocrule_{\ctrctprf}, \paymentrule_{\ctrctprf})$ be a mechanism implemented by \zj.
    For any cost profile $\costprf$, let $i^{\star}$ be the winner of $\mech_{\ctrctprf}(\costprf)$ such that her allocation and payment are $\alloc_{\ctrctprf, i^{\star}}(\costprf)$ and $\payment_{\ctrctprf, i^{\star}}(\costprf)$, respectively.
    Recall that the utility of the \zj and winner are given by $\inner{\ctrctprf}{\rwdprf_{i^{\star}}} - \payment_{\ctrctprf, i^{\star}}(\costprf)$ and $\payment_{\ctrctprf, i^{\star}}(\costprf) - \cost_{i^\star}$, respectively.

    Next, consider a single-item auction $\mech_{\ctrctprf}^{\dag} = (\allocrule_{\ctrctprf}^{\dag}, \paymentrule_{\ctrctprf}^{\dag})$ involving $\agentnum$ buyers.
    Let buyer $i$'s private value of the item be $\val_{i} = \inner{\ctrctprf}{\rwdprf_{i}} - \cost_{i}$, which is agent $i$'s \cwelfare under contract $\ctrctprf$ in $\mech_{\ctrctprf}$.
    Let the allocation rule in the auction be $\allocrule_{\ctrctprf}^{\dag}(\valprf) = \allocrule_{\ctrctprf}(\costprf)$, ensuring that both mechanisms share the same winner $i^{\star}$.
    Let the payment from winning buyer $i^{\star}$ to the seller be $\payment_{\ctrctprf, i^{\star}}^{\dag}(\valprf) = \inner{\ctrctprf}{\rwdprf_{i^{\star}}} - \payment_{\ctrctprf, i^{\star}}(\costprf)$.

    Buyer $i$'s utility in the auction is given by $(\inner{\ctrctprf}{\rwdprf_{i}} - \cost_{i}) - (\inner{\ctrctprf}{\rwdprf_{i}} - \payment_{\ctrctprf, i}(\costprf)) = \payment_{\ctrctprf, i}(\costprf) - \cost_{i}$ if agent $i$ is the winner in the mechanism, and $0$ otherwise.
    This utility is identical to agent $i$'s utility in $\mech_{\ctrctprf}$.
    Additionally, since the mechanism $\mech_{\ctrctprf}$ and the auction $\mech_{\ctrctprf}^{\dag}$ share the same winner, the truthfulness of $\mech_{\ctrctprf}$ implies the truthfulness of the auction $\mech_{\ctrctprf}^{\dag}$, and vice versa.
    Then, we can observe that the seller's revenue, given by the payment $\ctrctprf \rwdprf_{i^{\star}} - \payment_{\ctrctprf, i^{\star}}(\costprf)$ from the winner $i^{\star}$, is equal to the utility of the \zj in $\mech_{\ctrctprf}$.

    Finally, we characterize the distributions of the transformed private values $\set{\val_{i}}_{i = 1}^{\agentnum}$, which are equal to the \cwelfare distributions $\set{\valdis_{\ctrctprf, i}}_{i = 1}^{\agentnum}$.
    Recall that the \contribution $\welfare_{i} = \rwdprf_{i} - \cost_{i}$ of each agent $i$ independently follows the distribution $\welfaredis_{i}$.
    Since $\val_{i} = \inner{\ctrctprf}{\rwdprf_{i}} - \cost_{i} = \inner{\ctrctprf - \ctrctone}{\rwdprf_{i}} + \welfare_{i}$, the CDF of the value $\val_{i}$ can be derived as $\valdis_{\ctrctprf, i}(\temp) = \welfaredis_{i}(\inner{\ctrctone - \ctrctprf}{\rwdprf_{i}} \rwdprf_{i}+ \temp)$.
\end{proof}

Note that the \cwelfare $\cwf_{\ctrctprf, i} = \inner{\ctrctprf}{\rwdprf_{i}} - \cost_{i}$ can be negative.
Also, note that the \contribution of agent $i$ is exactly her \cwelfare $\cwf_{\ctrctprf, i}$ by setting $\ctrctprf = \ctrctone$.
We have assumed regularity of the \contribution distribution in \Cref{asp:welfare regular}.
Next, we show that this assumption implies the regularity of the \cwelfare distribution for any $\ctrctprf$.

\begin{lemma}\label[lemma]{lemma:valuedis welfaredis share regular/mhr}
    Given a contract $\ctrctprf \in \nnreals^\otcnum$ and agent $i \in \agents$, the \cwelfare distribution $\valdis_{\ctrctprf, i}$ is regular \textup{(}resp., MHR\textup{)} if and only if the \contribution distribution $\welfaredis_{i}$ is regular \textup{(}resp., MHR\textup{)}.
\end{lemma}

\begin{proof}
    For ease of presentation, we remove the index $i$ from the notations in the proof.
    Specifically, we rewrite the \contribution distribution $\welfaredis_{i}$ as $\welfaredis$, \cwelfare distribution $\valdis_{\ctrctprf, i}$ as $\valdis_{\ctrctprf}$, and reward $\rwdprf_{i}$ as $\rwdprf$.

    By definition, the virtual value function and the hazard rate function of $\welfaredis$, denoted as $\hr$ and $\virval$ respectively, are given by
    \begin{align*}
        \hr(\temp) = \frac{\welfaredens(\temp)}{1 - \welfaredis(\temp)}, &&
        \virval(\temp) = \temp - \frac{1 - \welfaredis(\temp)}{\welfaredens(\temp)}.
    \end{align*}
    The hazard rate function of the \cwelfare distribution $\valdis_{\ctrctprf}$, denoted as $\hr_{\ctrctprf}$, is
    \begin{align*}
        \hr_{\ctrctprf}(\temp) 
        = \frac{\valdens_{\ctrctprf}(\temp)}{1 - \valdis_{\ctrctprf}(\temp)} 
        = \frac{\welfaredens(\inner{\ctrctone - \ctrctprf}{\rwdprf_{i}} \rwdprf + \temp)}{1 - \welfaredis(\inner{\ctrctone - \ctrctprf}{\rwdprf_{i}} \rwdprf + \temp)},
    \end{align*}
    where the last equality follows from \Cref{lemma:convert mechanism}.
    The virtual value function of the \cwelfare distribution $\valdis_{\ctrctprf}$, denoted as $\virval_{\ctrctprf}$, is
    \begin{align*}
        \virval_{\ctrctprf}(\temp) 
        = \temp - \frac{1}{\hr_{\valdis_{\ctrctprf}}(\temp)} 
        = \temp - \frac{1 - \welfaredis(\inner{\ctrctone - \ctrctprf}{\rwdprf_{i}} \rwdprf + \temp)}{\welfaredens(\inner{\ctrctone - \ctrctprf}{\rwdprf_{i}} \rwdprf + \temp)}.
    \end{align*}
    Since the reward $\rwdprf$ is constant, for a given contract $\ctrctprf$, $\virval_{\ctrctprf}(\temp)$ (resp., $\hr_{\ctrctprf}(\temp)$) shares the same monotonicity as $\virval(\temp)$ (resp., $\hr(\temp)$).
\end{proof}

Now we are ready to prove \Cref{thm:opt mech}.

\begin{proof}[Proof of \Cref{thm:opt mech}]
    By \Cref{lemma:convert mechanism}, we can convert any truthful mechanism $\mech_{\ctrctprf}$ between the \zj and $n$ agents with private cost profile $\costprf$ into a truthful single-item auction $\mech_{\ctrctprf}^{\dag}$ between a seller and $\agentnum$ buyers with private value profile $\cwfprf_{\ctrctprf}$, and vice versa.
    Since the buyer's revenue in the auction equals the utility of the \zj, we only need to find the optimal auction that maximizes the seller's revenue.

    By \Cref{lemma:valuedis welfaredis share regular/mhr}, the regularity property of the \contribution distribution is preserved in the \cwelfare distribution.
    Since agents' costs are independently drawn from the regular distributions $\costdisprf$, for a given contract $\ctrctprf$, the corresponding values are also independently drawn from the regular distributions $\valdisprf_{\ctrctprf}$.
    Therefore, by \Cref{lemma: myerson optimal auction}, the optimal mechanism for the \zj is the virtual welfare maximizer mechanism.
    Then by \Cref{def:vwm}, we finish the proof.
\end{proof} 
\subsection{Proof of Theorem~\ref{thm:opt contract}}
\label{apx:opt contract proof}

In this section, we will prove \Cref{thm:opt contract}.

\thmOptContract*

To prove \Cref{thm:opt contract}, we first analyze the continuity of the monopoly reserve as follows.

\begin{lemma}[Continuity of monopoly reserve]
\label{lem:reserve continuous}
    For any fixed agent $i \in \agents$, consider the monopoly reserve price $\reserve_{\contract, i}$ as a function of the linear contract $\contract$.
    Then, $\reserve_{\contract, i}$ is continuous in w.r.t., $\contract\in[0, 1]$.
\end{lemma}

\begin{proof}
    For ease of presentation, we remove the index $i$ from the notations in the proof.
    Specifically, we rewrite the \contribution distribution $\welfaredis_{i}$ as $\welfaredis$, \cwelfare distribution $\valdis_{\contract, i}$ as $\valdis_{\contract}$, cost distribution $\costdis_{i}$ as $\costdis$, cost $\cost_{i}$ as $\cost$, reserve price $\reserve_{\contract, i}$ as $\reserve_{\contract}$, and reward $\rwdprf_{i}$ as $\rwdprf$.
    
    By the definition of reserve price, we have
    \begin{equation}
    \label{eq_rp_D}
        \reserve_{\contract}
        - \frac{1 - \valdis_{\contract}(\reserve_{\contract})}{\valdens_{\contract}(\reserve_{\contract})}
        = 0~.
    \end{equation}
    Recall that for each agent, its cost $\cost$ follows $\costdis$, and its \cwelfare $\welfare_{\contract} = \contract\rwdprf - \cost$ follows $\valdis_{\contract}$.
    This gives the relation between two distribution functions: $\valdis_{\contract}(\temp) = \costdis(- \temp + \contract\rwdprf)$ for any $\temp \in \supp{\costdis}$.
    
    Define two functions $\costfun(\contract) \deq - \reserve_{\contract} + \contract\rwdprf$ and $\costdisfun(\temp) \deq \temp + \frac{\costdis(\temp)}{\costdens(\temp)}$.
    By replacing $\reserve_{\contract}$ in \cref{eq_rp_D} with $-\costfun(\contract) + \contract\rwdprf$, we get that for any $\contract\in [0, 1]$, it holds
    \begin{align*}
        \costfun(\contract) + \frac{\costdis(\costfun(\contract))}{\costdens(\costfun(\contract))} = \contract\rwdprf,
        \quad&&\quad
        \costdisfun(\costfun(\contract)) = \contract\rwdprf.
    \end{align*} 
    Clearly, $\costdisfun(\costfun(\cdot))$ is continuous in $[0, 1]$.
    By \Cref{asp:welfare regular}, $\costdisfun(\cdot)$ is continuous in $\supp{\costdis}$. 
    Hence, we have $\costfun(\cdot)$ is continuous in $[0, 1]$.
    Since $\reserve_{\contract} = - \costfun(\contract) + \contract\rwdprf$, we finally obtain that $\reserve_{\contract}$ is also continuous in $[0, 1]$.
    This completes our proof.
\end{proof}

Now, we are ready to prove \Cref{thm:opt contract}.
\begin{proof}[Proof of \Cref{thm:opt contract}.]
    Recall that the principal's expected utility is $\inner{\ctrctone - \ctrctprf}{\rwdprf_{i}}$ if agent $i \in \agents$ is the winner of the \zj's mechanism.
    The principal's optimal contract design problem is defined as follows
    \begin{align} \label{eq:opt contract}
        \max\limits_{\ctrctprf \in \nnreals^{\otcnum}} \sum \nolimits_{i \in \agents} \inner{\ctrctone - \ctrctprf}{\rwdprf_{i}} \cdot \prob[\costprf \sim \costdisprf]{i \text{ is the winner of } \mech_{\ctrctprf}^{\star}(\costprf)}.
    \end{align}
    where $\mech_{\ctrctprf}^{\star}$ denotes the optimal mechanism of the \zj given contract $\ctrctprf$.

    By the characterization in \Cref{thm:opt mech}, the \zj's optimal mechanism is the virtual welfare maximizer mechanism with the value equal to the \cwelfare $\cwf_{\contract, i}$.
    According to the allocation rule of the virtual welfare maximizer mechanism, an agent is the winner if and only if her virtual value is non-negative and the highest among all agents.
    Therefore, the probability that the task is auctioned to agent $i$ is $\prob[\cwfprf_{\ctrctprf} \sim \valdisprf_{\ctrctprf}]{\virval_{\ctrctprf, i}(\cwf_{\ctrctprf, i}) \geq \max\nolimits_{j \in \agents}\plus{\virval_{\ctrctprf, j}(\cwf_{\ctrctprf, j})}}$. 
    Consequently, the principal's problem of maximizing her utility in \cref{eq:opt contract} can be rewritten as
    \begin{equation*} 
        \max_{\ctrctprf \in \nnreals^{\otcnum}} \sum_{i \in \agents} \inner{\ctrctone - \ctrctprf}{\rwdprf_{i}} \cdot \prob[\cwfprf_{\ctrctprf} \sim \valdisprf_{\ctrctprf}]{\virval_{\ctrctprf, i}(\cwf_{\ctrctprf, i}) \geq \max\nolimits_{j \in \agents}\plus{\virval_{\ctrctprf, j}(\cwf_{\ctrctprf, j})}}.
        \label{eq:opt U general}
    \end{equation*}
    Thus, we have proved \cref{thm: general} of \Cref{thm:opt contract}.

Next, we consider the case where all agents have the identical expected reward profiles $\rwdprf_{i} \equiv \rwdprf$.
    Let $\ctrctprfopt$ be the optimal contract that achieves the maximum utility for the principal in \cref{eq:opt contract}.
    Define $\ctrctlnropt \deq \frac{\inner{\ctrctprfopt}{\rwdprf}}{\exrwd}$, and we have $\inner{\ctrctprfopt}{\rwdprf} = \ctrctlnropt \exrwd$.
    Clearly, the optimality of $\ctrctprfopt$ ensures that $\inner{\ctrctprfopt}{\rwdprf} \leq \exrwd$, which implies that $\ctrctlnropt\in [0,1]$.
    By \Cref{cor:opt mech contract invariance}, the principal's utility remains unchanged if we replace the contract $\ctrctprfopt$ with the linear contract $\ctrctlnropt \cdot \ctrctone$.

    Consequently, we only need to focus on the linear contract $\ctrctlnr \cdot \ctrctone \in [0, 1]^{\otcnum}$.
    Define $\temp_{\ctrctlnr, i} \deq \inner{\ctrctone - \ctrctlnr \cdot \ctrctone}{\rwdprf} + \reserve_{\ctrctlnr, i} = (1 - \ctrctlnr) \exrwd + \reserve_{\ctrctlnr, i}$ for each $i \in \agents$, and the principal's problem of maximizing her utility in \cref{eq:opt U general} can be rewritten as
    \begin{align*}
        &{}\max\limits_{\ctrctlnr \in [0, 1]} \sum_{i \in \agents} (\inner{\ctrctone - \ctrctlnr \cdot \ctrctone}{\rwdprf}) \cdot \prob[\cwfprf_{\ctrctlnr} \sim \valdisprf_{\ctrctlnr}]{\virval_{\ctrctlnr, i}(\cwf_{\ctrctlnr, i}) \geq \max_{j \in \agents} \plus{\virval_{\ctrctlnr, j}(\cwf_{\ctrctlnr, j})}} 
        \nonumber\\
        \overset{(a)}{=}{}& \max\limits_{\ctrctlnr \in [0, 1]} \sum_{i \in \agents} (1 - \ctrctlnr) \exrwd \cdot \prob[\cwfprf_{\ctrctlnr} \sim \valdisprf_{\ctrctlnr}]{\cwf_{\ctrctlnr, i} \geq \reserve_{\ctrctlnr, i} \wedge \virval_{\ctrctlnr, i}(\cwf_{\ctrctlnr, i}) \geq \max_{j \in \agents}\virval_{\ctrctlnr, j}(\cwf_{\ctrctlnr, j})}
        \nonumber\\
        \overset{(b)}{=}{}& \max\limits_{\ctrctlnr \in [0, 1]} \sum_{i \in \agents} (1 - \ctrctlnr) \exrwd \cdot \prob[\welfareprf \sim \welfaredisprf]{\welfare_{i} \geq (1 - \ctrctlnr) \exrwd + \reserve_{\ctrctlnr, i} \wedge \virval_{i}(\welfare_{i}) \geq \max_{j \in \agents}\virval_{j}(\welfare_{j})} 
        \nonumber\\
        \overset{(c)}{=}{}& \max\limits_{\ctrctlnr \in [0, 1]} \sum_{i \in \agents} \virval_{i}(\temp_{\ctrctlnr, i}) \cdot \prob[\welfareprf \sim \welfaredisprf]{\welfare_{i} \geq \temp_{\ctrctlnr, i} \wedge \virval_{i}(\welfare_{i}) \geq \max_{j \in \agents}\virval_{j}(\welfare_{j})} 
        \nonumber\\
        \overset{(d)}{=}{}& \max\limits_{\ctrctlnr \in [0, 1]} \sum_{i \in \agents} \virval_{i}(\temp_{\ctrctlnr, i}) \cdot \prob[\welfareprf \sim \welfaredisprf]{\welfare_{i} \geq \temp_{\ctrctlnr, i} \wedge \virval_{i}(\welfare_{i}) = \max_{j \in \agents}\virval_{j}(\welfare_{j})}. 
        \label{eq:opt U identical reward}
    \end{align*}
    Equality~(a) holds since the value distribution $\valdis_{\ctrctlnr, i}$ is regular by \Cref{lemma:valuedis welfaredis share regular/mhr} and \Cref{asp:welfare regular}. 
    Equality~(b) holds because $\cwf_{\ctrctlnr, i} = (\ctrctlnr - 1) \exrwd + \welfare_{i}$ and the following fact:
    \begin{align*}
        \virval_{\ctrctlnr, i}(\cwf_{\ctrctlnr, i})
        =
        \cwf_{\ctrctlnr, i} - \frac{1 - \valdis_{\ctrctlnr, i}(\cwf_{\ctrctlnr, i})}{\valdens_{\ctrctlnr, i}(\cwf_{\ctrctlnr, i})}
        =
        (\ctrctlnr - 1) \exrwd + \welfare_{i} - \frac{1 - \welfaredis_{i}(\welfare_{i})}{\welfaredens_{i}(\welfare_{i})}
        = (\ctrctlnr - 1) \exrwd+\virval_{i}(\welfare_{i}),
    \end{align*}
    where the first equality is by the definition of $\virval_{\ctrctlnr, i}(\cwf_{\ctrctlnr, i})$, 
    and the second equality holds since $\welfaredis_{\ctrctlnr, i}(\cwf_{\ctrctlnr, i}) = \welfaredis_{i}(\welfare_{i})$
    and the last equality is by the definition of $\virval_{i}(\welfare_{i})$.
    Hence, for any $i, j \in \agents$, we have $\virval_{\ctrctlnr, i}(\cwf_{\ctrctlnr, i}) \geq \virval_{\ctrctlnr, j}(\cwf_{\ctrctlnr, j})$ if and only if $\virval_{i}(\welfare_{i}) \geq \virval_{j}(\welfare_{j})$.
    Equality~(c) holds as follows
    \begin{align*}
        \virval_{i}(\temp_{\ctrctlnr, i})
        = \temp_{\ctrctlnr, i} - \frac{1 - \welfaredis_{i}(\temp_{\ctrctlnr, i})}{\welfaredens_{i}(\temp_{\ctrctlnr, i})}
        = (1 - \ctrctlnr) \exrwd + \reserve_{\ctrctlnr, i} - \frac{1 - \valdis_{\ctrctlnr, i}(\reserve_{\ctrctlnr, i})}{\valdens_{\ctrctlnr, i}(\reserve_{\ctrctlnr, i})}
        = (1 - \ctrctlnr) \exrwd,
    \end{align*}
    where the first equality holds by definition, the second equality holds because $\temp_{\ctrctlnr, i} = \ctrctlnropt \exrwd + \reserve_{\ctrctlnr, i}$, $\valdis_{\ctrctlnr, i}(\reserve_{\ctrctlnr, i}) = \welfaredis_{i}(\temp_{\ctrctlnr, i})$ and $\valdens_{\ctrctlnr, i}(\reserve_{\ctrctlnr, i}) = \welfaredens_{i}(\temp_{\ctrctlnr, i})$, and the last equality is by the the fact that the reserve price $\reserve_{\ctrctlnr, i}$ is the zero point of virtual value function $\virval_{\ctrctlnr, i}(\cdot)$. 
    Equality~(d) holds since by \Cref{asp:welfare regular}, agent $i$'s \contribution distribution $\welfaredis_{i}$ is regular.
    Hence, we have \cref{thm: idential rwd} of \Cref{thm:opt contract}. 

At last, we proceed to prove \cref{thm: identical rwd and dist} of \Cref{thm:opt contract}. In this case, all agents have the identical expected reward profiles $\rwdprf_{i} \equiv \rwdprf$ and the identical \contribution distributions $\welfaredis_{i} \equiv \welfaredis$. 
    Note that identical expected reward profiles and \contribution distributions imply that all agents have the identical contracted \contribution distribution $\welfaredis_{\ctrctlnr,i}\equiv \welfaredis_{\ctrctlnr}$.
    Thus, for each agent $i \in \agents$, we have $\temp_{\ctrctlnr, i} = \temp_{\ctrctlnr} \deq \ctrctlnropt \exrwd + \reserve_{\ctrctlnr}$, where $\reserve_{\ctrctlnr}$ is the monopoly reserve price of the contracted \contribution distribution $\welfaredis_{\ctrctlnr}$. 
    In this case, the principal's problem of maximizing her utility can be rewritten as
    \begin{align*}
        {}&\max\limits_{\ctrctlnr\in[0, 1]} 
        \sum_{i \in \agents}
        \virval(\temp_{\ctrctlnr, i}) \cdot
        \prob[\welfareprf \sim \welfaredis^{\otimes \agentnum}]{\welfare_{i} \geq \temp_{\ctrctlnr, i} \wedge \virval(\welfare_{i}) = \max_{j \in \agents}\virval(\welfare_{j})} \\
        ={}&
        \max\limits_{\ctrctlnr\in[0, 1]} \virval(\temp_{\ctrctlnr}) \cdot
        \sum_{i \in \agents}
        \prob[\welfareprf \sim \welfaredis^{\otimes \agentnum}]{\welfare_{i} \geq \temp_{\ctrctlnr} \wedge \virval(\welfare_{i}) = \max_{j \in \agents}\virval(\welfare_{j})} \\
        \overset{(a)}{=}{}&
        \max\limits_{\ctrctlnr\in[0, 1]} \virval(\temp_{\ctrctlnr}) \cdot
        \sum_{i \in \agents}
        \prob[\welfareprf \sim \welfaredis^{\otimes \agentnum}]{\welfare_{i} \geq \max_{j \in \agents}\left\{\temp_{\ctrctlnr},
        \welfare_{j}\right\}} \\
        \overset{(b)}{=}{}&
        \max\limits_{\ctrctlnr \in[0, 1]} \virval(\temp_{\ctrctlnr})(1 - \welfaredis^{\agentnum}(\temp_{\ctrctlnr})) \\
        \overset{(c)}{=}{}&
        \max\limits_{\temp \in \supp{\welfaredis}} \virval(\temp) (1 - \welfaredis^{\agentnum}(\temp)),
    \end{align*}
    Equality~(a) holds according to the monotonicity of the virtual value function of the \contribution distribution.
    Equality~(b) holds by the following fact.
    \begin{align*}
        \sum_{i \in \agents}
        \prob[\welfareprf \sim \welfaredis^{\otimes \agentnum}]{\welfare_{i} \geq \max_{j \in \agents} \left\{\temp_{\ctrctlnr},
        \welfare_{j}\right\}}
        {}&=
        \sum_{i \in \agents}
        \prob[\welfareprf \sim \welfaredis^{\otimes \agentnum}]{\welfare_{i}\ge \max_{j\in \agents}\welfare_j \wedge \welfare_{i}\geq \temp_{\ctrctlnr}}\\
        {}&=
        \sum_{i \in \agents}
        \prob[\welfareprf \sim \welfaredis^{\otimes \agentnum}]{\welfaremax = \welfare_{i} \wedge \welfaremax\geq \temp_{\ctrctlnr}}\\
        {}&= 
        \prob[\welfareprf \sim \welfaredis^{\otimes \agentnum}]{\welfaremax \geq \temp_{\ctrctlnr}}\\
        {}&= 
        1 - \welfaredis^{\agentnum}(\temp_{\ctrctlnr}),
    \end{align*}
    where the first equality holds by the definition of $\max$ operator, the second equality holds by the definition of $\welfaremax$, and the third equality holds by the law of total probability.
    
    Finally, we prove equality~(c).
    By \Cref{lem:reserve continuous}, the function $\reserve_{\ctrctlnr}$ (w.r.t.\ $\ctrctlnr$) is continuous in $[0, 1]$.
    This implies the function $\temp_{\ctrctlnr}$ (w.r.t.\ $\ctrctlnr$) is also continuous in $[0, 1]$.
    Notice that $\temp_{0} = \reserve$ and $\temp_{1} = \exrwd$, where $\reserve$ is the monopoly reserve price of \contribution distribution $\welfaredis$.
    We have $[\reserve, \exrwd] \subseteq \{ \temp_{\ctrctlnr} : \ctrctlnr \in [0, 1]\} \subseteq \supp{\welfaredis}$.
    Hence, equality~(c) holds with ``$\leq$''.
    On the flip hand, since the \contribution distribution $\welfaredis$ is regular, by the definition of reserve price $\reserve$, we have $\virval(\temp_{\ctrctlnr})\le 0$ if $\temp_{\ctrctlnr} \le \reserve$.
    Thus, $\virval(\temp_{\ctrctlnr})(1 - \welfaredis^{\agentnum}(\temp_{\ctrctlnr}))$ only takes non-negative values if $\temp_{\ctrctlnr} \in [\reserve, \exrwd]$, leading that equality~(c) holds with ``$\geq$''.
    Consequently, we have proved equality~(c).
\end{proof}

\subsection{Proof of Theorem~\ref{thm:regular bounds}}
\label{apx:regular bounds proof}

In this section, we will prove \Cref{thm:regular bounds}.

\regularbounds*

\begin{example}\label{example:regular bounds}     
    We consider a single agent, i.e., $\agentnum = 1$.
    Let $K$ be a sufficiently large integer. 
    We construct a \contribution distribution $\welfaredis$ for the agent with its CDF given by
    \begin{align*}
        \welfaredis(\temp)=\left\{
            \begin{aligned}
                &0,&&\temp \in (-\infty, 1); \\
&1 - \frac{1}{\temp},&& \temp \in [1, \temp_{K});\\
                &1 - \frac{k-2+2^{1 - k}}{K(\temp-2^{K - k + 1})},&&\temp \in [\temp_{k}, \temp_{k - 1}); \\
                &1,&&\temp \in [\temp_0, +\infty),
            \end{aligned}
        \right.
    \end{align*}
    where $\temp_{0} = \frac{2^{K}}{\ln 2}$ and $\temp_{k} = \frac{k 2^{K}}{2^{k} - 1}$ for $k \in [K]$.
    It can be verified that the above distribution is regular.
\end{example}

\begin{lemma}[\citealp{FJ-24}, Theorems 1 and 2]\label{lemma:regular mhr order statistic}
    For any i.i.d.\ random variables $\val_{i}~(i \in \agents)$, the following statements hold.
    \begin{enumerate}
        \item If the variables $\val_{i}$ satisfy a regular distribution $\welfaredis$, the first-order statistic $\valmax$ also satisfies a regular distribution $\welfaredis^{\agentnum}$;
        \item If the variables $\val_{i}$ satisfy an MHR distribution $\welfaredis$, the first-order statistic $\valmax$ also satisfies an MHR distribution $\welfaredis^{\agentnum}$.
    \end{enumerate}
\end{lemma}

\begin{lemma}\label{lemma:regular reserve increase}
    Let $\virval^{(\agentnum)}(\temp)$ and $\reserve^{(\agentnum)}$ be the virtual value function and the monopoly reserve price of the distribution $\welfaredis^{\agentnum}$, respectively.
    For any i.i.d.\ random variable $\val_{i}~(i \in \agents)$, if the first order statistic $\valmax$ satisfies a regular distribution $\welfaredis^{\agentnum}$, the virtual value function is weakly decreasing and the monopoly reserve price is weakly increasing with respect to $\agentnum$, i.e., for every $\temp \in \supp{\welfaredis}$, it holds
    \begin{align*}
        \virval^{(\agentnum)}(\temp) \geq \virval^{(\agentnum + 1)}(\temp),
        \quad&&\quad
        \reserve^{(\agentnum + 1)} \geq \reserve^{(\agentnum)}.
    \end{align*}
\end{lemma}

\begin{proof}
    We first observe that the virtual value function is weakly decreasing, which implies that the monopoly reserve price is weakly increasing.
    Indeed, if $\virval^{(\agentnum)}(\temp) \geq \virval^{(\agentnum + 1)}(\temp)$ holds for any $\temp \in \supp{\welfaredis}$, then by the definition of the reserve price, we have
    \begin{equation*}
        \virval^{(\agentnum + 1)}\parent{\reserve^{(\agentnum + 1)}} 
        = 0 
        = \virval^{(\agentnum)}\parent{\reserve^{(\agentnum)}} 
        \geq \virval^{(\agentnum + 1)}\parent{\reserve^{(\agentnum + 1)}}.
    \end{equation*}
    Given that $\welfaredis^{\agentnum}$ is regular, we also notice that $\virval^{(\agentnum)}(\temp)$ is non-decreasing with respect to $\temp$.
    This implies that
    \begin{equation*}
         \reserve^{(\agentnum + 1)} \geq \reserve^{(\agentnum)}.
    \end{equation*}

    Now, we show that the virtual value function is weakly decreasing.
    In this proof, we may abbreviate $\welfaredis(\temp)$ as $\welfaredis$ for convenience.
    Besides, we denote the hazard rate function of $\welfaredis^{\agentnum}$ as $\hr^{(\agentnum)}$.
    
    Recall that the virtual value function $\virval(\temp) = \temp - \frac{1}{\hr(\temp)}$.
    Then, we have
    \begin{align*}
        \virval^{(\agentnum)}(\temp)
        = \temp - \frac{1}{\hr^{(\agentnum)}}=\temp - \frac{1 - \welfaredis^{\agentnum}}{\agentnum\welfaredis^{\agentnum - 1} \welfaredens}.
    \end{align*}
    
    We consider $\frac{1}{\hr^{(\agentnum})}$ as a function of $\agentnum$. 
    By differentiating, we have
    \begin{align*}
        \DDX{\agentnum}{\parent{\frac{1}{\hr^{(\agentnum)}}}}
        = \frac{-\agentnum \welfaredis^{\agentnum - 1} \ln \welfaredis - (1 - \welfaredis^{\agentnum})(1 + \agentnum \ln \welfaredis)}{\agentnum^{2} \welfaredis^{\agentnum - 1} \welfaredens}
        = \frac{\welfaredis^{\agentnum} - \agentnum \ln\welfaredis - 1}{\agentnum^{2}\welfaredis^{\agentnum - 1} \welfaredens}
        = \frac{\int_{0}^{\agentnum} (\welfaredis^{\agentnum} - 1) \ln\welfaredis \cdot \dd \agentnum}{\agentnum^{2} \welfaredis^{\agentnum - 1} \welfaredens}
    \end{align*}
    Since $\welfaredis$ is a CDF, we know $\welfaredens \geq 0$, $\welfaredis \leq 1$ and $\ln\welfaredis \leq 0$, which indicates that
    \begin{align*}
        \DDX{\agentnum}{\parent{\frac{1}{\hr^{(\agentnum)}}}}
        = \frac{\int_{0}^{\agentnum} (\welfaredis^{\agentnum} - 1) \ln\welfaredis \cdot \dd \agentnum}{\agentnum^{2} \welfaredis^{\agentnum - 1} \welfaredens}
        \geq 0.
    \end{align*}
    It follows that $\virval^{(\agentnum)}(\temp)$ is weakly decreasing with respect to $\agentnum$, which means that
    \begin{align*}
        \virval^{(\agentnum)}(\temp) \geq \virval^{(\agentnum + 1)}(\temp).
    \end{align*}
    This completes our proof.
\end{proof}

\begin{lemma}[\citealp{Yan-11}, Theorem 4.1] \label{lemma:regular OPA vs AP}
    To sell a single item among multiple buyers with i.i.d.\ and regular distributions, the revenue of anonymous pricing is an $\ratiooa$-approximation to that of the Bayesian revenue-optimal mechanism, where $\ratiooa = \frac{\ee}{\ee - 1}$.
\end{lemma}

\begin{definition}[The truncated equal-revenue distribution] \label{def:teqr distribution}
    Let $\revenuecurve^{*}$ be the monopoly revenue of the \contribution distribution $\welfaredis$ with the support $\supp{\welfaredis} = [\lowsupp, \upsupp]$.
    The truncated equal-revenue distribution \textup{(}w.r.t.\ $\welfaredis$\textup{)} is defined as the distribution $\welfaredis^{\dag}$ with support $\parents{\revenuecurve^{*}, \upsupp}$ and its CDF given by
    \begin{align*}
        \welfaredis^{\dag}(\temp) = \left\{
            \begin{aligned}
                &0, &&\temp \in \left(-\infty, \revenuecurve^{*}\right]; \\
                &1 - \frac{\revenuecurve^{*}}{\temp}, &&\temp \in \left(\revenuecurve^{*}, \upsupp\right); \\
                &1, &&\temp \in \left[\upsupp, +\infty\right).
            \end{aligned}
        \right.
    \end{align*}
\end{definition}

One can verify that the monopoly revenue of the truncated equal-revenue distribution $\welfaredis^{\dag}$ is also $\revenuecurve^{*}$, and $\welfaredis^{\dag}$ is also regular since the revenue curve of $\welfaredis^{\dag}$ is concave. 
Moreover, we notice that the CDF is discontinuous at $\upsupp$, which corresponds to a point-mass of $\revenuecurve^{*}/\upsupp$ at value $\upsupp$.

We next show that this truncated equal-revenue distribution gives a worst-case instance for the ratio between the revenue of anonymous pricing and the first-best benchmark.

\begin{lemma}\label{lemma:regular fb worst case}
    Among regular distributions with support $\parents{\lowsupp, \upsupp}$ $\parent{0 \leq \lowsupp \leq \upsupp \leq \rewardfun}$ and monopoly revenue $\revenuecurve^{*} \in [\lowsupp, \upsupp]$, the truncated equal-revenue distribution with monopoly revenue equal to $\revenuecurve^{*}$ achieves the largest first-best benchmark for a single agent.
\end{lemma}

\begin{proof}
    Let $\welfaredis$ be an arbitrary regular distribution with support $\parents{\lowsupp, \upsupp}$ and monopoly revenue $\revenuecurve^{*} \in [\lowsupp, \upsupp]$.
    Observe that the revenue curve of the truncated equal-revenue distribution $\welfaredis^{\dag}$ is $\revenuecurve^{\dag}(q) = \min \set{\upsupp q, \revenuecurve^{*}}$.
    Consider the revenue curve of distribution $\welfaredis$, given by $\revenuecurve(\quantile) = \temp \quantile$, where $\quantile = 1 - \welfaredis(\temp)$.
    On the one hand, we know $R(q) \leq \upsupp q$ since $\temp \in [\lowsupp, \upsupp]$; on the other hand, we have $\revenuecurve(q) \leq \revenuecurve^{*}$ by the definition of $\revenuecurve^{*}$.
    It follows that the revenue curve of distribution $\welfaredis$ is always positioned below that of the truncated equal-revenue distribution $\welfaredis^{\dag}$, i.e., $\revenuecurve(\quantile) \leq \revenuecurve^{\dag}(\quantile)$ for all $q \in [0, 1]$.
Consequently, since both $\revenuecurve(\cdot)$ and $\revenuecurve^{\dag}(\cdot)$ are concave, for any $\temp \in [\lowsupp, \upsupp]$, we have
    \begin{align*}
        1 - \welfaredis(\temp) \leq 1 - \welfaredis^{\dag}(\temp).
    \end{align*}
    
    Recall that the first-best benchmark for a single agent with \contribution distribution $\welfaredis$ is given by
    \begin{align*}
        \utipfb=\int_{0}^{+\infty} (1 - \welfaredis(\temp)) \cdot \dd \temp 
        \leq \int_{0}^{+\infty} (1 - \welfaredis^{\dag}(\temp)) \cdot \dd \temp 
        = \revenuecurve^{*} \cdot \parent{1 + \ln \frac{\upsupp}{\revenuecurve^{*}}}.
    \end{align*}
    Thus, we conclude that the largest first-best benchmark for one agent is achieved by the truncated equal-revenue distribution.
\end{proof}

Now, we are ready to prove \Cref{thm:regular bounds}.

\begin{proof}[Proof of \Cref{thm:regular bounds}.]
    We first show the upper bound parts utilizing the above lemmas, and then show the lower bound part by analyzing \Cref{example:regular bounds}.

    \xhdr{Upper bounding \podm.}
    We first observe that $\welfaredis^{\agentnum}$ is also regular by \Cref{lemma:regular mhr order statistic}.
    Thus, we have $\virval^{(\agentnum)}(\temp)\geq\virval^{(\agentnum + 1)}(\temp)$ for any $\temp$ in $\supp{\welfaredis}$ by \Cref{lemma:regular reserve increase}.
    Then we can obtain that
    \begin{align}\label{equation: U* >= phi^(n)(z)pr}
        \utip^{*}
        = \max\limits_{\temp\in\supp{\welfaredis}} \virval(\temp)\parent{1 - \welfaredis^{\agentnum}(\temp)}
        \geq \max\limits_{\temp\in\supp{\welfaredis}} \virval^{(\agentnum)}\parent{\temp)(1 - \welfaredis^{\agentnum}(\temp)}.
    \end{align}

    Next, let quantile $\quantile = 1 - \welfaredis^{\agentnum}(\temp)$ and $\revenuecurve$ be the revenue curve of $\welfaredis^{\agentnum}$.
    Since the virtual value of $\welfaredis^{\agentnum}$ is equal to the slope of the revenue curve $\revenuecurve(\quantile)$ at quantile $\quantile$.
    We can further obtain that
    \begin{align} \label{equation: U* >= qR'(q)}
        \utip^{*}
        \geq \max\limits_{\quantile\in[0,1]} \quantile \cdot \revenuecurve'(\quantile).
    \end{align}

    Now, we provide two arguments for the cases where $4 \leq \upsupp/\lowsupp \leq \ratiolu$ and $1 \leq \upsupp/\lowsupp \leq 4$ separately.
    
    We first consider the case where $4 \leq \upsupp/\lowsupp \leq \ratiolu$.
    
    Let $K = 2\bceil{\log \ratiolu}$.
    Divide the quantile space $[0, 1]$ of $\welfaredis^{\agentnum}$ into $K + 1$ pieces by introducing the following points: $1 > \quantile_{0} > \quantile_{1} > \cdots > \quantile_{K} = 0$.
    We set the slope of the revenue curve (the virtual value) $\revenuecurve'(\quantile_{0}) = 0$ and $\revenuecurve'(\quantile_{k}) = \lowsupp \cdot 2^{k - 2\bceil{\log \ratiolu} + \log \ratiolu}$ for every index $k \in [K]$.
    Note that $\revenuecurve'(\quantile_{K}) = \upsupp$.
    One can verify that this construction holds for any regular distribution $\welfaredis^{\agentnum}$ since the virtual value function is non-decreasing.
    Consequently, for every $\quantile \in [\quantile_{k}, \quantile_{k - 1}]$ with $k \in [2 : K]$, it holds
    \begin{align}\label{equation:R' in k k - 1}
        \revenuecurve'(\quantile_{k - 1}) 
        \leq \revenuecurve'(\quantile)
        \leq \revenuecurve'(\quantile_{k}) 
        = 2\revenuecurve'(\quantile_{k - 1}),
    \end{align}
    and for every $\quantile\in[\quantile_{1}, \quantile_{0}]$, it holds
    \begin{align}\label{equation:R' in 1 0}
        0 = \revenuecurve'(\quantile_{0}) 
        \leq \revenuecurve'(\quantile) 
        \leq \revenuecurve'(\quantile_{1}) 
        = 2 \lowsupp \cdot 2^{\log \ratiolu - 2\bceil{\log \ratiolu}} 
        \leq \frac{2 \lowsupp}{\ratiolu}.
    \end{align}
    Since the support of $\welfaredis$ is $[\lowsupp, \upsupp]$, we know $\revenuecurve(1) \geq \lowsupp$.
    Thus, we know that monopoly revenue $\revenuecurve^{*} = \revenuecurve(\quantile_{0}) \geq \revenuecurve(1) \geq \lowsupp$.
    Moreover, since $\revenuecurve(\quantile_{0}) = 0$, we have
    \begin{align*}
        \revenuecurve^{*}
        &={} \int_{0}^{\quantile_{0}}\revenuecurve'(\quantile) \cdot \dd \quantile \\
        &={} \sum_{k \in [K]}\int_{\quantile_{k}}^{\quantile_{k - 1}}\revenuecurve'(\quantile) \cdot \dd \quantile \\
        &\overset{(a)}{\leq}{} \sum_{k \in [2 : K]}(\quantile_{k - 1} - \quantile_{k}) \cdot 2\revenuecurve'(\quantile_{k - 1}) + \int_{\quantile_{1}}^{\quantile_{0}}\revenuecurve'(\quantile) \cdot \dd \quantile \\
        &\leq{} 2 \cdot \sum_{k\in[2 : K]}\quantile_{k - 1} \cdot \revenuecurve'(\quantile_{k - 1}) + \int_{\quantile_{1}}^{\quantile_{0}}\revenuecurve'(\quantile) \cdot \dd \quantile \\
        &\leq{} 2 \cdot \sum_{k\in[2 : K]}\quantile_{k - 1} \cdot \revenuecurve'(\quantile_{k - 1}) + \int_{0}^{1}\revenuecurve'(\quantile_{1}) \cdot \dd \quantile \\
        &\overset{(b)}{\leq}{} 2 \cdot \sum_{k\in[2 : K]}\quantile_{k - 1} \cdot \revenuecurve'(\quantile_{k - 1}) + \frac{2 \lowsupp}{\ratiolu}.
    \end{align*}
    where inequality~(a) holds by \cref{equation:R' in k k - 1} and inequality~(b) holds by \cref{equation:R' in 1 0}.
    It follows that
    \begin{align*}
        \sum_{k\in[2 : K]}\quantile_{k - 1} \cdot \revenuecurve'(\quantile_{k - 1}) 
        \geq \frac{1}{2} \cdot \left(\revenuecurve^{*} - \frac{2 \lowsupp}{\ratiolu} \right),
    \end{align*}
    which implies that there exists $k^{\dag} \in [2:K]$ such that
    \begin{equation*}
        \quantile_{k^{\dag}}\cdot\revenuecurve'(\quantile_{k^{\dag}})
        \geq \frac{1}{2K} \cdot \left(\revenuecurve^{*} - \frac{2 \lowsupp}{\ratiolu} \right)
        = \frac{1}{4\bceil{\log \ratiolu}} \cdot \left(\revenuecurve^{*} - \frac{2 \lowsupp}{\ratiolu} \right)
        = \frac{\revenuecurve^{*}}{8\bceil{\log \ratiolu}} + \frac{1}{8\bceil{\log \ratiolu}} \cdot \left(\revenuecurve^{*} - \frac{4 \lowsupp}{\ratiolu} \right).
    \end{equation*}
    
    Notice that $\revenuecurve^{*}\geq\lowsupp$ and $\ratiolu \geq \upsupp/\lowsupp \geq 4$, we get $\revenuecurve^{*} - 4\lowsupp/\ratiolu \geq 0$.
    As a result, we obtain that
    \begin{align*}
        \quantile_{k^{\dag}}\cdot\revenuecurve'(\quantile_{k^{\dag}})
        \geq \frac{\revenuecurve^{*}}{8\bceil{\log \ratiolu}}.
    \end{align*}

    Note that by \Cref{lemma:regular OPA vs AP}, the second-best benchmark satisfies $\revenuecurve^{*} \leq \utipsb \leq \ratiooa\cdot\revenuecurve^{*}$. 
    Combine \cref{equation: U* >= qR'(q)}, we can obtain that
    \begin{align*}
        \utip^{*}
        \geq \max\limits_{\quantile\in[0,1]} \quantile \cdot \revenuecurve'(\quantile) 
        \geq \quantile_{k^{\dag}} \cdot \revenuecurve'(\quantile_{k^{\dag}}) 
        \geq \frac{\revenuecurve^{*}}{8\bceil{\log \ratiolu}} 
        \geq \frac{\utipsb}{8\ratiooa\bceil{\log \ratiolu}}.
    \end{align*}
    Therefore, we have $\ratiosb = \frac{\utipsb}{\utip^{*}} \leq\bigO{\ratiolu}{\log \ratiolu}$ if $4 \leq \upsupp/\lowsupp \leq \ratiolu$. 
    
    Finally, we consider the case where $1 \leq \upsupp/\lowsupp \leq 4$.
    
    We let $\quantile^{\dag}$ be the largest quantile such that $\revenuecurve'(\quantile^{\dag}) \geq \lowsupp / 2$ and $\monoquan$ be the monopoly quantile such that $\revenuecurve^{*} = \revenuecurve\parent{\monoquan}$.
    Notice $\quantile^{\dag}$ is well-defined since $\revenuecurve'(0) = \upsupp$.
    By definition, we have $0 \leq \quantile^{\dag} \leq \monoquan \leq 1$ and $\revenuecurve'(\quantile^{\dag}) \leq \lowsupp$.
    Besides, if $\revenuecurve'(\quantile^{\dag}) \geq \lowsupp / 2$, it must hold $\quantile^{\dag} = \quantile^{*} = 1$.
    
    On the one hand, if $\revenuecurve'(\quantile^{\dag}) = \lowsupp / 2$, by definition, we have
    \begin{align*}
        \lowsupp
        \leq \revenuecurve^{*} 
        = \int_{0}^{\monoquan} \revenuecurve'(\quantile) \cdot \dd \quantile 
        = \int_{0}^{\quantile^{\dag}} \revenuecurve'(\quantile) \cdot \dd \quantile + \int_{\quantile^{\dag}}^{\monoquan} \revenuecurve'(\quantile) \cdot \dd \quantile 
        \overset{(a)}{\leq} \quantile^{\dag} \cdot \upsupp + \revenuecurve'(\quantile^{\dag}) \cdot \parent{1 - \quantile^{\dag}}.
    \end{align*}
    where inequality(a) holds since $\revenuecurve'(\quantile)$ is non-increasing due the regularity of $\welfaredis^{\agentnum}$.
    Moreover, notice that $\revenuecurve^{*} \geq \lowsupp$.
    Thus, we can obtain that
    \begin{align*}
        \quantile^{\dag} 
        \geq \frac{\lowsupp - \revenuecurve'(\quantile^{\dag})}{\upsupp - \revenuecurve'(\quantile^{\dag})}
        \overset{(a)}{\geq}\frac{\lowsupp - \revenuecurve'(\quantile^{\dag})}{\upsupp}
        \overset{(b)}{=} \frac{\lowsupp - \frac{\lowsupp}{2}}{\upsupp}
        = \frac{\lowsupp}{2 \upsupp} 
        \geq \frac{1}{8}.
    \end{align*}
    where inequality~(a) and equality~(b) hold since $\revenuecurve'(\quantile^{\dag}) = \lowsupp / 2 \geq 0$.
    
    Consequently, we can obtain that
    \begin{align*}
        \utip^{*}
        \geq \max\limits_{\quantile\in[0,1]} \quantile \cdot \revenuecurve'(\quantile) 
        \geq \quantile^{\dag} \cdot \revenuecurve'(\quantile^{\dag})
        \geq \frac{\revenuecurve'(\quantile^{\dag})}{8}
        \overset{(a)}{\geq} \frac{\lowsupp}{16}
        \overset{(b)}{\geq} \frac{\revenuecurve^{*}}{64}
        \geq \frac{\utipsb}{64\ratiooa}.
    \end{align*}
    where inequality~(a) holds by $\revenuecurve'(\quantile^{\dag}) \geq \lowsupp / 2$ and inequality~(b) holds by $\upsupp/\lowsupp \leq 4$ and $\revenuecurve^{*} \leq \upsupp$.

    On the other hand, if $\revenuecurve'(\quantile^{\dag}) > \lowsupp / 2$, then we know $\quantile^{\dag} = 1$.
    We finally obtain that 
    \begin{align*}
        \utip^{*}
        \geq \max\limits_{\quantile\in[0,1]} \quantile \cdot \revenuecurve'(\quantile) 
        \geq \quantile^{\dag} \cdot \revenuecurve'(\quantile^{\dag})
        \geq \revenuecurve'(\quantile^{\dag})
        \geq \frac{\lowsupp}{2}
        \geq \frac{\revenuecurve^{*}}{8}
        \geq \frac{\utipsb}{8\ratiooa}.
    \end{align*}
    
    Therefore, we can also obtain that $\ratiosb = \frac{\utipsb}{\utip^{*}} \leq \bigO{\ratiolu}{\log \ratiolu}$ if $1 \leq \upsupp/\lowsupp \leq 4$.

    \xhdr{Upper bounding \poa.}
    We next prove the upper bound of \poa.
    
    By \Cref{lemma:regular fb worst case}, we have that for any monopoly revenue $\revenuecurve^{*} \geq 0$, the largest first-best benchmark for one agent is achieved by the truncated equal-revenue distribution with monopoly revenue equal to $\revenuecurve^{*}$. 
    We can obtain that
    \begin{align*}
        \utipfb
        = \int_{0}^{+\infty} \parent{1 - \welfaredis^{\agentnum}(\temp)} \cdot \dd \temp 
        \leq \revenuecurve^{*} \cdot \parent{1 + \ln \frac{\upsupp}{\revenuecurve^{*}}} 
        \overset{(a)}{\leq} \revenuecurve^{*} \cdot (1 + \ln \ratiolu),
    \end{align*}
    where inequality~(a) holds since $\revenuecurve^{*} \geq \lowsupp$.
    Therefore, we have
    \begin{align*}
        \frac{\utipfb}{\utipsb}
        \leq \frac{\utipfb}{\revenuecurve^{*}}
        \leq \frac{\revenuecurve^{*} \cdot (1 + \ln \ratiolu)}{\revenuecurve^{*}}
        = \bigO{\ratiolu}{\log \ratiolu}.
    \end{align*}
    
    As a result, we can obtain that for any regular distribution,
    \begin{align*}
        \ratiofb
        = \ratiosb\cdot\frac{\utipfb}{\utipsb}
        \leq \bigO{\ratiolu}{\log \ratiolu}\cdot\bigO{\ratiolu}{\log \ratiolu}
        \leq \bigO{\ratiolu}{(\log \ratiolu)^{2}}.
    \end{align*}

    \xhdr{Lower bounding \podm and \poa.}
    We now analyze \Cref{example:regular bounds}, which shows the lower bound part of the two ratios, i.e., $\ratiofb = \bigomega{\ratiolu}{\parent{\log \ratiolu}^{2}}$ and $\ratiosb = \bigomega{\ratiolu}{\log \ratiolu}$.

    By definition, the virtual value for $\temp \in [\temp_{k}, \temp_{k - 1})$ is
    \begin{align*}
        \virval(\temp)=\revenuecurve'(\quantile)=2^{K - k + 1},
    \end{align*}
    and the quantile $\quantile_{k} = 1 - \welfaredis(\temp_{k}) = \quantile_{k - 1} + \frac{2^{k - 1}}{K2^{K}}$ and $\quantile_{0} = 0$, which leads to $\quantile_{k}=\frac{2^{k}-1}{K 2^{K}}$ for every $k \in [K]$.
    
    As a sanity check, $\revenuecurve(\quantile_{k}) = \frac{k}{K}$, and the virtual value $\virval(\temp)$ is negative for $\temp \in (0, \temp_{k}]$. Moreover, we can obtain that
    \begin{align*}
    \utip^{*}=\max_{\quantile\in[0, 1]} \quantile \cdot \revenuecurve'(\quantile)
    \overset{(a)}{=} \max_{k\in[K]} \quantile_{k} \cdot 2^{K - k + 1}
    = \max_{k\in[K]} \frac{2-2^{1 - k}}{K}
    = \frac{2-2^{1 - K}}{K}.
    \end{align*}
    where equality~(a) holds since $\quantile \cdot \revenuecurve'(\quantile) = \quantile \cdot 2^{K - k + 1} \leq \quantile_{k} \cdot 2^{K - k + 1}$ for any $\quantile \in [\quantile_{k-1}, \quantile_{k}]$.
    
    In this example, we have the second-best benchmark $\utipsb = \revenuecurve(1) = 1$.
    Consequently, we can obtain that $\ratiosb=\frac{\utipsb}{\utip^{*}}=\bigomega{\ratiolu}{K}$.

    Next, we consider the \poa.
    We have
    \begin{align*}
        \utipfb 
        ={}& \int_{0}^{\temp_{0}} 1 - \welfaredis(\temp) \cdot \dd \temp \\
        \geq{}& \sum_{k \in [2:K]} \int_{\temp_{k}}^{\temp_{k-1}} 1 - \welfaredis(\temp) \cdot \dd \temp \\
        \overset{(a)}{=}{}& \sum_{k \in [2:K]} \frac{k-2+2^{1 - k}}{K} \cdot \ln \parent{2 + \frac{1}{2^{k-1} - 1}} \\
        \overset{(b)}{\geq}{}& \sum_{2 \in [K]} \frac{k-2}{K} \cdot \ln 2 \\
        \overset{(c)}{=}{}& \frac{K^{2} - 3K + 3}{2K} \cdot \ln 2,
    \end{align*}
    where equality~(a) and (b) hold by calculations, and inequality~(c) holds since $2^{1 - k} \geq 0$ and $1/\parent{2^{k - 1} - 1} \geq 0$.

    Therefore, we can obtain that $\ratiofb = \frac{\utipfb}{\utip^{*}} = \bigomega{\ratiolu}{K^{2}}$. Since $\ratiolu = \temp_{0} = \frac{2^{K}}{\ln 2}$, we have $\ratiofb = \bigomega{\ratiolu}{\parent{\log \ratiolu}^{2}}$ and $\ratiosb = \bigomega{\ratiolu}{\log \ratiolu}$.
\end{proof}
 
\subsection{Proof of Theorem~\ref{thm:mhr ratio bounds for n in naturals}}
\label{apx:mhr ratio bounds proof}

\mhrratiobounds*

\begin{example}\label{example:mhr bounds}
    Consider $\agentnum$ agents.
    Let $K$ be a real number larger than $1$ and denote a non-positive \contribution $\tilde{\welfare} = -\ln (K - (K - 1)\ee^{-\rewardfun})$.
    We construct a \contribution distribution $\welfare$ for the agent as
    \begin{equation*}
        \welfaredis_{\rewardfun}(\temp) = \left\{
            \begin{aligned}
                &0,&&\temp \in \left(-\infty, \tilde{\welfare} \right); \\
                &\frac{1}{K} \cdot \frac{1 - \ee^{-\temp}}{1 - \ee^{-\rewardfun}} + \parent{1 - \frac{1}{K}}, &&\temp \in \left[\tilde{\welfare}, \rewardfun \right]; \\
                &1,&&\temp \in (\rewardfun, +\infty).
            \end{aligned}
        \right.
    \end{equation*}
We know that $\welfaredis_{\rewardfun}$ is an MHR distribution with $\welfaredis(\tilde{\welfare}) = 0$ and $\welfaredis(0) = 1 - 1/K$.
\end{example}

\begin{lemma}[\citealp{AGM-09}, Lemma 1] \label{lemma:mhr prob bound}
    For any random variable $\val$ that follows an MHR distribution $\welfaredis$ with support in $[0, +\infty)$, $1 -  \welfaredis(\reserve) \geq \frac{1}{\ee}$, where $\reserve$ is the monopoly reserve price of the distribution $\welfaredis$.
\end{lemma}

It is important to note that \Cref{lemma:mhr prob bound} holds only when the support of the distribution is non-negative. 
We next extend this conclusion to cases with arbitrary support sets.

\begin{corollary} \label{corollary:mhr prob bound}
    For any random variable $\val$ that follows an MHR distribution $\welfaredis$, $1 - \welfaredis(\reserve) \geq \frac{1}{\ee} \cdot \prob{\val \geq0}$ where $\reserve$ is the reserve price for the distribution $\welfaredis$.
\end{corollary}

\begin{proof}
    Recall that the hazard rate function of distribution $\welfaredis$ is $\hr(\temp)$.
    By definition, we have $1 - \welfaredis(\temp) = e^{-\int_{-\infty}^{\temp}\hr(\temp) \cdot \dd \temp}$.
    By definition of reserve price, we have $\hr(\reserve)=\frac{1}{\reserve}$.
    Therefore, we can obtain that
    \begin{align*}
        1 - \welfaredis(\reserve)
        ={}& \ee^{-\int_{-\infty}^{\reserve} \hr(\temp) \cdot \dd \temp} \\
        ={}& \ee^{-\int_{-\infty}^{0} \hr(\temp) \cdot \dd \temp - \int_{0}^{\reserve} \hr(\temp) \cdot \dd \temp} \\
        \overset{(a)}{\geq}{}& \ee^{-\int_{-\infty}^{0} \hr(\temp) \cdot \dd \temp - \int_{0}^{\reserve} \frac{1}{\reserve} \cdot \dd \temp} \\
        \overset{(b)}{=}{}& \frac{1}{\ee} (1 - \welfaredis(0)) \\
        ={}& \frac{1}{\ee} \prob{\val \geq 0}.
    \end{align*}
    where inequality~(a) holds since $\welfaredis$ has MHR, we have $\hr(\temp) \leq \frac{1}{\reserve}$ for any $\temp\leq\reserve$ and equality~(b) holds since $\int_{0}^{\reserve} \frac{1}{\reserve} \cdot \dd \temp = 1$ and $\ee^{-\int_{-\infty}^{0} \hr(\temp) \cdot \dd \temp} = 1 - \welfaredis(0)$.
\end{proof}

A consequence of \Cref{corollary:mhr prob bound} is that the principal's utility under a linear contract $\contract$ can be lower-bounded as follows.

\begin{lemma} \label{lemma:repeated proof}
    For agents with identical reward $\rewardfun$ and MHR distribution $\welfaredis$, given linear contract $\contract \in [0, 1]$, the principal's utility, denoted by $\utip(\contract)$, satisfies that
    \begin{align*}
        \utip(\contract)
        \geq \max\nolimits_{\tempnum \in \agents}\rewardfun (1 - \contract) \parent{1 -  \parent{1 - \frac{1}{\ee}(1 - \welfaredis^{\tempnum}(\rewardfun (1 - \contract)))}^{\agentnum/\tempnum}}.
    \end{align*}
\end{lemma}

\begin{proof}
    In this proof, we specially use $\valmax^{(\tempnum)}$ to denote the first-order statistic out of $\tempnum \in \mathbb{N}$ i.i.d.\ variables drawing from a distribution $\valdis_{\contract}$.
    By \Cref{thm:opt contract}, we have
    \begin{align*}
        \utip(\contract) 
        = \virval(\temp) \parent{1 - \welfaredis^{\agentnum}(\temp)}
        = \rewardfun (1 - \contract)(1 - \valdis_{\contract}^{\agentnum}(\reserve_{\contract})),
    \end{align*}
    where $\temp = \rewardfun(1 - \contract) + \reserve_{\contract}$.
    Notice that by \Cref{lemma:regular mhr order statistic}, the distribution $\welfaredis^{\tempnum}$ is MHR for any $\tempnum \in \mathbb{N}$.
    We denote the monopoly reserve price of $\valdis_{\contract}^{\tempnum}$ as $\reserve_{\contract}^{(\tempnum)}$.
    Based on Corollary~\ref{corollary:mhr prob bound}, we directly obtain that
    \begin{align*}
        \utip(\contract) 
        ={}&\rewardfun (1 - \contract)(1 - \valdis_{\contract}^{\agentnum}(\reserve_{\contract}))\\
        \geq{}&\rewardfun (1 - \contract)(1 - \valdis_{\contract}^{\agentnum}(\reserve_{\contract}^{(\tempnum)}))\\
        \geq{}&\rewardfun (1 - \contract) \parent{1 - \parent{1 - \frac{1}{\ee} \prob{\valmax^{(\tempnum)} \geq 0}}^{\agentnum/\tempnum} } \\
        ={}&\rewardfun (1 - \contract) \parent{1 -  \parent{1 - \frac{1}{\ee} \prob{\welfaremax^{(\tempnum)} \geq \rewardfun (1 - \contract)}}^{\agentnum/\tempnum}} \\
        ={}&\rewardfun (1 - \contract) \parent{1 - \parent{1 - \frac{1}{\ee}(1 - \welfaredis^{\tempnum}(\rewardfun (1 - \contract)))}^{\agentnum/\tempnum}}.
    \end{align*}
    Moreover, when $\tempnum = n$, we have
    \begin{align*}
        \utip(\contract)
        \geq \frac{1}{\ee} \rewardfun(1 - \contract)(1 - \welfaredis^{\agentnum}(\rewardfun(1 - \contract))).
    \end{align*}
    And when $\tempnum = 1$, we have
    \begin{align*}
        \utip(\contract)
        \geq \rewardfun (1 - \contract) \parent{1 -  \parent{1 - \frac{1}{\ee}(1 - \welfaredis(\rewardfun (1 - \contract)))}^{\agentnum}}.
    \end{align*}
\end{proof}

\Cref{lemma:repeated proof} serves as a crucial tool for our subsequent analysis of the two ratios.
Specifically, we will repeatedly utilize the expressions for $\tempnum = n$ and $\tempnum = 1$ in our forthcoming proofs.

\begin{lemma}[\citealp{DRY-15}, Theorem 3.11] \label{lemma:mhr ratio bound}
    For any MHR distribution, its expected value is at most $\ee$ times more than the expected monopoly revenue.
\end{lemma}

\begin{lemma}[\citealp{JLQ-19}, Theorem 1] \label{lemma:mhr OPA vs AP}
    To sell a single item among multiple buyers with i.i.d.\ and MHR distributions, the tight ratio of Myerson's optimal auction to anonymous pricing is $1.2683$.
\end{lemma}

The above two lemmas are important ratio conclusions under the MHR distributions.
Now we are ready to prove \Cref{thm:mhr ratio bounds for n in naturals}.

\begin{proof}[Proof of \Cref{thm:mhr ratio bounds for n in naturals}.]
    We first show the upper bound parts utilizing the above lemmas and then show the lower bound part by analyzing~\Cref{example:mhr bounds}.

    \xhdr{Upper bounding \poa.}
    We first focus on the upper bound of \poa.
    Let $\reserve^{(\agentnum)}$ be the reserve price for distribution $\welfaredis^{\agentnum}$.
    We have
    \begin{align*}
        \utip^{*} 
        \overset{(a)}{=}{}& \max_{\contract\in[0, 1]} \utip(\contract) \\
        \overset{(b)}{\geq}{}& \max_{\contract\in[0, 1]} \frac{1}{e}\rewardfun(1 - \contract)(1 - \welfaredis^{\agentnum}(\rewardfun(1 - \contract)))\\
        \overset{(c)}{=}{}& \frac{1}{\ee} \reserve^{(\agentnum)}(1 - \welfaredis^{\agentnum}(\reserve^{(\agentnum)}))\\
        \overset{(d)}{\geq}{}& \frac{1}{\ee^{2}}\expect[\welfareprf \sim \welfaredis^{\otimes \agentnum}]{\plus{\welfaremax}}\\
        \overset{(e)}{=}{}& \frac{1}{\ee^{2}}\utipfb,
    \end{align*}    
    where equality~(a) holds by definition of $\utip^{*}$, inequality~(b) holds by Lemma~\ref{lemma:repeated proof}, equality~(c) holds by definition of $\reserve^{(\agentnum)}$, inequality~(d) holds by Lemma~\ref{lemma:mhr ratio bound} and equality~(e) holds by the definition of $\utipfb$. 
    Thus we have $\ratiofb = \frac{\utipfb}{\utip^{*}} \leq \ee^{2}$ by definition.

    \xhdr{Upper bounding \podm.}
    Next, we consider the upper bound of \podm.

    We first consider a single agent, i.e., $\agentnum = 1$.
    Note that when $\agentnum = 1$, the second-best benchmark $\utipsb = \reserve(1 - \welfaredis(\reserve))$.
    By above inequality~(c), we can obtain that $\ratiosb = \frac{\utipsb}{\utip^{*}} \leq \ee$.
    
    When, $\agentnum > 1$, by Lemma~\ref{lemma:repeated proof}, we have
    \begin{align*}
        \utip^{*} \geq \reserve^{(\agentnum)} \parent{1 -  \parent{1 - \frac{1}{\ee}(1 - \welfaredis(\reserve^{(\agentnum)}))}^{\agentnum}}.
    \end{align*}

    It follows that
    \begin{align*}
        \ratiosb
        ={}& \frac{\utipsb}{\utip^{*}} \\
        \leq{}& \frac{\utipsb}{\reserve^{(\agentnum)} \parent{1 -  \parent{1 - \frac{1}{\ee}(1 - \welfaredis(\reserve^{(\agentnum)}))}^{\agentnum}}}\\
        ={}& \frac{\utipsb}{\reserve^{(\agentnum)} (1 - \welfaredis^{\agentnum}(\reserve^{(\agentnum)}))} \cdot \frac{1 - \welfaredis^{\agentnum}(\reserve^{(\agentnum)})}{1 -  \parent{1 - \frac{1}{\ee}(1 - \welfaredis(\reserve^{(\agentnum)}))}^{\agentnum}}\\
        ={}& \ratioma_{\agentnum} \cdot \tempfunc_{\agentnum}(\welfaredis(\reserve^{(\agentnum)})).
    \end{align*}
    where $\ratioma_{\agentnum} = \frac{\utipsb}{\reserve^{(\agentnum)} (1 - \welfaredis^{\agentnum}(\reserve^{(\agentnum)}))}$ and $\tempfunc_{\agentnum}(\temp) = \frac{1 - \temp^{\agentnum}}{1 -  \parent{1 - \frac{1}{\ee}(1 - \temp)}^{\agentnum}}$.
    
    Note that $\tempfunc_{\agentnum}(\cdot)$ is non-decreasing in $[0, 1]$ and $1 - \welfaredis^{\agentnum}(\reserve^{(\agentnum)}) \geq \frac{1}{\ee}$ by \Cref{lemma:mhr prob bound}, we have
    \begin{align*}
        \ratiosb \leq \ratioma_{\agentnum} \cdot \tempfunc_{\agentnum}\parent{\parent{1 - \frac{1}{\ee}}^{\frac{1}{\agentnum}}}.
    \end{align*}
    \Cref{lemma:mhr OPA vs AP} provides the numerical results of $\ratioma_{\agentnum}$ for different $\agentnum$, which allows us to obtain the numerical results of the lower bound of \podm. 
    We offer the numerical results of the lower bound in \Cref{tab:mhr rsb} when $\agentnum \leq 44$.
    When $\agentnum \geq 44$, it can be verified that \podm can never reach $\ommhrupbound$ since $\ratioma_{\agentnum}$ and $\tempfunc_{\agentnum}((1 - 1/\ee)^{1/\agentnum}))$ are both decreasing as $\agentnum$ increases.
    Hence the lower bound of $\ommhrupbound$ is reached when $\agentnum = 7$.
    
    \begin{table}[ht]
        \centering
        \clearpage{}\begin{tabular}{ccccccccccc}
\hline
$\agentnum$     & 2                    & 3                    & 4                    & 5                    & 6                    & 7                    & 8                    & 9                    & 10                   & 11                   \\ \hline
$\ratiosb$ & 3                    & 3.025                & 3.031                & 3.032                & 3.033                & $\boldsymbol{\ommhrupbound}$                & 3.033                & 3.032                & 3.032                & 3.032                \\ \hline
\multicolumn{1}{l}{}       & \multicolumn{1}{l}{} & \multicolumn{1}{l}{} & \multicolumn{1}{l}{} & \multicolumn{1}{l}{} & \multicolumn{1}{l}{} & \multicolumn{1}{l}{} & \multicolumn{1}{l}{} & \multicolumn{1}{l}{} & \multicolumn{1}{l}{} & \multicolumn{1}{l}{} \vspace{-1em}\\ \hline
12                         & 13                   & 14                   & 15                   & 16                   & 17                   & 18                   & 19                   & 20                   & 21                   & 22                   \\ \hline
3.032                      & 3.031                & 3.031                & 3.031                & 3.030                & 3.029                & 3.027                & 3.025                & 3.023                & 3.022                & 3.019                \\ \hline
\multicolumn{1}{l}{}       & \multicolumn{1}{l}{} & \multicolumn{1}{l}{} & \multicolumn{1}{l}{} & \multicolumn{1}{l}{} & \multicolumn{1}{l}{} & \multicolumn{1}{l}{} & \multicolumn{1}{l}{} & \multicolumn{1}{l}{} & \multicolumn{1}{l}{} & \multicolumn{1}{l}{} \vspace{-1em}\\ \hline
23                         & 24                   & 25                   & 26                   & 27                   & 28                   & 29                   & 30                   & 31                   & 32                   & 33                   \\ \hline
3.017                      & 3.015                & 3.013                & 3.011                & 3.009                & 3.007                & 3.005                & 3.003                & 3.001                & 2.998                & 2.997                \\ \hline
\multicolumn{1}{l}{}       & \multicolumn{1}{l}{} & \multicolumn{1}{l}{} & \multicolumn{1}{l}{} & \multicolumn{1}{l}{} & \multicolumn{1}{l}{} & \multicolumn{1}{l}{} & \multicolumn{1}{l}{} & \multicolumn{1}{l}{} & \multicolumn{1}{l}{} & \multicolumn{1}{l}{} \vspace{-1em}\\ \hline
34                         & 35                   & 36                   & 37                   & 38                   & 39                   & 40                   & 41                   & 42                   & 43                   & 44                   \\ \hline
2.995                      & 2.993                & 2.991                & 2.989                & 2.987                & 2.986                & 2.984                & 2.982                & 2.980                & 2.979                & 2.977                \\ \hline
\end{tabular}\clearpage{}
        \caption{The upper bounds of \podm for MHR distributions when $\agentnum \leq 44$.}
        \label{tab:mhr rsb}
    \end{table}
    
    \xhdr{Lower bounding \podm and \poa.}
    Now we analyze the lower bound part of the two ratios by considering~\Cref{example:mhr bounds}.
    Formally, we show that for every fixed $\agentnum$ and any small real $\varepsilon > 0$, there exists a parameter $K > 1$ and a reward $\rewardfun$ satisfying that $\ratiofb \geq \ee^{2} - \varepsilon$ and $\ratiosb \geq \ee - \varepsilon$.

    First, we have the virtual price of the conditional exponential distribution is
    \begin{equation*}
        \virval(\temp) 
        = \temp - \frac{1 - \welfaredis_{\rewardfun}(\temp)}{\welfaredens_{\rewardfun}(\temp)} 
        = \temp - \frac{\ee^{-\temp} - \ee^{-\rewardfun}}{\ee^{-\temp}} 
        = \temp - 1 + \ee^{\temp - \rewardfun},
    \end{equation*}
    which does not rely on the parameter $K$.
    
    Next, we set $K = \agentnum^{\rewardfun + 1}/(1 - \ee^{-\rewardfun})$, the distribution function $\welfaredis_{\rewardfun}(\cdot)$ can be write as
    \begin{equation*}
        \welfaredis_{\rewardfun}(\temp) 
        = \frac{1}{\agentnum^{\rewardfun + 1}} \cdot (1 - \ee^{-\temp}) + 1 - \frac{1}{\agentnum^{\rewardfun + 1}}(1 - \ee^{-\rewardfun})
        = 1 - \frac{1}{\agentnum^{\rewardfun + 1}} \cdot (\ee^{-\temp} - \ee^{-\rewardfun}).
    \end{equation*}
        Observe that for any $\temp \geq 0$, we have inequality
    \begin{equation*}
        1 - \frac{\temp}{\agentnum^{\rewardfun}} 
        \leq \parent{1 - \frac{\temp}{\agentnum^{\rewardfun + 1}}}^{\agentnum} 
        \leq 1 - \frac{\temp}{\agentnum^{\rewardfun}} + \frac{\temp^{2}}{2\agentnum^{2\rewardfun + 1}},
    \end{equation*}
    which gives the lower and upper bounds of $1 - \welfaredis_{\rewardfun}^{\agentnum}(\cdot)$ as
    \begin{equation*}
        \frac{\ee^{-\temp} - \ee^{-\rewardfun}}{\agentnum^{\rewardfun}} - \frac{\parent{\ee^{-\temp} - \ee^{-\rewardfun}}^{2}}{2\agentnum^{2\rewardfun + 1}}
        \leq 1 - \welfaredis_{\rewardfun}^{\agentnum}(\temp) 
        \leq \frac{\ee^{-\temp} - \ee^{-\rewardfun}}{\agentnum^{\rewardfun}}.
    \end{equation*}
    
    We are ready to estimate the bounds of \podm and \poa by using the above bounds of $1 - \welfaredis_{\rewardfun}^{\agentnum}(\cdot)$.
    On the one hand,
    \begin{align*}
        \utip^{*}
        ={}& \max_{\temp \in [0, \rewardfun]} \virval(\temp) (1 - \welfaredis_{\rewardfun}^{\agentnum}(\temp)) \\
        \leq{}& \max_{\temp \in [0, \rewardfun]} \parent{\temp - 1 + \ee^{\temp - \rewardfun}} \frac{\ee^{-\temp} - \ee^{-\rewardfun}}{\agentnum^{\rewardfun}} \\
        \leq{}& \frac{1}{\agentnum^{\rewardfun}} \max_{\temp \in [0, \rewardfun]} \parent{\temp - 1 + \ee^{\temp - \rewardfun}} \ee^{-\temp} \\
        ={}& \frac{1}{\agentnum^{\rewardfun}}(\ee^{-2} + \ee^{-\rewardfun}).
    \end{align*}
    
    On the other hand,
    \begin{align*}
        \utipfb
        ={}& \int_{0}^{\rewardfun} (1 - \welfaredis^{\agentnum}(\temp)) \cdot \dd \temp \\
        \geq{}& \int_{0}^{\rewardfun} \parent{\frac{\ee^{-\temp} - \ee^{-\rewardfun}}{\agentnum} - \frac{\parent{\ee^{-\temp} - \ee^{-\rewardfun}}^{2}}{2\agentnum^{2}}} \cdot \dd \temp \\
        \geq{}& \int_{0}^{\rewardfun} \parent{\frac{\ee^{-\temp} - \ee^{\rewardfun}}{\agentnum} - \frac{\ee^{-2\temp}}{2\agentnum^{2}}} \cdot \dd \temp \\
        \geq{}& \frac{1}{\agentnum^{\rewardfun}} (1 - (\rewardfun + 1)\ee^{-\rewardfun}) - \frac{1}{4\agentnum^{2\rewardfun + 1}}(1 - \ee^{-2\rewardfun}),
    \end{align*}
    and
    \begin{align*}
        \utipsb
        ={}& \int_{0}^{\rewardfun} \virval(\temp) (1 - \welfaredis_{\rewardfun}^{\agentnum}(\temp)) \cdot \dd \temp \\
        \geq{}& \int_{0}^{\rewardfun} \parent{\temp - 1 + \ee^{\temp - \rewardfun}} \parent{\frac{\ee^{-\temp} - \ee^{-\rewardfun}}{\agentnum} - \frac{\parent{\ee^{-\temp} - \ee^{-\rewardfun}}^{2}}{2\agentnum^{2}}} \cdot \dd \temp \\
        \geq{}& \int_{0}^{\rewardfun} \parent{\temp - 1 + \ee^{\temp - \rewardfun}} \parent{\frac{\ee^{-\temp} - \ee^{-\rewardfun}}{\agentnum} - \frac{\ee^{-2\temp}}{2\agentnum^{2}}} \cdot \dd \temp \\
        ={}& \frac{1}{\agentnum^{\rewardfun}}\parent{\ee^{-1} - \frac{1}{2}(\rewardfun^{2} -2\rewardfun + 5)\ee^{-\rewardfun} + \ee^{1 - 2\rewardfun}} - \frac{1}{8\agentnum^{2\rewardfun + 1}} \parent{\ee^{-2} + \ee^{-\rewardfun - 1} + (2\rewardfun + 3)\ee^{-2\rewardfun}}.
    \end{align*}

    Therefore, we can finally obtain the lower bounds
    \begin{align*}
        \ratiofb = \frac{\utipfb}{\utip^{*}} 
        = \ee^{2} - \bigO{\rewardfun}{\rewardfun\ee^{-\rewardfun} + \agentnum^{-\rewardfun - 1}},
    \end{align*}
    and
    \begin{align*}
        \ratiosb 
        = \frac{\utipsb}{\utip^{*}} 
        = \ee - \bigO{\rewardfun}{\ee^{-\rewardfun} + \agentnum^{-\rewardfun - 1}}.
    \end{align*}
    Finally, for an arbitrary large $\rewardfun$, we have $\ratiofb \geq \ee^{2} - \varepsilon$ and $\ratiosb \geq \ee - \varepsilon$.
\end{proof}
 
\subsection{Proof of Theorem~\ref{thm:r_fb asymp lb}}
\label{apx:r_fb asymp lb proof}

In this section, we prove \Cref{thm:r_fb asymp lb}.

\fbasymp*

 Our analysis relies on the following technical lemmas.

\begin{lemma}[\citealp{GPZ-21}, Lemma 3.3] \label{lemma: MHR concentration}
     Let $\mu_{i}$ be the expectation of the $i$th lowest order statistic derived from $\agentnum$ i.i.d.\ random variables following a continuous positive MHR distribution $\welfaredis$.
     For any real $\pert \in [0, 1]$,
     \begin{equation*}
        \welfaredis(\pert \cdot \mu_{i}) \leq 1 - \exp( -\pert \cdot (\harmonic_{\agentnum} - \harmonic_{\agentnum - i})), 
     \end{equation*}
    where $\harmonic_{\agentnum} = 1 + \frac{1}{2} + \cdots + \frac{1}{\agentnum}$ is the $\agentnum$th harmonic number.
\end{lemma}

\begin{lemma} \label{lemma: lb of fb ratio}
    Suppose that the \contribution $\welfare_{i}$ for each agent $i \in \agents$ satisfies an MHR distribution $\welfaredis$, and $\bar{\welfaredis}$ is the conditional distribution conditioning on the \contribution $\welfare_{i} \geq 0$.
    Let $\pert = \rewardfun(1 - \contract)/\bar{\mu}_{\agentnum}$, where $\bar{\mu}_{\agentnum}$ is the excepted value of the first-order statistic derived from $\agentnum$ i.i.d.\ random variables following the distribution $\bar{\welfaredis}$.
    For every linear contract $\contract \in [0, 1]$, the utility of the principal $\utip(\contract)$ is at least $\pert \cdot \parent{1 - (1 - (1 - \welfaredis(0))\ee^{-\pert \harmonic_{\agentnum} - 1})^{\agentnum}}$ times more than first-best benchmark $\utipfb$, i.e.,
    \begin{align*}
        \utip(\contract) 
        \geq \pert \cdot \parent{1 - \parent{1 - (1 - \welfaredis(0))\ee^{-\pert \harmonic_{\agentnum} - 1}}^{\agentnum}} \cdot \utipfb.
    \end{align*}
\end{lemma}

\begin{proof}
    By \Cref{lemma:repeated proof}, we have:
    \begin{align*}
        \utip(\contract) \geq \rewardfun (1 - \contract) \parent{1 -  \parent{1 - \frac{1}{\ee}(1 - \welfaredis(\rewardfun (1 - \contract)))}^{\agentnum}}.
    \end{align*}
    Observe that the probability distribution function satisfies that 
    \begin{equation*}
        \bar{\welfaredis}(\temp) 
        = \frac{\welfaredis(\temp) - \welfaredis(0)}{1 - \welfaredis(0)}
        \leq \welfaredis(\temp)
    \end{equation*}
    for every $\temp \in [0, \infty)$.
    Recall that $\bar{\mu}_{\agentnum}$ is the expected value of the first-order statistic derived from $\agentnum$ i.i.d.\ random variables following the distribution $\bar{\welfaredis}$.
    We have
    \begin{align*}
        \utipfb
        = \expect[\welfareprf \sim \welfaredis^{\otimes \agentnum}]{\plus{\welfaremax}}
        = \int_{0}^{\infty} (1 - \welfaredis^{\agentnum}(\temp)) \cdot \dd \temp
        \leq \int_{0}^{\infty} (1 - \bar{\welfaredis}^{\agentnum}(\temp)) \cdot \dd \temp
        = \bar{\mu_{\agentnum}}.
    \end{align*}
    Substituting the bounds of $\utip(\contract)$ and $\utipfb$, we obtain that
    \begin{align*}
        \frac{\utip(\contract)}{\utipfb}
        \geq{}& \frac{\pert \cdot \bar{\mu}_{\agentnum} \parent{1 - \parent{1 - \frac{1}{\ee}(1 - \welfaredis(\pert \cdot \bar{\mu}_{\agentnum})}^{\agentnum}}}{\bar{\mu}_{\agentnum}} \\
        ={}& \pert \parent{1 - \parent{1 - \frac{1}{\ee}(1 - \welfaredis(\pert \cdot \bar{\mu}_{\agentnum})}^{\agentnum}} \\
        ={}& \pert \cdot \parent{1 - \parent{1 - \frac{1}{\ee}(1 - \welfaredis(0))(1 - \bar{\welfaredis}(\pert \cdot \bar{\mu}_{\agentnum}))}^{\agentnum}}.
    \end{align*}

    Notice that $\bar{\welfaredis}$ is MHR since $\welfaredis$ is MHR.
    According to \Cref{lemma: MHR concentration}, we have
    \begin{equation*}
        \bar{\welfaredis}(\pert \cdot \bar{\mu}_{\agentnum}) \leq 1 - \ee^{-\pert \cdot \harmonic_{\agentnum}}.
    \end{equation*}
    Therefore, we can further obtain that
    \begin{align*}
        \frac{\utip(\contract)}{\utipfb}
        \geq{}& \pert \cdot \parent{1 - \parent{1 - \frac{1}{\ee}(1 - \welfaredis(0)) \ee^{-\pert \cdot \harmonic_{\agentnum}}}^{\agentnum}} \\
        ={}& \pert \cdot \parent{1 - \parent{1 - (1 - \welfaredis(0)) \ee^{-\pert \harmonic_{\agentnum} - 1}}^{\agentnum}},
    \end{align*}
    which completes the proof.
\end{proof}

Now we are ready to prove \Cref{thm:r_fb asymp lb}.

\begin{proof}[Proof of \Cref{thm:r_fb asymp lb}.]
    Consider the design of setting $\contract^{\dag} = 1 - \frac{\bar{\mu}_{\agentnum}}{r} (1 - \frac{\ln \harmonic_{\agentnum}}{\harmonic_{\agentnum}})$.
    According to the definition of $\pert$ in \Cref{lemma: lb of fb ratio}, we have $\pert^{\dag} = 1 - \frac{\ln \harmonic_{\agentnum}}{\harmonic_{\agentnum}}$.
    In addition, using the lower bound of $\frac{\utip(\contract)}{\utipfb}$ in \Cref{lemma: lb of fb ratio}, we obtain that
    \begin{align*}
        \frac{1}{\ratiofb} 
        ={} \frac{\utip(\contract)}{\utipfb}
        \geq{} \pert^{\dag} \cdot \parent{1 - \parent{1 - \frac{1 - \welfaredis(0)}{\ee^{\pert^{\dag} \harmonic_{\agentnum} + 1}}}^{\agentnum}}
        \geq{} \parent{1 - \frac{\ln \harmonic_{\agentnum}}{\harmonic_{\agentnum}}} \parent{1 - \parent{1 - \frac{1 - \welfaredis(0)}{\ee^{\harmonic_{\agentnum} - \ln \harmonic_{\agentnum} + 1}}}^{\agentnum}}.
    \end{align*}

    Notice that $\ln \agentnum \leq \harmonic_{\agentnum} \leq \ln \agentnum + \gamma + \frac{1}{2} \leq \ln \agentnum + 2$, where $\gamma$ is the Euler constant.
    Thus, we have
    \begin{equation*}
         \ee^{\harmonic_{\agentnum} - \ln \harmonic_{\agentnum} + 1}
         \leq \ee^{\ln \agentnum - \ln \ln \agentnum + 2}
         = \frac{\agentnum}{\ln \agentnum} \cdot \ee^{2}.
    \end{equation*}
    
    Consequently, since $\welfaredis(0) < 1$, we can further obtain the asymptotic lower bound as follows.
    \begin{align*}
        \frac{1}{\ratiofb} 
        \geq{}& \parent{1 - \frac{\ln \harmonic_{\agentnum}}{\harmonic_{\agentnum}}} \parent{1 - \parent{1 - \frac{1 - \welfaredis(0)}{\ee^{\harmonic_{\agentnum} - \ln \harmonic_{\agentnum} + 1}}}^{\agentnum}} \\
        \geq{}& \parent{1 - \frac{\ln \harmonic_{\agentnum}}{\harmonic_{\agentnum}}} \parent{1 - \parent{1 - \frac{1 - \welfaredis(0)}{\ee^{2}} \cdot \frac{\ln \agentnum}{\agentnum}}^{\agentnum}} \\
        ={}& \parent{1 - \bigO{\agentnum}{\frac{\log \log \agentnum}{\log \agentnum}}} \parent{1 - \agentnum^{-\bigO{\agentnum}{1}}} \\
        ={}& 1 - \bigO{\agentnum}{\frac{\log \log \agentnum}{\log \agentnum}}.
    \end{align*}
    Finally, since $\ratiosb \leq \ratiofb$, we have
    \begin{equation*}
        \ratiosb 
        \leq \ratiofb
        \leq \parent{1 - \frac{\ln \harmonic_{\agentnum}}{\harmonic_{\agentnum}}}^{-1} \parent{1 - \parent{1 - \frac{1 - \welfaredis(0)}{\ee^{\harmonic_{\agentnum} - \ln \harmonic_{\agentnum} + 1}}}^{\agentnum}}^{-1}
        = 1 + \bigO{\agentnum}{\frac{\log \log \agentnum}{\log \agentnum}},
    \end{equation*}
    which completes the proof.
\end{proof}
 
\subsection{Proof of Theorem~\ref{thm:prin's utility under ap}}
\label{apx:prin's utility under ap proof}

\thmOptContractUnderAP*
\begin{proof}
    For any optimal contract $\ctrctprfopt$, since agents have identical expected reward profiles $\rwdprf$, there exists a linear contract $\ctrctlnropt$ such that 
    \begin{align*}
        \inner{\ctrctprfopt}{\rwdprf} = \inner{\ctrctlnropt \cdot \ctrctone}{\rwdprf},
    \end{align*}
    which implies that 
    \begin{align*}
        \inner{\ctrctone - \ctrctprfopt}{\rwdprf} = \inner{\ctrctone - \ctrctlnropt \cdot \ctrctone}{\rwdprf}.
    \end{align*}

    Therefore, we can obtain that the agents' \cwelfare $\cwf_{\ctrctprfopt} = \cwf_{\ctrctlnropt}$ and the corresponding distribution $\valdis_{\ctrctprfopt} = \valdis_{\ctrctlnropt}$.
    It follows that the principal's utilities are equivalent under these two contracts, i.e.,
    \begin{equation*}
        \inner{\ctrctone - \ctrctprfopt}{\rwdprf} \cdot \parent{1 - \welfaredis^{\agentnum}_{\ctrctprfopt}(\reserve^{(\agentnum)}_{\ctrctprfopt})} = \inner{\ctrctone - \ctrctlnropt \cdot \ctrctone}{\rwdprf} \cdot \parent{1 - \welfaredis^{\agentnum}_{\ctrctlnropt}(\reserve^{(\agentnum)}_{\ctrctlnropt})}
    \end{equation*}
    which completes the proof as desired.
\end{proof} 
\subsection{Proof of Theorem~\ref{thm:regular bounds under ap}}
\label{apx:regular bounds under ap proof}
In this section, we prove \Cref{thm:regular bounds under ap}.

\thmRegularBoundUnderAP*

\begin{proof}
     We separately show the upper bound and the lower bound parts, which are analogously to that of \Cref{thm:regular bounds}.

    \xhdr{Upper bounding \podm and \poa.}
    We first consider the upper bounds.
    Recall that \cref{equation: U* >= phi^(n)(z)pr} in the proof of \Cref{thm:regular bounds} shows that 
    \begin{align*}
        \utip^{*}
        = \max\limits_{\temp\in\supp{\welfaredis}} \virval(\temp)\parent{1 - \welfaredis^{\agentnum}(\temp)}
        \geq \max\limits_{\temp\in\supp{\welfaredis}} \virval^{(\agentnum)}\parent{\temp)(1 - \welfaredis^{\agentnum}(\temp)}
        = \utippost.
    \end{align*}
    Thereafter, the remaining part of this proof is identical to that of \Cref{thm:regular bounds}.
    
    \xhdr{Lower bounding \podm and \poa.}
    Now we analyze the lower bound part of the two ratios, considering \Cref{example:regular bounds} as well.
    Note that for a single agent (i.e, $\agentnum = 1$), we have $\utip^{*} = \utippost$.
    Therefore, \Cref{example:regular bounds} still serves as the example of a lower bound for the two ratios and the analysis is identical to that of \Cref{thm:regular bounds}.
\end{proof} 
\subsection{Proof of Theorem~\ref{thm:mhr bounds under ap}}
\label{apx:mhr bounds under ap proof}

In this section, we prove \Cref{thm:mhr bounds under ap}.

\thmMHRBoundUnderAP*

\begin{proof}
    We first show the upper bound parts and then show the lower bound part by analyzing~\Cref{example:mhr bounds} with a single agent.
    
    \xhdr{Upper bounding \podm.}
    We first show the upper bound of \podm.
    By \Cref{lemma:repeated proof} we have 
    \begin{align*}
        \utip^{*} 
        \geq \max_{\contract\in[0, 1]} \rewardfun (1 - \contract)(1 - \valdis_{\contract}^{\agentnum}(\reserve_{\contract}^{(\agentnum)})))
        = \utippost
        \geq \frac{1}{\ee} \cdot \reserve^{(\agentnum)}(1 - \welfaredis^{\agentnum}(\reserve^{(\agentnum)})).
    \end{align*}

    It follows that 
    \begin{align*}
        \ratiosb 
        {}= \frac{\utipsb}{\utippost}
        {}\leq \ee \cdot \frac{\utipsb}{\reserve^{(\agentnum)}(1 - \welfaredis^{\agentnum}(\reserve^{(\agentnum)}))}
        {}\overset{(a)}{=} \ee \cdot 1.2683
        {}\leq 3.448,
    \end{align*}
    where equality~(a) holds by \Cref{lemma:mhr OPA vs AP}.
    
    \xhdr{Upper bounding \poa.}
    Next, we consider the upper bound of \poa.
    We can also obtain the following from above: 
    \begin{align*}
        \utippost
        \geq \frac{1}{\ee} \reserve^{(\agentnum)}(1 - \welfaredis^{\agentnum}(\reserve^{(\agentnum)}))
        \geq \frac{1}{\ee} \utipfb,
    \end{align*}
    where the analysis is identical to that of \Cref{thm:mhr ratio bounds for n in naturals}.

    \xhdr{Lower bounding \podm and \poa.}
    Now we analyze the lower bound part of the two ratios by considering \Cref{example:mhr bounds} with only a single agent (i.e., $\agentnum = 1$).
    Since for a single agent, we have $\utip^{*} = \utippost$, \Cref{example:mhr bounds} still serves as the example of a lower bound for the two ratios. 
    Therefore, the analysis is identical to that of \Cref{thm:mhr ratio bounds for n in naturals}, where we assign $\agentnum = 1$.
\end{proof} 
\subsection{Proof of Theorem~\ref{thm:robust regular bounds}}
\label{apx:robust regular bounds proof}

In this section, we prove \Cref{thm:robust regular bounds}.

\robustbounds*

 Our analysis relies on the following technical lemmas.

\begin{lemma} \label{lem:opt-ratio-bound}
    Let $\lownum, \upnum \in \naturals$ with $\lownum \leq \upnum$.
    Then, for any regular distribution $\welfaredis$ and contract $\contract\in[0,1]$, we have
    \begin{equation*}
        \frac{\utip(\contract,\lownum)}{\utip(\contract,\upnum)}  
        \geq \frac{\lownum}{\upnum}.
    \end{equation*}
\end{lemma}

\begin{proof}
    Recall that $\utip(\contract,\agentnum) = \virval(\temp) (1 - \welfaredis^{\agentnum}(\temp))$, where $\temp = \rewardfun(1-\contract) + \valreserve_{\contract}$.
    We have
    \begin{align*}
        \frac{\utip(\contract,\lownum)}{\utip(\contract,\upnum)}
        =\frac{\virval(\temp) (1 - \welfaredis^{\lownum}(\temp))}{\virval(\temp) (1 - \welfaredis^{\upnum}(\temp))}
        ={} \frac{1 - \welfaredis^{\lownum}(\temp)}{1 - \welfaredis^{\upnum}(\temp)}.
    \end{align*}

    Consider the function $x \mapsto \frac{1 - x^{\lownum}}{1 - x^{\upnum}}$ for $x \in [0, 1]$.
    One can easily verify that this function is decreasing for $x \in [0, 1]$.
    Thus, we have
    \begin{align*}
        \frac{1 - \welfaredis^{\lownum}(\temp)}{1 - \welfaredis^{\upnum}(\temp)}
        \geq \lim_{x \to 1}\frac{1 - x^{\lownum}}{1 - x^{\upnum}}
        = \lim_{x \to 1} \frac{-\lownum x^{\lownum - 1}}{-\upnum x^{\upnum - 1}}
        = \frac{\lownum}{\upnum},
    \end{align*}
    which gives the desired inequality.
\end{proof}

\begin{lemma} \label{lem:benchmark-ratio-bound}
    Let $\lownum, \upnum \in \naturals$ with $\lownum \leq \upnum$.
    Then, for any regular distribution $\welfaredis$, we have
    \begin{equation*}
        \frac{\utipfb(\lownum)}{\utipfb(\upnum)} 
        \geq \frac{\utipsb(\lownum)}{\utipsb(\upnum)} 
        \geq \frac{\lownum}{\upnum}.
    \end{equation*}
\end{lemma}

\begin{proof}
    We first prove the second inequality.
    For any $\agentnum \in \mathbb{R}^{+}$, we can rewrite $\utipfb(\agentnum)$ as
    \begin{equation*}
        \utipsb(\agentnum) 
        = \int_{0}^{\infty} \virval(\temp) (1 - \welfaredis^{\agentnum}(\temp)) \cdot \dd \temp
        = \int_{0}^{\infty} \virval(\temp) \int_{0}^{\agentnum} \welfaredis^{u}(\temp) (- \ln \welfaredis(\temp)) \cdot \dd u \dd \temp.
    \end{equation*}
    Since the above integration is non-negative and measurable, according to Tonelli's theorem, we can exchange the order of integration as
    \begin{equation*}
        \utipsb(\agentnum) 
        = \int_{0}^{\agentnum} \int_{0}^{\infty} \virval(\temp) \welfaredis^{u}(\temp) (- \ln \welfaredis(\temp)) \cdot \dd \temp \dd u.
    \end{equation*}

    Define the function $g(u) = \int_{0}^{\infty} \virval(\temp) \welfaredis^{u}(\temp) (- \ln \welfaredis(\temp)) \cdot \dd \temp$.
    We know $\utipsb(\lownum) = \int_{0}^{\lownum} g(u) \cdot \dd u$ and $\utipsb(\upnum) = \int_{0}^{\upnum} g(u) \cdot \dd u$.
    Thus, the desired inequality is equivalent to
    \begin{equation*}
        \frac{1}{\lownum} \int_{0}^{\lownum} g(u) \cdot \dd u
        \geq \frac{1}{\upnum} \int_{0}^{\upnum} g(u) \cdot \dd u.
    \end{equation*}
    Therefore, it suffices to show that $g(\cdot)$ is non-increasing according to the monotonicity of the average integral.

    Fix $u_{1}, u_{2} \in \mathbb{R}^{+}$ with $u_{1} \leq u_{2}$.
    To show that $g(u_{2}) \leq g(u_{1})$, we employ the integral Chebyshev inequality: for a measure $\mu$ and two functions $\xi$ and $\zeta$ with opposite monotonicity.
    it holds
    \begin{equation*}
        \int \dd \mu \int \xi \zeta \cdot \dd \mu \leq \int \xi \cdot \dd \mu \int \zeta \cdot \dd \mu.
    \end{equation*}
    Let $\dd \mu(\temp) = \welfaredis^{u_{1}}(\temp) \cdot \dd \temp$.
    Define $\xi(\temp) = \virval(\temp) (- \ln \welfaredis(\temp))$ and $\zeta(\temp) = \welfaredis^{u_{2} - u_{1}}(\temp)$.
    Given that $\welfaredis$ is regular, $\xi(\cdot)$ is non-increasing, while $\zeta(\cdot)$ is non-decreasing.
    Applying the inequality yields
    \begin{align*}
        g(u_{2})
        ={}& \int_{0}^{\infty} \xi(\temp) \zeta(\temp) \cdot \dd \mu(\temp) \\    
        \leq{}& \frac{1}{\int_{0}^{\infty} \dd \mu(\temp)} \int_{0}^{\infty} \xi(\temp) \cdot \dd \mu(\temp) \int_{0}^{\infty} \zeta(\temp) \cdot \dd \mu(\temp) \\
        ={}& g(u_{1}) \cdot \frac{\int_{0}^{\infty} \welfaredis^{u_{2}}(\temp) \cdot \dd \temp}{\int_{0}^{\infty} \welfaredis^{u_{1}}(\temp) \cdot \dd \temp} \\
        \leq{}& g(u_{1}),
    \end{align*}
    where the last inequality holds since $\welfaredis^{u}(\temp)$ is non-increasing in $u$ for any fixed $\temp$.
    Consequently, $g(u)$ is non-increasing for $u \in \mathbb{R}^{+}$, implying the desired inequality.

    Next, we prove the first inequality.
    By definition, it is sufficient to prove that
    \begin{equation*}
        \frac{\int_{0}^{\infty} (1 - \welfaredis^{\lownum}(\temp)) \cdot \dd \temp}{\int_{0}^{\infty} (1 - \welfaredis^{\upnum}(\temp)) \cdot \dd \temp}
        \geq \frac{\int_{0}^{\infty} \virval(\temp) (1 - \welfaredis^{\lownum}(\temp)) \cdot \dd \temp}{\int_{0}^{\infty} \virval(\temp) (1 - \welfaredis^{\upnum}(\temp)) \cdot \dd \temp}.
    \end{equation*}
    We can rearrange the above inequality as
    \begin{equation*}
        \int_{0}^{\infty} \virval(\temp) (1 - \welfaredis^{\upnum}(\temp)) \cdot \dd \temp \int_{0}^{\infty} (1 - \welfaredis^{\lownum}(\temp^{\dag})) \cdot \dd \temp^{\dag}
        \geq \int_{0}^{\infty} \virval(\temp^{\dag}) (1 - \welfaredis^{\lownum}(\temp^{\dag})) \cdot \dd \temp^{\dag} \int_{0}^{\infty} (1 - \welfaredis^{\upnum}(\temp)) \cdot \dd \temp,
    \end{equation*}
    which is equivalent to
    \begin{equation*}
        \int_{0}^{\infty} \int_{0}^{\infty} \parent{\virval(\temp) - \virval(\temp^{\dag})} (1 - \welfaredis^{\upnum}(\temp)) (1 - \welfaredis^{\lownum}(\temp^{\dag})) \cdot \dd \temp \dd \temp^{\dag} \geq 0.
    \end{equation*}
    
    Exchanging the variables $\temp$ and $\temp^{\dag}$ in the double integral and taking the average, we can further derive that
    \begin{align*}
        &\int_{0}^{\infty} \int_{0}^{\infty} \parent{\virval(\temp) - \virval(\temp^{\dag})} (1 - \welfaredis^{\upnum}(\temp)) (1 - \welfaredis^{\lownum}(\temp^{\dag})) \cdot \dd \temp \dd \temp^{\dag} \\
        ={}&\frac{1}{2}\int_{0}^{\infty} \int_{0}^{\infty} \parent{\virval(\temp) - \virval(\temp^{\dag})} \parent{(1 - \welfaredis^{\upnum}(\temp)) (1 - \welfaredis^{\lownum}(\temp^{\dag})) - (1 - \welfaredis^{\upnum}(\temp^{\dag})) (1 - \welfaredis^{\lownum}(\temp))} \cdot \dd \temp \dd \temp^{\dag} \\
        ={}& \int_{0}^{\infty} \int_{\temp^{\dag}}^{\infty} \parent{\virval(\temp) - \virval(\temp^{\dag})} \parent{(1 - \welfaredis^{\upnum}(\temp)) (1 - \welfaredis^{\lownum}(\temp^{\dag})) - (1 - \welfaredis^{\upnum}(\temp^{\dag})) (1 - \welfaredis^{\lownum}(\temp))} \cdot \dd \temp \dd \temp^{\dag}.
    \end{align*}
    Since the integrand is symmetric in $\temp$ and $\temp^{\dag}$, the second equality holds by splitting the integral into two parts where $\temp \geq \temp^{\dag}$ and $\temp < \temp^{\dag}$.

    Therefore, it is sufficient to show
    \begin{equation*}
        \int_{0}^{\infty} \int_{\temp^{\dag}}^{\infty} \parent{\virval(\temp) - \virval(\temp^{\dag})} A(\temp, \temp^{\dag})
        \geq 0,
    \end{equation*}
    where $A(\temp, \temp^{\dag}) = (1 - \welfaredis^{\upnum}(\temp)) (1 - \welfaredis^{\lownum}(\temp^{\dag})) - (1 - \welfaredis^{\upnum}(\temp^{\dag})) (1 - \welfaredis^{\lownum}(\temp))$.
    One can easily verify that the function $x \mapsto \frac{1 - x^{\lownum}}{1 - x^{\upnum}}$ is decreasing for $x \in [0, 1]$.
    Thus, we know $A(\temp, \temp^{\dag}) \geq 0$ for any $\temp \geq \temp^{\dag}$.
    In addition, since $\welfaredis$ is a regular distribution, its virtual value function $\virval(\cdot)$ is non-decreasing.
    Thus, we have $\virval(\temp) - \virval(\temp^{\dag}) \geq 0$ for any $\temp \geq \temp^{\dag}$.
    Consequently, the integrand is always non-negative for any $\temp \geq \temp^{\dag}$, which proves the desired inequality.
\end{proof}

Now, we are ready to prove \Cref{thm:robust regular bounds}.

\begin{proof}[Proof of \Cref{thm:robust regular bounds}.]
    Recall that the regular \contribution distribution $\welfaredis$ has the bounded support $\supp{\welfaredis}=[\lowsupp, \upsupp]$, $(0\leq\lowsupp\leq\upsupp\leq\rewardfun)$, and the ratio of its upper and lower bounds satisfies $1 \leq \upsupp/\lowsupp \leq \ratiolu$.
    
    Consider the design of setting $\contract^{\dag}$ as the optimal contract for some number of agents $\agentnum^{\dag}$, i.e., 
    \begin{align*}
        \contract^{\dag} 
        = \argmax_{\contract\in[0, 1]} \utip(\contract, \agentnum^{\dag})
        = \argmax_{\temp\in\supp{\welfaredis}} \virval(\temp) (1 - \welfaredis^{\agentnum^{\dag}}(\temp)).
    \end{align*}
    Let $n^{*} \in [\lownum: \upnum]$ denote the integer that maximizes $\frac{\utipfb(\agentnum)}{\utip(\contract^{\dag}, \agentnum)}$ under this contract $\contract^{\dag}$, i.e.,
    \begin{align*}
        n^{*} = \argmax_{\agentnum\in[\lownum:\upnum]} \frac{\utipfb(\agentnum)}{\utip(\contract^{\dag}, \agentnum)}.
    \end{align*}
    By definition, we have
    \begin{align*}
        \ratiofb(\lownum,\upnum)
        = \min_{\contract\in[0,1]} \max_{\agentnum\in[\lownum:\upnum]} \frac{\utipfb(\agentnum)}{\utip(\contract, \agentnum)}
        \leq \max_{\agentnum\in[\lownum:\upnum]} \frac{\utipfb(\agentnum)}{\utip(\contract^{\dag}, \agentnum)}
        = \frac{\utipfb(\agentnum^{*})}{\utip(\contract^{\dag}, \agentnum^{*})}.
    \end{align*}
    Using the monotonicity of $\utip(\contract^{\dag}, \cdot)$ and $\utipfb(\cdot)$, we bound this ratio according to the relation between $\agentnum^{\dag}$ and $n^{*}$.
    When $\agentnum^{\dag} \leq n^{*}$, by~\Cref{lem:benchmark-ratio-bound}, we have
    \begin{align*}
        \ratiofb(\lownum,\upnum)
        \leq{} \frac{\utipfb(\agentnum^{*})}{\utip(\contract^{\dag}, \agentnum^{*})}
        ={} \frac{\utipfb(\agentnum^{*})}{\utipfb(\agentnum^{\dag})} \cdot \frac{\utipfb(\agentnum^{\dag})}{\utip(\contract^{\dag}, \agentnum^{*})}
        \leq \frac{\agentnum^{*}}{\agentnum^{\dag}} \cdot \frac{\utipfb(\agentnum^{\dag})}{\utip(\contract^{\dag}, \agentnum^{\dag})}
        \leq \frac{\upnum}{\lownum} \cdot \bigO{\ratiolu}{(\log \ratiolu)^{2}}.
    \end{align*}    
    When $\agentnum^{\dag} > n^{*}$, by~\Cref{lem:opt-ratio-bound}, we have
    \begin{align*}
        \ratiofb(\lownum,\upnum)
        \leq{} \frac{\utipfb(\agentnum^{*})}{\utip(\contract^{\dag}, \agentnum^{*})}
        ={} \frac{\utipfb(\agentnum^{*})}{\utip(\contract^{\dag}, \agentnum^{\dag})} \cdot \frac{\utip(\contract^{\dag}, \agentnum^{\dag})}{\utip(\contract^{\dag}, \agentnum^{*})}
        \leq \frac{\utipfb(\agentnum^{\dag})}{\utip(\contract^{\dag}, \agentnum^{\dag})} \cdot \frac{\agentnum^{\dag}}{\agentnum^{*}}
        \leq \bigO{\ratiolu}{(\log \ratiolu)^{2}} \cdot \frac{\upnum}{\lownum}.
    \end{align*}    
    Similarly, we can also obtain that by~\Cref{lem:opt-ratio-bound} that
    \begin{equation*}
        \ratiosb(\lownum,\upnum)
        \leq \max_{\agentnum\in[\lownum:\upnum]} \frac{\utipsb(\agentnum)}{\utip(\contract^{\dag}, \agentnum)}
        \leq \frac{\upnum}{\lownum} \cdot \bigO{\ratiolu}{\log \ratiolu}.     
    \end{equation*}
    which completes the proof.
\end{proof}
 
\subsection{Proof of Theorem~\ref{thm:robust MHR bounds}}
\label{apx:robust MHR bounds proof}

In this section, we prove \Cref{thm:robust MHR bounds}. 

\mhrrobustbounds*

\begin{lemma}[\citealp{BFFBDS-17}, Lemma 5.3] \label{lem:ratio of ELSO}
    Let $\welfaredis$ be a positive MHR distribution.
    Let $\mu_{\agentnum}$ be the expected value of the first-order static out of $\agentnum$ i.i.d.\ variables drawing from distribution $\welfaredis$.
    For all $\lownum \leq \upnum$, it holds
    \begin{equation*}
        \frac{\mu_{\lownum}}{\mu_{\upnum}} \geq \frac{\harmonic_{\lownum}}{\harmonic_{\upnum}} \geq \frac{\ln \lownum}{\ln \upnum}.
    \end{equation*}
\end{lemma}

Our analysis relies on the following technical lemmas.

\begin{restatable}{lemma}{lemUFB}
\label{lem:U_fb}
    Let $\welfaredis$ be a distribution and $\bar{\welfaredis}$ be the conditional distribution conditioning on the event that the random value following $\welfaredis$ is non-negative \textup{(}i.e., the distribution function satisfies that $\bar{\welfaredis}(\temp) = \frac{\welfaredis(\temp) - \welfaredis(0)}{1 - \welfaredis(0)}$\textup{)}.
    Then we have
    \begin{equation*}
        \utipfb(\agentnum) = \sum_{i = 0}^{\agentnum} \binom{\agentnum}{i} (1 - \welfaredis(0))^{i}\welfaredis(0)^{\agentnum - i} \cdot \bar{\mu}_{i},
    \end{equation*}
    where $\bar{\mu}_{\agentnum}$ denotes the expected value of the first-order statistic out of $\agentnum$ i.i.d.\ variables drawing from distribution $\bar{\welfaredis}$.
    
    Moreover, if $\bar{\welfaredis}$ is the standard exponential distribution, \textup{(}i.e., $\bar{\welfaredis}(\temp) = 1 - \ee^{-\temp}$ for $\temp \geq 0$; and $\bar{\welfaredis}(\temp) = 0$ otherwise\textup{)}, then we have
    \begin{equation*}
        \utipfb(\agentnum) = \int_{0}^{1} \frac{1 - ((1 - \welfaredis(0))\temp + \welfaredis(0))^{\agentnum}}{1 - \temp} \dd \temp.
    \end{equation*}
\end{restatable}

\begin{proof}
    Recall that $\bar{\mu}_{\agentnum}$ is the excepted value of the first-order static out of $\agentnum$ i.i.d.\ variables drawing from distribution $\bar{\welfaredis}$.
    We also let $\bar{\mu}_{0} = 0$.
    Define $t = 1 - \welfaredis(0)$ and $B(\agentnum, k) = \binom{\agentnum}{k} t^{k}(1 - t)^{\agentnum - k}$.
    Using the binomial theorem as follows, the first-best benchmark $\utipfb(\agentnum)$ can be represented by $\bar{\mu}_{i}$ ($i \in \agents$) as follows.
    \begin{align*}
        \utipfb(\agentnum)
        ={}& \expect[\welfareprf \sim \welfaredis^{\otimes \agentnum}]{\plus{\welfaremax}} \\
        ={}& \int_{0}^{\infty} \temp \cdot \dd \welfaredis^{\agentnum}(\temp) \\
        ={}& \int_{0}^{\infty} \temp \cdot \dd (t\bar{\welfaredis}(\temp) + (1 - t))^{\agentnum} \\
        ={}& \int_{0}^{\infty} \temp \cdot \dd \sum_{i = 0}^{\agentnum} B(\agentnum, i) \bar{\welfaredis}^{i}(\temp) \\
        ={}& \sum_{i = 0}^{\agentnum} B(\agentnum, i) \cdot \int_{0}^{\infty} \temp \cdot \dd \bar{\welfaredis}^{i}(\temp) \\
        ={}& \sum_{i = 0}^{\agentnum} B(\agentnum, i) \cdot \bar{\mu}_{i}.
    \end{align*}

    When $\bar{\welfaredis}$ is the standard exponential distribution, the expected values of its first statistic order is $\bar{\mu}_{\agentnum} = \harmonic_{\agentnum}$ (See Equation (4.6.6) in \cite{H-93}).
    By using the identical equalities $\harmonic_{\agentnum} = \sum_{j = 1}^{\agentnum} \binom{\agentnum}{j} \frac{(-1)^{j}}{j}$ and $\frac{(-1)^{j}}{j} = \int_{0}^{1} -(-\temp)^{j - 1} \cdot \dd \temp$, we can derive that
    \begin{align*}
        \utipfb(\agentnum)
        ={}& \sum_{i = 0}^{\agentnum} B(\lownum, i) \harmonic_{i} \\
        ={}& \sum_{i = 0}^{\agentnum} \binom{\agentnum}{i}t^{i} (1 - t)^{\agentnum - i} \cdot \sum_{j = 1}^{i} \binom{i}{j} \frac{(-1)^{j}}{j} \\
        ={}& \sum_{i = 0}^{\agentnum} \sum_{j = 1}^{i} \binom{\agentnum}{i} \binom{i}{j} t^{i} (1 - t)^{\agentnum - i} \cdot \int_{0}^{1} -(-\temp)^{j - 1} \cdot \dd \temp \\
        ={}& \int_{0}^{1} -\sum_{i = 0}^{\agentnum} \sum_{j = 1}^{i} \binom{\agentnum}{i} \binom{i}{j}t^{i} (1 - t)^{\agentnum - i}  (-\temp)^{j - 1} \cdot \dd \temp.
    \end{align*}
    Notice that $\sum_{i = 0}^{\agentnum} \sum_{j = 0}^{i} \binom{\agentnum}{i} \binom{i}{j} x^{i}y^{j}z^{\agentnum - i - j} = (x + y + z)^{\agentnum}$.
    We further have
    \begin{align*}
        \utipfb(\agentnum)
        ={}& \int_{0}^{1} -\sum_{i = 0}^{\agentnum} \sum_{j = 1}^{i} \binom{\agentnum}{i} \binom{i}{j}t^{i} (1 - t)^{\agentnum - i}  (-\temp)^{j - 1} \cdot \dd \temp \\
        ={}& \int_{0}^{1} \frac{1}{\temp} \sum_{i = 0}^{\agentnum} \sum_{j = 1}^{i} \binom{\agentnum}{i} \binom{i}{j}t^{i} (-(1 - t)\temp)^{j} (1 - t)^{\agentnum - i - j}   \cdot \dd \temp \\
        ={}& \int_{0}^{1} \frac{1}{\temp} \parent{\sum_{i = 0}^{\agentnum} \sum_{j = 0}^{i} \binom{\agentnum}{i} \binom{i}{j}t^{i} (-(1 - t)\temp)^{j} (1 - t)^{\agentnum - i - j} - \sum_{i = 0}^{\agentnum} \binom{\agentnum}{i} t^{i}(1 - t)^{\agentnum - i}} \cdot \dd \temp \\
        =& \int_{0}^{1} \frac{(t - (1 - t)\temp + 1 - t)^{\agentnum} - 1}{\temp} \cdot \dd \temp \\
        =& \int_{0}^{1} \frac{(1 - (1 - t)\temp)^{\agentnum} - 1}{\temp} \cdot \dd \temp.
    \end{align*}
    By replacing $\temp$ with $1 - \temp^{\dag}$, since $\temp \in [0, 1]$, we finally derive that
    \begin{align*}
        \utipfb(\agentnum)
        =& \int_{0}^{1} \frac{(1 - (1 - t)\temp)^{\agentnum} - 1}{\temp} \cdot \dd \temp. \\
        =& \int_{0}^{1} \frac{(1 - (1 - t)(1 - \temp^{\dag}))^{\agentnum}}{1 - \temp^{\dag}} \cdot \dd (1 - \temp^{\dag}) \\
        =& \int_{0}^{1} \frac{1 - (t\temp^{\dag} + 1 - t)^{\agentnum}}{1 - \temp^{\dag}} \cdot \dd \temp^{\dag}.
    \end{align*}
    which completes the proof.
\end{proof}

\begin{lemma}  
\label{lem:harmonic bound}
    Define $\hat{\harmonic}_{0} = 0$ and
    \begin{equation*}
        \hat{\harmonic}_{\agentnum} = \int_{0}^{1} \frac{1 - (t \temp + 1 - t)^{\agentnum}}{1 - \temp} \dd \temp
    \end{equation*}
    for any positive number $\agentnum\in\naturals^{+}$.
    Then, for any real $t \in [0, 1]$ and $\lownum \leq \upnum$, we have
    \begin{equation*}
        \frac{\hat{\harmonic}_{\lownum}}{\hat{\harmonic}_{\upnum}} 
        \geq \frac{\ln(\frac{1}{2} t\lownum + 1)}{\ln(\frac{1}{2} t\upnum + 1)}.
    \end{equation*}
\end{lemma}

\begin{proof}
    By induction, it is sufficient to show that
    \begin{equation*}
        \frac{\hat{\harmonic}_{\agentnum + 1}}{\hat{\harmonic}_{\agentnum}}
        \leq \frac{\ln(\pert t(\agentnum + 1) + 1)}{\ln(\pert t\agentnum + 1)},
    \end{equation*}
    for any integer $\agentnum \in \mathbb{N}$, where $\pert = \frac{1}{2}$.
    Letting $\Delta_{\agentnum} = \hat{\harmonic}_{\agentnum + 1} - \hat{\harmonic}_{\agentnum}$, it is sufficient to show that
    \begin{equation*}
        \frac{\sum_{i = 0}^{\agentnum} \Delta_{i}}{\sum_{i = 0}^{\agentnum - 1} \Delta_{i}}
        \leq \frac{\ln(\pert t(\agentnum + 1) + 1)}{\ln(\pert t\agentnum + 1)},
    \end{equation*}
    which is equivalent to showing
    \begin{equation*}
        \frac{\Delta_{\agentnum}}{\sum_{i = 0}^{\agentnum - 1} \Delta_{i}}
        \leq \frac{\ln(\pert t(\agentnum + 1) + 1) - \ln(\pert t\agentnum + 1)}{\ln(\pert t\agentnum + 1)}
        = \frac{\ln \parent{ \frac{\pert t}{\pert t\agentnum + 1} + 1}}{\ln(\pert t\agentnum + 1)}.
    \end{equation*}

    We rewrite the above condition as
    \begin{equation*}
        \sum_{i = 0}^{\agentnum - 1} \frac{\Delta_{i}}{\Delta_{\agentnum}}
        \geq \frac{\ln(\pert t\agentnum + 1)}{\ln \parent{ \frac{\pert t}{\pert t\agentnum + 1} + 1}}
        = \sum_{i = 0}^{\agentnum - 1} \frac{\ln \parent{ \frac{\pert t}{\pert ti + 1} + 1}}{\ln \parent{ \frac{\pert t}{\pert t\agentnum + 1} + 1}}.
    \end{equation*}
    Therefore, by induction, it is sufficient to show that
    \begin{equation*}
        \frac{\Delta_{\agentnum - 1}}{\Delta_{\agentnum}}
        \geq \frac{\ln \parent{ \frac{\pert t}{\pert t(\agentnum - 1) + 1} + 1}}{\ln \parent{ \frac{\pert t}{\pert t\agentnum + 1} + 1}},
    \end{equation*}
    for any integer $\agentnum \in \mathbb{N}^{+}$.

    Define a real function $\xi \colon \temp \mapsto \temp \ln\parent{\frac{1}{\temp} + 1}$.
    We know $\xi(\cdot)$ is increasing since 
    \begin{equation*}
        \xi'(\temp) 
        = \ln\parent{\frac{1}{\temp} + 1} - \frac{1}{\temp + 1} 
        = \int_{\temp}^{\infty} \frac{\dd u}{u(u + 1)^{2}}
        > 0,
    \end{equation*}
    Hence, we further have $\xi \parent{\frac{\pert t \agentnum + 1}{\pert t}} \geq \xi \parent{\frac{\pert t (\agentnum - 1) + 1}{\pert t}}$.
    This implies that
    \begin{equation*}
        \frac{\pert t \agentnum + 1}{\pert t} \cdot \ln \parent{ \frac{\pert t}{\pert t \agentnum + 1} + 1}
        \geq \frac{\pert t (\agentnum - 1) + 1}{\pert t} \cdot \ln \parent{ \frac{\pert t}{\pert t(\agentnum - 1) + 1} + 1},
    \end{equation*}
    which is equivalent to the fact that
    \begin{equation*}
        \frac{\pert t \agentnum + 1}{\pert t (\agentnum - 1) + 1} 
        \geq \frac{\ln \parent{ \frac{\pert t}{\pert t(\agentnum - 1) + 1} + 1}}{\ln \parent{ \frac{\pert t}{\pert t\agentnum + 1} + 1}}.
    \end{equation*}
    Finally, according to \Cref{lem:diff Hhat} below, we finally derive that
    \begin{align*}
        \frac{\Delta_{\agentnum - 1}}{\Delta_{\agentnum}} 
        \geq \frac{\pert t \agentnum + 1}{\pert t (\agentnum - 1) + 1}
        \geq \frac{\ln \parent{ \frac{\pert t}{\pert t(\agentnum - 1) + 1} + 1}}{\ln \parent{ \frac{\pert t}{\pert t\agentnum + 1} + 1}}.
    \end{align*}
    which gives us the desired conclusion.
\end{proof}

\begin{lemma} \label{lem:diff Hhat}
    Define $\hat{\harmonic}_{0} = 0$ and
    \begin{equation*}
        \hat{\harmonic}_{\agentnum} = \int_{0}^{1} \frac{1 - (t \temp + 1 - t)^{\agentnum}}{1 - \temp} \cdot \dd \temp,
    \end{equation*}
    for any positive number $\agentnum\in\naturals^{+}$ and real $t \in [0,1]$.
    Then, we have
    \begin{equation*}
        \frac{\hat{\harmonic}_{\agentnum} - \hat{\harmonic}_{\agentnum - 1}}{\hat{\harmonic}_{\agentnum + 1} - \hat{\harmonic}_{\agentnum}} 
        \geq \frac{\frac{1}{2} t \agentnum + 1}{\frac{1}{2} t (\agentnum - 1) + 1}.
    \end{equation*}
\end{lemma}

\begin{proof}
    Let $\Delta_{\agentnum} = \hat{\harmonic}_{\agentnum + 1} - \hat{\harmonic}_{\agentnum}$ and $\delta = \frac{1}{2}$.
    It is sufficient to show
    \begin{equation*}
        \frac{\Delta_{\agentnum - 1}}{\Delta_{\agentnum}} 
        \geq \frac{\pert t \agentnum + 1}{\pert t (\agentnum - 1) + 1}.
    \end{equation*}
    First, by the definition of $\hat{\harmonic}_{\agentnum}$, we have
    \begin{align*}
        \Delta_{\agentnum} 
        ={}& \hat{\harmonic}_{\agentnum + 1} - \hat{\harmonic}_{\agentnum} \\
        ={}& \int_{0}^{1} \parent{\frac{1 - (t \temp + 1 - t)^{\agentnum + 1}}{1 - \temp} - \frac{1 - (t \temp + 1 - t)^{\agentnum}}{1 - \temp}} \cdot \dd \temp \\
        ={}& t \cdot \int_{0}^{1} (t \temp + 1 - t)^{\agentnum} \cdot \dd \temp \\
        ={}& \frac{1 - (1 - t)^{\agentnum + 1}}{\agentnum + 1}.
    \end{align*}
    To show the desired inequality, we construct the following difference:
     \begin{align*} 
         &\Delta_{\agentnum - 1} \cdot (\pert t (\agentnum - 1) + 1) - \Delta_{\agentnum} \cdot (\pert t \agentnum + 1) \\
         ={}& (1 - (1 - t)^{\agentnum}) \cdot \parent{\pert t + \frac{1 - \pert t}{\agentnum}} - (1 - (1 - t)^{\agentnum + 1}) \cdot \parent{\pert t + \frac{1 - \pert t}{\agentnum + 1}} \\
         ={}& \frac{1 - \pert t}{\agentnum (\agentnum + 1)} - (1 - t)^{\agentnum} \cdot \parent{\pert t + \frac{1 - \pert t}{\agentnum}} + (1 - t)^{\agentnum + 1} \cdot \parent{\pert t + \frac{1 - \pert t}{\agentnum + 1}} \\
         ={}& \frac{1 - \pert t}{\agentnum (\agentnum + 1 )} - (1 - t)^{\agentnum} \cdot \parent{\pert t^{2} + \frac{(nt + 1)(1 - \pert t)}{\agentnum(\agentnum + 1)}} \\
         ={}& \frac{1 - \pert t}{\agentnum (\agentnum + 1 )} - (1 - t)^{\agentnum} \cdot \frac{\pert t^{2} \agentnum^{2} + t \agentnum + 1 - \pert t}{\agentnum(\agentnum + 1)}.
     \end{align*}     
     This indicates that we only need to show the function
     \begin{align*}
         g_{t}(\temp) 
={}& (1 - \pert t) - (1 - t)^{\temp} \parent{\pert t^{2} \temp^{2} + t \temp + 1 - \pert t}
     \end{align*}
     is non-negative when $\temp \in \mathbb{N}$. 
     Observe $g_{t}(0) = 0$, it is sufficient to show $g_{t}(\cdot)$ is weakly increasing for any fixed $t \in [0, 1]$.
     
     Notice that the first derivation of the function $g_{t}(\cdot)$ can be calculated as
     \begin{align*}
         g'_{t}(\temp)
         = -(1 - t)^{\temp} \parent{\ln(1 - t) (\pert t^{2} \temp^{2} + t \temp + 1 - \pert t) + (2\pert t^{2} \temp + t)}.
     \end{align*}
     Clearly, we have $g'_{1}(\temp) = 2\pert\temp + 1 \geq 0$, indicating that $g_{1}(\cdot)$ is weakly increasing.
     Now, we assume that $t \in [0, 1)$.
     
     Consider the function $\frac{-g'_{t}(\temp)}{(1 - t)^{\temp}}$ and analysis the sign of this function for any $t \in [0, 1)$.
     On the one hand, the value of this function at the point $\temp = 0$ is
     \begin{align*}
         \left. \frac{-g'_{t}(\temp)}{(1 - t)^{\temp}} \right\vert_{\temp = 0}
         = \ln (1 - t)\parent{1 - \frac{1}{2} t} + t
         = -\frac{1}{2} \int_{1 - t}^{1} \left( \frac{1}{u} + \ln u - 1  \right) \cdot \dd u
         \leq 0.
     \end{align*}
     On the other hand, since $\ln(1 - t) < 0$ for $t \in [0, 1)$, the first derivation of this function can be calculated as
     \begin{align*}
         \DDX{\temp} \parent{\frac{-g'_{t}(\temp)}{(1 - t)^{\temp}}}
         = t \parent{\ln(1 - t)(2\pert t \temp + 1) + 2\pert t}
         = t \parent{\ln(1 - t) (t \temp + 1) + t}
        \leq t(\ln(1 - t) + t) \leq 0,
     \end{align*}
     which implies that the function $-g'_{t}(\temp)/(1 - t)^{\temp}$ is weakly decreasing.
     Combining the above two facts, we derive that for any $\temp \geq 0$, it holds $-g'_{t}(\temp)/(1 - t)^{\temp} \leq 0$, leading to $g'_{t}(\temp) \geq 0$.
     
     As a result, we finally obtain that $g_{t}(\cdot)$ is weakly increasing for $t \in [0, 1)$.
     This completes our proof.    
\end{proof}

\begin{restatable}{lemma}{lemFBRatio}
\label{lem:FB ratio}
    For $\lownum \leq \upnum$, the ratio between $\utipfb(\lownum)$ and $\utipfb(\upnum)$ can be bounded by $\frac{\ln \parent{\frac{1}{2}(1 - \welfaredis(0))\lownum + 1}}{\ln \parent{\frac{1}{2}(1 - \welfaredis(0))\upnum + 1}}$ with an $\smallO{\lownum, \upnum}{1}$ error, i.e.,
    \begin{equation*}
        \frac{\utipfb(\lownum)}{\utipfb(\upnum)} 
        \geq \frac{\ln \parent{\frac{1}{2}(1 - \welfaredis(0))\lownum + 1}}{\ln \parent{\frac{1}{2}(1 - \welfaredis(0))\upnum + 1}} - \smallO{\lownum, \upnum}{1}.
    \end{equation*}
\end{restatable}

\begin{proof}
    Let $\bar{\welfaredis}$ be the conditional distribution conditioning on the event that the random value following $\welfaredis$ is non-negative.
    Denote by $\bar{\mu}_{\agentnum}$ the expected value of the first-order statistic out of $\agentnum$ i.i.d.\ variables drawing from distribution $\bar{\welfaredis}$.
    Define $t = 1 - \welfaredis(0)$ and $B(\agentnum, k) = \binom{\agentnum}{k} t^{k}(1 - t)^{\agentnum - k}$.
    According to \Cref{lem:U_fb}, we know
    \begin{equation*}
        \utipfb(\agentnum) = \sum_{i = 0}^{\agentnum} \binom{\agentnum}{i} t^{i}(1 - t)^{\agentnum - i} \cdot \bar{\mu}_{i}.
    \end{equation*}

    For given $\lownum$ and $\upnum$ with $\lownum \leq \upnum$, we define a positive real $k = \sqrt{\lownum \upnum} t$, so that $\lownum t\leq k \leq \upnum t$.    
    Based on \Cref{lem:ratio of ELSO}, we can obtain a lower bound of $\frac{\utipfb(\lownum)}{\utipfb(\upnum)}$ by replacing $\bar{\mu}_{i}$ with $\harmonic_{i}$ as follows.
    \begin{align*}
        \frac{\utipfb(\lownum)}{\utipfb(\upnum)}
        \geq{} \frac{\sum_{i = 0}^{\floor{k}} B(\lownum, i) \bar{\mu}_{i} + \sum_{i = \floor{k}}^{\lownum} B(\lownum, i) \bar{\mu}_{k}}{\sum_{i = 0}^{\ceil{k}} B(\upnum, i) \bar{\mu}_{k} + \sum_{i = \ceil{k}}^{\upnum} B(\upnum, i) \bar{\mu}_{i}}
        \geq{} \frac{\sum_{i = 0}^{\floor{k}} B(\lownum, i) \harmonic_{i} + \sum_{i = \floor{k}}^{\lownum} B(\lownum, i) \harmonic_{k}}{\sum_{i = 0}^{\ceil{k}} B(\upnum, i) \harmonic_{k} + \sum_{i = \ceil{k}}^{\upnum} B(\upnum, i) \harmonic_{i}}.
    \end{align*}
    
    In addition, by using the Chernoff inequalities of $B(\cdot, \cdot)$ and \Cref{lem:ratio of ELSO}, we further obtain that
    \begin{align*}
        \frac{\utipfb(\lownum)}{\utipfb(\upnum)}
        \geq{}& \frac{\sum_{i = 0}^{\floor{k}} B(\lownum, i) \harmonic_{i} + \sum_{i = \floor{k}}^{\lownum} B(\lownum, i) \harmonic_{k}}{\sum_{i = 0}^{\ceil{k}} B(\upnum, i) \harmonic_{k} + \sum_{i = \ceil{k}}^{\upnum} B(\upnum, i) \harmonic_{i}} \\
        \geq{}& \frac{\sum_{i = 0}^{\lownum} B(\lownum, i) \harmonic_{i} - (\harmonic_{\lownum} - \harmonic_{k})\sum_{i = \floor{k}}^{\lownum} B(\lownum, i)}{\sum_{i = 0}^{\upnum} B(\upnum, i) \harmonic_{i} + \harmonic_{k} \sum_{i = 0}^{\ceil{k}} B(\upnum, i)}  \\
        \geq{}& \frac{\sum_{i = 0}^{\lownum} B(\lownum, i) \harmonic_{i} - (\harmonic_{\lownum} - \harmonic_{k})\ee^{-2\lownum(1 - t - (1 - \floor{k}/\agentnum))^{2})}}{\sum_{i = 0}^{\upnum} B(\upnum, i) \harmonic_{i} + \harmonic_{k} \ee^{-2\upnum(t - \ceil{k}/\agentnum)^{2})}}  \\
        \geq{}& \frac{\sum_{i = 0}^{\lownum} B(\lownum, i) \harmonic_{i}}{\sum_{i = 0}^{\upnum} B(\upnum, i) \harmonic_{i}} - \bigO{\lownum, \upnum}{\ln \upnum \cdot \ee^{-t\sqrt{\lownum \upnum}}}.
    \end{align*}

    Recall that $\hat{\harmonic}_{0} = 0$ and
    \begin{equation*}
        \hat{\harmonic}_{\agentnum} = \int_{0}^{1} \frac{1 - (t \temp + 1 - t)^{\agentnum}}{1 - \temp} \cdot \dd \temp,
    \end{equation*}
    for any positive number $\agentnum\in\naturals^{+}$ and real $t \in [0,1]$.
    Combining \cref{lem:U_fb,lem:harmonic bound}, it directly follows that
    \begin{align*}
        \frac{\utipfb(\lownum)}{\utipfb(\upnum)}
        \geq \frac{\sum_{i = 0}^{\lownum} B(\lownum, i) \harmonic_{i}}{\sum_{i = 0}^{\upnum} B(\upnum, i) \harmonic_{i}} - \bigO{\lownum, \upnum}{\ln \upnum \cdot \ee^{-t\sqrt{\lownum \upnum}}}
        =\frac{\hat{\harmonic}_{\lownum}}{\hat{\harmonic}_{\upnum}} - \smallO{\lownum, \upnum}{1}
        \geq \frac{\ln \parent{\frac{1}{2}t\lownum + 1}}{\ln \parent{\frac{1}{2}t\upnum + 1}} - \smallO{\lownum, \upnum}{1},
    \end{align*}
    which completes the proof.
\end{proof}

Now, we are ready to prove \Cref{thm:robust MHR bounds}.

\begin{proof}[Proof of \Cref{thm:robust MHR bounds}.]
    Suppose that for each agent $i \in \agents$, the \contribution $\welfare_{i}$ satisfies an MHR distribution $\welfaredis$.
    According to \Cref{lem:FB ratio}, we have
    \begin{align*}
        \frac{\utipfb(\lownum)}{\utipfb(\upnum)} 
        \geq \frac{\ln \parent{\frac{1}{2} (1 - \welfaredis(0)) \lownum + 1}}{\ln \parent{\frac{1}{2} (1 - \welfaredis(0)) \upnum + 1}} - \smallO{\lownum, \upnum}{1}.
    \end{align*}   
    Consider the design of setting $\contract^{\dag}$ as the optimal contract for the number of agents being $\lownum$, i.e., 
    \begin{align*}
        \contract^{\dag} = \argmax_{\contract\in[0, 1]}\rewardfun(1 - \contract)(1 - \valdis_{\contract}^{\lownum}(\reserve_{\contract})).
    \end{align*}
    Exploiting the monotonicity of $\utip(\contract^{\dag}, \cdot)$ and $\utipfb(\cdot)$, $\ratiofb(\lownum,\upnum)$ can be upper bounded by
    \begin{align*}
        \ratiofb(\lownum,\upnum)
        \leq \max_{\agentnum\in[\lownum:\upnum]} \frac{\utipfb(n)}{\utip(\contract^{\dag}, \agentnum)}
        \leq \frac{\utipfb(\upnum)}{\utip\parent{\contract^{\dag}, \lownum}}
        \leq \frac{\utipfb(\upnum)}{\utipfb(\lownum)} \cdot \frac{\utipfb(\lownum)}{\utip\parent{\contract^{\dag}, \lownum}}
        \leq \frac{\ln \parent{\frac{1}{2} (1 - \welfaredis(0)) \upnum + 1}}{\ln \parent{\frac{1}{2} (1 - \welfaredis(0)) \lownum + 1}} \cdot \ee^{2} + \smallO{\lownum, \upnum}{1},
    \end{align*}
    where the last inequality holds by the definition of $\contract^{\dag}$ and \Cref{thm:mhr ratio bounds for n in naturals}.    
    Moreover, by \Cref{thm:r_fb asymp lb}, we can also obtain that
    \begin{equation*}
        \ratiosb(\lownum, \upnum)
        \leq \ratiofb(\lownum, \upnum) 
        =  \frac{\ln \parent{\frac{1}{2} (1 - \welfaredis(0)) \upnum + 1}}{\ln \parent{\frac{1}{2} (1 - \welfaredis(0)) \lownum + 1}} \cdot \ee^{2} + \bigO{\lownum,\upnum}{\frac{\log \log \lownum \cdot \log \upnum}{(\log \lownum)^{2}}},
    \end{equation*}
    which completes the proof.
\end{proof}
 
\subsection{Proof of Theorem~\ref{thm:robust negative}}
\label{apx:robust negative proof}

In this section, we will prove \Cref{thm:robust negative}, which provide negative results for the case that the principal has little information about the market size.

\thmRobustNegative*
\begin{proof}
    Suppose that $\welfare_{i}$ follows the conditional exponential distribution conditioning on the event that the random variable is no more than $\rewardfun$.
    That is
    \begin{equation*}
        \welfaredis_{\rewardfun}(\temp) = \left\{
            \begin{aligned}
                &0,&&\temp \in (-\infty, 0); \\
                &\frac{1 - \ee^{-\temp}}{1 - \ee^{-\rewardfun}},&&\temp \in [0, \rewardfun]; \\
                &1,&&\temp \in (\rewardfun, +\infty).
            \end{aligned}
        \right.
    \end{equation*}
    We know that $\welfaredis_{\rewardfun}$ is MHR for each $\rewardfun > 0$.
    The PDF of $\welfaredis_{\rewardfun}$ is
    \begin{equation*}
        \welfaredens_{\rewardfun}(\temp) = \frac{\ee^{-\temp}}{1 - \ee^{-\rewardfun}},
    \end{equation*}
    and the virtual price is
    \begin{equation*}
        \virval(\temp) = \temp - \frac{1 - \welfaredis_{\rewardfun}(\temp)}{\welfaredens_{\rewardfun}(\temp)} = \temp - \frac{\ee^{-\temp} - \ee^{-\rewardfun}}{\ee^{-\temp}} = \temp - 1 + \ee^{\temp - \rewardfun}.
    \end{equation*}
    
    We will prove a stronger statement: for any positive real $\varepsilon > 0$ and contract $\contract \in [0, 1]$, there exists large enough reward $\rewardfun$ and $\upnum$, such that it always holds that $\ratiofb(\contract) \leq \varepsilon$.
    Formally, we let $\rewardfun$ and $\upnum$ satisfy the following relations
    \begin{align*}
        (\rewardfun + 1)\ee^{-\rewardfun} \leq \varepsilon,\quad
        \frac{\ln \rewardfun}{\rewardfun (\ee^{-1} - \varepsilon)} \leq \varepsilon,\quad
        \frac{\ln \rewardfun}{\rewardfun - 1 - \varepsilon} \leq \varepsilon,\quad
        (\rewardfun - 1) \ee^{-\upnum} \leq \varepsilon.
    \end{align*}
    
    Notice that the monopoly reserve price $\reserve_{\contract}$ is positive.
    By definition, given any $\contract \in [0, 1]$, we have
    \begin{align*}
        \utip(\contract, \agentnum) 
        ={}& \rewardfun (1 - \contract) \parent{1 - \welfaredis_{\rewardfun}^{\agentnum}(\rewardfun (1 - \contract) + \reserve_{\contract})} \\
        \leq{}& \rewardfun (1 - \contract) \parent{1 - \welfaredis_{\rewardfun}^{\agentnum}(\rewardfun (1 - \contract))} \\
={}& \rewardfun (1 - \contract) \left(1 - \left( \frac{1 - \ee^{-\rewardfun (1 - \contract)}}{1 - \ee^{-\rewardfun}} \right)^{\agentnum} \right) \\
        \leq{}& \rewardfun (1 - \contract) \left(1 - \left(1 - \ee^{-\rewardfun (1 - \contract)} \right)^{\agentnum} \right),
    \end{align*}
    and
    \begin{align*}
        \utipsb(1)
        = \max_{\temp \in [0, \rewardfun]} \temp \left(1 - \frac{1 - \ee^{-\temp}}{1 - \ee^{-\rewardfun}} \right)
        \geq \frac{\ee^{-1} - \ee^{-\rewardfun}}{1 - \ee^{-\rewardfun}}
        \geq \ee^{-1} - \varepsilon, 
    \end{align*}
and
    \begin{align*}
        \utipsb(\upnum)
        ={}& \int_{0}^{\rewardfun} \max \{\virval(\temp), 0\} \cdot \dd \welfaredis_{\rewardfun}^{\upnum}(\temp) \\
        \geq{}& \int_{1}^{\rewardfun} (\temp - 1) \cdot \dd \welfaredis_{\rewardfun}^{\upnum}(\temp) \\
        ={}& \rewardfun - 1 - \int_{0}^{\rewardfun - 1} \welfaredis_{\rewardfun}^{\upnum}(\temp + 1) \cdot \dd \temp \\
        ={}& \rewardfun - 1 - \int_{0}^{\rewardfun - 1} \left(1 - \frac{1 - \ee^{-\temp - 1}}{1 - \ee^{-\rewardfun}} \right)^{\upnum} \cdot \dd \temp \\
        \geq{}& \rewardfun - 1 - (\rewardfun - 1) \left(1 - \frac{1 - \ee^{-1}}{1 - \ee^{-\rewardfun}} \right)^{\upnum} \\
        \geq{}& \rewardfun - 1 - (\rewardfun - 1) \ee^{-\upnum} \\
        \geq{}& \rewardfun - 1 - \varepsilon.
    \end{align*}
    
    Next, we bound the ratio when $\agentnum = 1$ and $\agentnum = b$, respectively.
    \begin{align*}
        \frac{\utipsb(1)}{\utip(\contract, 1)}
        \geq \frac{\ee^{-1} - \varepsilon}{\rewardfun (1 - \contract) \left(1 - (1 - \ee^{-\rewardfun (1 - \contract)}) \right)}
        = \frac{\ee^{-1} - \varepsilon}{\rewardfun (1 - \contract) \ee^{-\rewardfun (1 - \contract)}},
    \end{align*}
    and
    \begin{equation*}
        \frac{\utipsb(\upnum)}{\utip(\contract, \upnum)}
        \geq \frac{(\rewardfun - \varepsilon)(1 - \varepsilon)}{\rewardfun (1 - \contract) \left(1 - (1 - \ee^{-\rewardfun (1 - \contract)})^{\upnum} \right)}
        \geq \frac{\rewardfun - 1 - \varepsilon}{\rewardfun (1 - \contract)}.
    \end{equation*}
    
    Now, we are ready to bound the upper bound of $\ratiosb$.
    Define $\contract^{*} = 1 - (\ln \rewardfun)/\rewardfun > 0$.
    Then, we know that $\rewardfun (1 - \contract^{*}) = \ln \rewardfun$.
    By dividing the $[0, 1]$ into two subintervals, we have
    \begin{align*}
        \ratiosb(1,\upnum)
        ={}& \min_{\contract \in [0, 1]} \max_{\agentnum\in [\upnum]} \frac{\utipsb(\agentnum)}{\utip(\contract, \agentnum)} \\
        \geq{}& \min_{\contract\in[0, 1]} \max \left\{ \frac{\utipsb(1)}{\utip(\contract, 1)}, \frac{\utipsb(\upnum)}{\utip(\contract, \upnum)} \right\} \\
        ={}& \min \set{\min_{\contract \in [0, \contract^{*}]} \max \left\{ \frac{\utipsb(1)}{\utip(\contract, 1)}, \frac{\utipsb(\upnum)}{\utip(\contract, \upnum)} \right\}, \min_{\contract \in [\contract^{*}, 1]} \max \left\{ \frac{\utipsb(1)}{\utip(\contract, 1)}, \frac{\utipsb(\upnum)}{\utip(\contract, \upnum)} \right\}} \\
        \geq{}& \min \left\{ \min_{\contract \in [0, \contract^{*}]} \frac{\utipsb(1)}{\utip(\contract, 1)}, \min_{\contract \in [\contract^{*}, 1]} \frac{\utipsb(\upnum)}{\utip(\contract, \upnum)} \right\}.
    \end{align*}
    On the one hand,
    \begin{align*}
        \min_{\contract \in [0, \contract^{*}]} \frac{\utipsb(1)}{\utip(\contract, 1)}
        \geq \min_{\contract \in [0, \contract^{*}]} \frac{\ee^{-1} - \varepsilon}{\rewardfun (1 - \contract) \ee^{-\rewardfun (1 - \contract)}}
        \geq \frac{\rewardfun(\ee^{-1} - \varepsilon)}{\ln \rewardfun}
        \geq \frac{1}{\varepsilon}.
    \end{align*}
    On the other hand,
    \begin{align*}
        \min_{\contract \in [\contract^{*}, 1]} \frac{\utipsb(\upnum)}{\utip(\contract, \upnum)}
        \geq \min_{\contract \in [\contract^{*}, 1]} \frac{\rewardfun - 1 - \varepsilon}{\rewardfun (1 - \contract)} 
        \geq \frac{\rewardfun - 1 - \varepsilon}{\rewardfun (1 - \contract^{*})}
        \geq \frac{\rewardfun - 1 - \varepsilon}{\ln \rewardfun}
        \geq \frac{1}{\varepsilon}.
    \end{align*}
    As a result, we finally derive that $\ratiofb(1,\upnum) \geq \ratiosb(1,\upnum) \geq \frac{1}{\varepsilon}$.
\end{proof}

\end{document}